%% file: Weiproparxiv2adopted.tex
\documentclass[aps,11pt]{revtex4}
\usepackage{algorithm}
\usepackage{makeidx}
\usepackage{multirow}
\usepackage{array}
\usepackage{amsmath,amssymb,amsthm}
\usepackage{mathrsfs}
\usepackage{amsfonts}
\usepackage{setspace}
\usepackage[]{xcolor}
\usepackage{graphicx,color,scalefnt}
\usepackage{epstopdf,mathtools,amsfonts}
\usepackage{ulem}
\usepackage{hyperref}
\usepackage{subfigure,epstopdf}
\newcommand{\be}{\begin{equation}}
	\newcommand{\ee}{\end{equation}}
\newcommand{\bea}{\begin{eqnarray}}
	\newcommand{\eea}{\end{eqnarray}}
\newtheorem{theorem}{Theorem}[section]

\newcommand{\ud}{\mathrm{d}}

\begin{document}
	%\singlespacing
	\title{\Large \bf Maximum $\log_q$ Likelihood Estimation for Parameters of Weibull Distribution and Properties: Monte Carlo Simulation}
	
	\smallskip 
		\author{Mehmet Niyazi \c{C}ankaya$^{a1,a2}$}
	\author{\bf Roberto Vila$^{a3}$}

	\bigskip 
		\affiliation{$^{a1}$ Faculty of Applied Sciences, Department of International Trading and Finance, U\c{s}ak University,}
	\affiliation{$^{a2}$  Faculty of Art and Sciences, Department of Statistics, U\c{s}ak University, U\c{s}ak, Turkey}
	\affiliation{$^{a3}$ Departamento de Estatística, Universidade de Brasília, Brazil}
	
	\begin{abstract}
		\begin{center}
			\text{\normalsize{Abstract}}
		\end{center}
		\vspace{-0,3cm}
The maximum ${\log}_q$ likelihood estimation method is a generalization of the known maximum $\log$
likelihood method to overcome the problem for modeling non-identical observations (inliers and outliers). The parameter 
$q$ is a tuning constant to manage the modeling capability.  Weibull is a flexible and popular distribution for problems in engineering. In this study, this method is used to estimate the parameters of
Weibull distribution when non-identical observations exist.  Since the main idea is based on modeling
capability of objective function $\rho(x;\boldsymbol{\theta})=\log_q\big[f(x;\boldsymbol{\theta})\big]$, we observe that the
finiteness of score functions cannot play a role in the robust estimation for inliers.  The properties of Weibull distribution are examined. In the numerical experiment, the
parameters of Weibull distribution are estimated by $\log_q$ and its special form, $\log$, likelihood methods if the
different designs of contamination into underlying Weibull distribution are
applied. The optimization is performed via genetic
algorithm.  The modeling competence of
$\rho(x;\boldsymbol{\theta})$ and insensitiveness
to non-identical observations are observed by Monte Carlo simulation. The value of
$q$  can be chosen by use of the mean squared error in simulation
and the $p$-value of Kolmogorov-Smirnov test statistic used for evaluation of
fitting competence. Thus, we can overcome
the problem about determining of the value of $q$ for real data sets.
		
		
	\end{abstract}
	\nopagebreak 
\maketitle

\section{Introduction}
\label{sec1}
After the study of a real-world phenomenon or the realization of an experiment,
it may be desirable to model the experimental data by means of a proposed
parametric model $f(x;\boldsymbol{\theta})$. In other words, the experimental
data set is assumed to be a member of a parametric model. However, it cannot be a
realistic assumption for the real world which will be modeled only with
certain values of parameters  in a  model. The observations can be mixed with a 
different parameter of the same distribution or a different distribution. That
is, a contamination exists into majority of the distribution, which leads to have non-identical observations. The type of
contamination is defined as deviant observations. In other words, one or more
observations are made to behave differently from what it is present when creating
a deviant observation. Deviant observations can be divided into inlier and
outlier deviations. The data set may have both inward/inlier and outward/outlier
deviations at the same time. Inward deviations are generated by short-tailed
distributions and outward deviations by thick-tailed distributions. For
inward deviations, it can also be realized by generating random numbers from
uniform distribution in the closed range $[a,b]$. In fact, deviant observations
in the data set indicate that an assumption trusting on identically distributed
random variables will not be realistic to model a phenomenia (\cite{tikuinlier,LehmannCas98}). If the
assumption showing that the data set includes identically distributed random
observations is violated, robust estimation method for the parameters of model
$f(x;\boldsymbol{\theta})$ have been applied by use of the different objective
functions. Robust methods trust the used objective function (\cite{Hampeletal86}). Deformed
algebras such as Tsallis and Kaniakadis statistics are important to derive a
neighborhood of a parametric model in order to overcome the problem which will
occur when the assumption for ideniticality is violated (\cite{Wadatwopara,Bercher10,Bercher12a}). 

The origin of robust estimation method was started by biologist at 18$^{th}$ century
(\cite{Hampeletal86}). The main working principle is based on the estimating equations
(EEs) (\cite{God60,GodTh78}). EEs can be derived by use of maximum composite likelihood
estimation method. Tsallis $q$-entropy is creator of
deformed logarithm ($\log_q$). In the direction of estimation method, 
$\log_q$ from Tsallis $q$-entropy has been studied recently. For this aim,
maximum  $\log_q$ likelihood estimation (MLqE) method are studied for outward
observations. MLqE which is a generalization of the maximum
likelihood (MLE) method \cite{LehmannCas98} is used to obtain robust and also efficient estimators (\cite{Tsallisbook09,FerrariYang10}).
Different deformed logarithms can be obtained from entropy functions (\cite{Wadatwopara}). However, the deformed or the generalized logarithms should map the probability density function $f(x_i;\boldsymbol{\theta})$ as one-to-one and overlay (\cite{Lindsay94,CanKor18}). 
Kaniakadis' deformed logarithm (${\log}_{\kappa{}}$) can have same property with
${\log}_q$ due to fact that $(\alpha,\beta)$-difference operator in fractional calculus (FC) is used
to generate entropies. The other genaralized entropies can have same role
in the estimation procedure (\cite{Wadatwopara,Jansent}). Note that the performance of efficiency 
can be managed by using the generalized entropies and generalized logarithms from
FC. MLqE is simple and its computational implementation is not heavy. No extra
condition is requried to apply MLqE for estimation when it is compared with
divergences (\cite{alphabetadivergences}).  The parameters of Gamma distribution are estimated by MLqE
method (\cite{MLqEGamma}). As an another approach for the robust estimation,
divergences have been used. Minimization of
entropies and divergences are equivalent to maximization of the generalized
maximum likelihood estimation method (\cite{Lindsay94,PardoSD}).

FC has been started to play an important role in the
estimations of parameters of a probability density (p.d.) function
$f(x;\boldsymbol{\theta}).$  Tsallis
distributions ($q$-distributions) have encountered a large success because of
their remarkable agreement with experimental data. The parameter $q$ behaves as a microscope and generates neigborhoods of
p.d. function $f(x;\boldsymbol{\theta})$. Thus, the different behaviour of p.d.
function $f(x;\boldsymbol{\theta})$ can be explored by use of ${\log}_q$ (see \cite{Bercher10,Bercher12a} and references
therein). After we use $\log_q$ from Tsallis $q$-entropy to estimate robustly the parameters $\boldsymbol{\theta}$ of underlying distribution $f(x;\boldsymbol{\theta})$, that is, we will have
estimators which are not affected by inliers and outliers in a data set, the
optimization of $\sum_{i=1}^n\log_q\big[f(x_i;\boldsymbol{\theta})\big]$ according to the
parameters $\boldsymbol{\theta}$ is an another challeging problem for the estimation process
if we have a nonlinear function. The optimization is performed by means of genetic algorithm (GA)
which is not attached to local points. Thus, we will have estimators.
GA mimicing evolutionary
biology as a stochasticity is a derivative-independent method. We use hybrid
method in `ga' module in MATLAB2013a to decrease the computational error as well.
In hybrid method, GA and derivative-based methods work simultenaously. Further,
GA has an effective attitude to reach the global point in the fuction
$\rho{}(x;\boldsymbol{\theta})=\Lambda[f(x;\boldsymbol{\theta})]$ (\cite{introga}).

The aim is to estimate robustly shape and scale parameters of Weibull
distribution when the different types of contamination into artifical data set
are added.   The properties of Weibull distribution are examined extensively. 
The roles of $\Lambda{}$ from the deformed algebras and DPD, etc. as an objective
function $\rho{}$, are observed easily by using the illustrative representations (see Fig. \ref{functionscompare}),
which is a key to analyze their roles on p.d. function $f(x;\boldsymbol{\theta})$. 
Outliers and inliers are added into the artificial data sets at the simulation
and the different types of contaminations to real data sets are applied to
observe the robustness and the modeling capability of 
$\rho(x;\boldsymbol{\theta})=\log_q\big[f(x;\boldsymbol{\theta})\big]$.
The value of parameter $q$ is determined according to the $p$-value of KS test
statistic of estimates of parameters for Weibull distribution.
Fisher information based on ${\log}_q$  is used to evaluate the
variance-covariance matrix of estimators derived by MLqE.

The organization of study is as follows. Section \ref{Section2} includes Weibull distribution and its properties which indicate that Weibull can be used for the modeling fruitfully. We provide the essential tools to pass the modeling sketch of estimation procedure such as convexity, (concavity), entropies, etc \cite{Haberman} in the M-estimation (\cite{Hub81}). We also provide main tools to get the elements of Fisher information matrix. Section \ref{section3inference} introduces estimation methods \cite{varincomposite} and estimating equations (\cite{God60,GodTh78}). Section \ref{section4rob} introduces the tools and robustness to outliers and also we propose a tool based on score function for robustness to inliers. Section \ref{optimizationnumerical} provides the tool used for optimization and numerical experiments. The last section \ref{conclusionslab} is given for conclusions.

\section{Weibull Distribution}\label{Section2}
Weibull distribution is chosen because it has many applications in the field of the applied science. Further, it has many features such as existence of cumulative distribution (c.d.) function, moments and entropies, etc (\cite{Malikcondition,Cankaya2018}). We say that a non-negative random variable $X$ has a Weibull distribution with vector parameter $\boldsymbol{\theta}=(\alpha,\beta)$, denoted by $X\sim \text{Weibull}(\boldsymbol{\theta})$, if its probability density function is given by 
\begin{align}\label{def-Weibull}
	f(x;\boldsymbol{\theta})
	=
	{\alpha\over \beta}\, 
	\bigg({x\over\beta}\bigg)^{\alpha-1}
	\exp\Bigg[-\bigg({x\over\beta}\bigg)^{\alpha}\Bigg],
	\quad x\geqslant 0; \ \alpha,\beta>0,
\end{align}
where $\alpha$ is the shape parameter and $\beta$ is the scale parameter (\cite{Weibullref}). 

\subsection{Properties}\label{convexconcavonemode}
\label{Properties}
If $X\sim \text{Weibull}(\boldsymbol{\theta})$ then the following properties are satisfied (see \cite{reviewwei} for review of modified  Weibull distributions).

\begin{itemize}	
	\item[1)] {\bf Asymptotic behavior of $f(x;\boldsymbol{\theta})$.}
	The behavior of $f(x;\boldsymbol{\theta})$ with $x\to 0$ or $x\to \infty$ is as follows:
	\begin{align*}
		\lim_{x\to 0} f(x;\boldsymbol{\theta})
		=
		\begin{cases}
			\infty & \text{for} \quad 0<\alpha<1,
			\\
			\dfrac{1}{\beta} & \text{for} \quad \alpha=1,
			\\
			0 & \text{for} \quad \alpha>1,
		\end{cases}
	\end{align*}
	\begin{align*}
		\lim_{x\to \infty} f(x;\boldsymbol{\theta})
		=
		0 \quad \forall \alpha>0.
	\end{align*}

	\item[2)] {\bf Monotonicity, unimodality, concavity and convexity of $f(x;\boldsymbol{\theta})$.}
	The point $x$ is a mode of the Weibull density, if
	and only if it is the solution of the following equation
	\begin{align*}
		{{\rm d} f(x;\boldsymbol{\theta})\over {\rm d} x}=
		\displaystyle {f(x;\boldsymbol{\theta})\over x}\,
		\Bigg[(\alpha-1)-\alpha\bigg({x\over\beta}\bigg)^\alpha\Bigg]=0.
	\end{align*}
	Solving this equation, we get the following critical point
	\begin{align*}
		x_0=
		\begin{cases} \displaystyle
			\beta \bigg({\alpha-1\over\alpha}\bigg)^{1/\alpha}
			& \text{for} \ \alpha>1,
			\\[0,4cm]
			0 & \text{for} \ \alpha\leqslant 1.
		\end{cases}
	\end{align*}
	A simple calculation shows that
	\begin{align*}
		{{\rm d}^2 f(x;\boldsymbol{\theta})\over {\rm d} x^2}=
		\displaystyle {f(x;\boldsymbol{\theta})\over x^2}\,
		\Bigg[\alpha^2\bigg({x\over\beta}\bigg)^{2\alpha}
		-3\alpha(\alpha-1)\bigg({x\over\beta}\bigg)^{\alpha}
		+(\alpha-1)(\alpha-2)
		\Bigg].
	\end{align*}
	Note that
	\begin{align*}
		{{\rm d}^2 f(x_0;\boldsymbol{\theta})\over {\rm d} x^2}
		=
		-\alpha(\alpha-1)<0, \quad \alpha>1,
	\end{align*}
	and
	\begin{align*}
		{{\rm d}^2 f(x;\boldsymbol{\theta})\over {\rm d} x^2}=0
		\quad \iff \quad 
		x=x_{\pm}=\beta\Bigg[\dfrac{3(\alpha-1)\pm\sqrt{(\alpha-1)(5\alpha-1)}}{2\alpha}\Bigg]^{1/\alpha}.
	\end{align*}
	Therefore, the following properties follow immediately:
	
	For $\alpha> 1$,
	\begin{itemize}
		\item[$\bullet$] $f(x;\boldsymbol{\theta})$ increases as $x\to x_0$ and decreases thereafter.
		\item[$\bullet$] The point $x_0$ is the unique mode for the Weibull density.
		\item[$\bullet$] For $\alpha\neq 2$, the inflection points $x_{\pm}$ satisfy inequality $x_{-}<x_+$. Furthermore, $f(x_0;\boldsymbol{\theta})$ is convex on $(0,x_{-})\cup(x_+,\infty)$ and is concave on $(x_{-},x_+)$.
		\item[$\bullet$] For $\alpha= 2$, $x_{-}=0$ and $x_+=6\beta/4$. Then	$f(x_0;\boldsymbol{\theta})$ is concave on $(0,x_{+})$ and is convex on the interval $(x_+,\infty)$.
	\end{itemize}
	
	For $\alpha\leqslant 1$,
	\begin{itemize}
		\item[$\bullet$] $f(x_0;\boldsymbol{\theta})$ decreases monotonically and is convex.
		\item[$\bullet$] The mode is non-existent.
	\end{itemize}

	\item[3)] {\bf Reliability.} If $F(t;\boldsymbol{\theta})$  denotes the cumulative distribution function of $X$, then the reliability function is written as
	\begin{align*}
		\displaystyle
		R(t;\boldsymbol{\theta})
		=
		1-F(t;\boldsymbol{\theta})
		=
		\exp\Bigg[-\bigg({t\over\beta}\bigg)^\alpha\Bigg].
	\end{align*}
	\item[4)] {\bf Hazard rate.} The hazard rate is given by
	\begin{align*}
		H(t;\boldsymbol{\theta})
		=
		\dfrac{f(t;\boldsymbol{\theta})}{R(t;\boldsymbol{\theta})}
		= \displaystyle
		{\alpha\over \beta }\,  \bigg({t\over\beta}\bigg)^{\alpha-1}.
	\end{align*}
	The function $H(t;\boldsymbol{\theta})$ is increasing when $\alpha> 1$, decreasing when $\alpha<1$ and constant when $\alpha = 1$. 
	\item[5)] {\bf Truncated moment.} Integration by parts gives
	\begin{align*}
		\displaystyle 
		\mathbb{E}\big[{1}_{\{X\geqslant t\}} X^s\big]
		=
		t^s
		\exp\Bigg[-\bigg({t\over\beta}\bigg)^\alpha\Bigg]
		+ 
		{s\beta^s\over\alpha}\,
		\Gamma\bigg({s\over\alpha},{t^\alpha\over\beta^\alpha}\bigg),
		\quad s\geqslant 0,
	\end{align*}	
	where $\Gamma(s,x)$, $s>0$, is the upper incomplete gamma function, and where we  adopt the notation $\Gamma(0,x)=0$.
	%		Similarly we find that
	%		\begin{align}\label{int-formula-1}
	%		\mathbb{E}\big[{1}_{\{X\leqslant t\}} X^s\big]
	%		=
	%		\begin{cases}
	%		\displaystyle
	%		1-\exp\Bigg[-\bigg({t\over\beta}\bigg)^\alpha\Bigg] & \text{for} \quad s= 0,
	%		\\[0,5cm]
	%		\displaystyle
	%		-t^s
	%		\exp\Bigg[-\bigg({t\over\beta}\bigg)^\alpha\Bigg]
	%		+ 
	%		{s\beta^s\over\alpha}\,
	%		\gamma\bigg({s\over\alpha},{t^\alpha\over\beta^\alpha}\bigg)
	%		& \text{for} \quad s> 0, 
	%		\end{cases}
	%		\end{align}
	%		where $\gamma(s,x)$, $s>0$, is the lower incomplete gamma function.
	
	\item[6)]  {\bf Moment of the residual life.}	For each $t\geqslant 0$, we have
	\begin{align*}
		\displaystyle 
		\phi_n(t;\boldsymbol{\theta})
		&=
		\mathbb{E}\big[(X-t)^n\vert X\geqslant t\big]
		\\[0,3cm]	
		&=
		\sum_{k=1}^{n} \binom{n}{k} (-t)^{n-k}
		\Bigg\{
		t^k+\beta^k \Gamma\bigg(1+\dfrac{k}{\alpha}, \dfrac{t^\alpha}{\beta^\alpha}\bigg) \exp\Bigg[\bigg({t\over\beta}\bigg)^\alpha\Bigg]
		\Bigg\}
		+(-t)^n.
	\end{align*}
	
	\vspace*{0,3cm}\hspace*{0,5cm}
	The proof of this identity is immediate since
	\[
	\phi_n(t;\boldsymbol{\theta})
	=
	{1\over R(t;\boldsymbol{\theta})}\,
	\sum_{k=1}^{n} \binom{n}{k} (-t)^{s-k}
	\mathbb{E}\big[{1}_{\{X\geqslant t\}} X^n\big]
	+ (-t)^n,
	\]
	where $R(t;\boldsymbol{\theta})$ and
	$\mathbb{E}\big[{1}_{\{X\geqslant t\}} X^n\big]$
	are given in Items 3) and 5), respectively.
	
	\item[7)] {\bf First main tool.} By using the formula of Item (6) in \cite{Cankaya2018}, we have
	\begin{align*}
		&\displaystyle
		\mathbb{E}\big[X^s f^{r-1}(X;\boldsymbol{\theta})\big]
		=
		\dfrac{\alpha^{r-1} \beta^{s-r+1}}{r^{r+(s-r+1)/\alpha}} \,
		\beta^s \Gamma\bigg(r+{s-r+1\over\alpha}\bigg), \quad s-r+1>-\alpha r.
	\end{align*}
	
	\begin{itemize}
		\item[$\bullet$]
		{\bf Real moments.} By taking $r=1$ in the first main tool, we reach
		\begin{align*}
			\displaystyle 
			\mathbb{E}(X^s)
			=
			\beta^s \Gamma\bigg(1+{s\over\alpha}\bigg), \quad s>-\alpha,
		\end{align*}
		where $\Gamma(z)$ is the complete gamma function.
		Taking $t\to 0$ in  the moment of the residual life $\phi_n(t;\boldsymbol{\theta})$ we  see that the entire moments are justified by the above identity.
		\item[$\bullet$] 
		{\bf Tsallis entropy.} 
		As an immediate application of the first main tool, we have that if $q(\alpha-1)>-1$ then the Tsallis entropy \cite{Tsallis1988} is given by
		\begin{align*}
			\displaystyle
			S_q(X)&=	 
			\dfrac{1}{q-1}\, \Bigg[1-\int_{0}^{\infty} f^q(x;\boldsymbol{\theta}) \, {\rm d}x\Bigg]
			\\[0,3cm]	
			&=	
			\dfrac{1}{q-1}\, \Bigg[1-\bigg(\dfrac{\alpha}{\beta}\bigg)^{q-1}
			\dfrac{1}{q^{q+(1-q)/\alpha}}\,
			\Gamma\bigg(q+{1-q\over\alpha}\bigg)\Bigg],
			\quad q\in\mathbb{R} \backslash \{1\}.
		\end{align*}
		\begin{itemize}
			\item %[{\scalefont{0,45}$\blacksquare$}] 
			{\bf Quadratic entropy.}
			It is followed directly by Tsallis entropy that, 
			if $\alpha>1/2$ then the quadratic entropy can be written as
			\begin{align*}
				\displaystyle 
				H_2(X)
				&=
				-\log\Bigg[\int_{0}^{\infty} f^2(x;\boldsymbol{\theta}) \, {\rm d}x\Bigg]
				\\[0,2cm]	
				&=
				\log(\beta)
				-\log(\alpha)
				+\bigg(2-\dfrac{1}{\alpha}\bigg) \log(2)
				-\log\Bigg[
				{\displaystyle 
					\Gamma\bigg(2-{1\over\alpha}\bigg)
				}
				\Bigg].
			\end{align*}	
			%
			%\vspace*{0,3cm} \hspace*{0,5cm}
			%The proof of this formula follows directly from Tsallis entropy.
			
			\item %[{\scalefont{0,45}$\blacksquare$}] 
			{\bf Shannon entropy.} By combining the Tsallis entropy with the well known relation
			$
			H_1(X)=\lim_{q\to 1} S_q(X),
			$ we have that the Shannon entropy can be written as
			\begin{align*}
				\displaystyle
				H_1(X)
				=
				-\int_{0}^{\infty} f(x;\boldsymbol{\theta}) \log\big[f(x;\boldsymbol{\theta})\big] \, {\rm d}x
				= \displaystyle
				1+\log\bigg(\dfrac{\beta}{\alpha}\bigg)+\bigg(1-\dfrac{1}{\alpha}\bigg) \gamma,
			\end{align*}
			where $\gamma=-{{\rm d} \Gamma(x)\over {\rm d}x}\big\vert_{x=1}\approx 0.57721$ is the Euler-Mascheroni constant.
			%
			%\vspace*{0,3cm}\hspace*{0,5cm}
			%The proof of the Shannon entropy formula follows directly by combining the Tsallis entropy with the following known relation
			%\begin{align*}
			%H_1(X)=\lim_{q\to 1} S_q(X).
			%\end{align*}		
		\end{itemize}
	\end{itemize}
	
	\item[8)] {\bf Moment generating function.}
	By applying Fubini's Theorem we have that the moment generating function, $M_X(t)=\mathbb{E}[\exp({tX})]$, can be expressed as follows 
	\begin{align*}
		\displaystyle 
		M_X(t)=
		\begin{cases}
			\displaystyle
			\dfrac{1}{1-\beta t}
			& \text{for} \ \alpha=1 \ \text{and} \ \vert t\vert<1/\beta,
			\\[0,4cm]
			\displaystyle
			\sum_{n=0}^{\infty} {(\beta t)^n\over n!}\,
			\Gamma\bigg(1+{n\over\alpha}\bigg)
			& \text{for} \ \alpha>1 \ \text{and} \ t\in\mathbb{R}.
		\end{cases}
	\end{align*}
	%		\textcolor{red}{"We apply Fubini's Theorem to have second case of $M_X(t)$" I think let us write this sentence. }
	\item[9)] {\bf Light-tailed distribution.} From Item 8) it follows that,  if $\alpha\geqslant 1$ then there exists $t_0 > 0$ such that $\mathbb{P}(X > x)\leqslant \exp(-t_0x)$ for $x$ large enough.

	\item[10)] {\bf Second main tool.}  A simple change of variable shows that
	\begin{align*}
		& \displaystyle
		\mathbb{E}\big[X^s\log(X) f^{r-1}(X;\boldsymbol{\theta})\big]
		\\[0,3cm]	
		& =
		{\beta^{s-r+1} \alpha^{r-1}\over r^{1+\zeta/\alpha}} \,
		\Gamma\bigg(1+\dfrac{\zeta}{\alpha}\bigg)
		\Bigg[\log\bigg({\beta\over r^{1/\alpha}}\bigg)
		+
		{1\over\alpha}\, 
		\Psi^{(0)}\bigg(1+\dfrac{\zeta}{\alpha}\bigg)
		\Bigg], \quad \zeta>-\alpha,
	\end{align*}
	where $\zeta=s+(r-1)(\alpha-1)$, $r\in\mathbb{R}$, and  $\Psi^{(m)}(z)={{\rm d}^{m+1}\log[\Gamma(z)] \over {\rm d}z^{m+1}}$ is the polygamma function of order $m$.
	
	\vspace*{0,3cm}\hspace*{0,5cm}
	Indeed, by taking the change of variable $w=r\big({x\over\beta}\big)^\alpha$ we have, for each $\zeta>-\alpha$,
	\begin{align}
		\mathbb{E}\big[X^s\log(X) f^{r-1}(X;\boldsymbol{\theta})\big]
		&= \displaystyle 
		\bigg({\alpha\over\beta}\bigg)^r\, 
		\int_{0}^{\infty} x^{s}\log(y)\,  \bigg({x\over\beta}\bigg)^{r(\alpha-1)} \exp\Bigg[-r\bigg({x\over\beta}\bigg)^{\alpha}\Bigg]\, {\rm d}x \nonumber
		\\[0,3cm]
		&=
		{\beta^{s-r+1} \alpha^{r-1}\over r^{1+\zeta/\alpha}} \,
		\int_{0}^{\infty} w^{\zeta\over\alpha}
		\Bigg[\log\bigg({\beta\over r^{1/\alpha}}\bigg)+{1\over\alpha}\, \log(w)\Bigg]
		\exp(-w)\, {\rm d}w. \label{form-1}
	\end{align}
	By combining the following known formulas
	\begin{align}
		&\int_{0}^{\infty} w^{\zeta\over\alpha} \exp(-w)\, {\rm d}w
		=
		\Gamma\bigg(1+\dfrac{\zeta}{\alpha}\bigg), \label{formula1}
		\\[0,2cm]
		&\int_{0}^{\infty} w^{\zeta\over\alpha} \log(w)  \exp(-w)\, {\rm d}w
		=
		\Gamma\bigg(1+\dfrac{\zeta}{\alpha}\bigg)
		\Psi^{(0)}\bigg(1+\dfrac{\zeta}{\alpha}\bigg), \label{formula2}
	\end{align}
	with the identity \eqref{form-1}, the formula for
	the expectation $\mathbb{E}\big[X^s\log(X) f^{r-1}(X;\boldsymbol{\theta})\big]$ follows.
	
	\begin{itemize}
		\item[$\bullet$] In the particular case when $r=1$ we have
		\begin{align*}
			\displaystyle
			\mathbb{E}[X^s\log(X)]
			=
			\beta^s \Gamma\bigg(1+\dfrac{s}{\alpha}\bigg)
			\Bigg[\log(\beta)  
			+
			{1\over\alpha}\, 
			\Psi^{(0)}\bigg(1+\dfrac{s}{\alpha}\bigg)
			\Bigg], \quad s>-\alpha.
		\end{align*}
	\end{itemize}

	\item[11)] {\bf Third main tool.} Analogously to the proof of Item 10), a simple change of variable shows that, for each $s>-\alpha$,
	\begin{align*}
		\displaystyle
		&\mathbb{E}\big[X^s\log^2(X)  f^{r-1}(X;\boldsymbol{\theta})\big]
		\\[0,2cm]	
		&=
		{2\beta^{s-r+1} \alpha^{r-2}\over r^{1+\zeta/\alpha}} \,
		\Gamma\bigg(1+\dfrac{\zeta}{\alpha}\bigg)
		\log\bigg({\beta\over r^{1/\alpha}}\bigg) 
		\Psi^{(0)}\bigg(1+\dfrac{\zeta}{\alpha}\bigg)
		\displaystyle
		\\[0,2cm]	
		&
		+
		{\beta^{s-r+1} \alpha^{r-1}\over r^{1+\zeta/\alpha}} \,
		\Gamma\bigg(1+\dfrac{\zeta}{\alpha}\bigg)
		\Bigg\{
		\log^2\bigg({\beta\over r^{1/\alpha}}\bigg)  
		+
		{1\over\alpha^2}\, \bigg[
		\Psi^{(0)}\bigg(1+\dfrac{\zeta}{\alpha}\bigg)^2
		+
		\Psi^{(1)}\bigg(1+\dfrac{\zeta}{\alpha}\bigg) \bigg]
		\Bigg\},
	\end{align*}
	where $\zeta=s+(r-1)(\alpha-1)$ and $r\in\mathbb{R}$ and $\zeta>-\alpha$.
	
	\vspace*{0,1cm}
	\hspace*{0,5cm} 
	The proof of 
	this formula follows by taking the change of variable $w=r\big({x\over\beta}\big)^\alpha$ and by combining the formulas \eqref{formula1} and \eqref{formula2} with the following formula, for each $\zeta>-\alpha$,
	\begin{align*}
		\hspace{-0,5cm}
		\int_{0}^{\infty} w^{\zeta\over\alpha} \log^2(w) \exp(-w)\, {\rm d}w
		=
		\Gamma\bigg(1+\dfrac{\zeta}{\alpha}\bigg)
		\Bigg[
		\Psi^{(0)}\bigg(1+\dfrac{\zeta}{\alpha}\bigg)^2
		+
		\Psi^{(1)}\bigg(1+\dfrac{\zeta}{\alpha}\bigg) \Bigg].
	\end{align*} 
	\begin{itemize}
		\item[$\bullet$] In the particular case when $r=1$ we have, for $s>-\alpha$,
		{\scalefont{0,87}
			\begin{align*}
				\displaystyle
				%&
				\hspace{-1,3cm}
				\mathbb{E}[X^s\log^2(X)]
				%\\[0,2cm]	
				%&\hspace{-0,5cm}
				=
				\beta^s \Gamma\bigg(1+\dfrac{s}{\alpha}\bigg)
				\Bigg\{
				\log^2(\beta)  
				+
				{1\over\alpha^2}\, \bigg[
				\Psi^{(0)}\bigg(1+\dfrac{s}{\alpha}\bigg)^2
				+
				\Psi^{(1)}\bigg(1+\dfrac{s}{\alpha}\bigg) \bigg]
				\displaystyle
				+
				{2 \log(\beta)\over\alpha}\, 
				\Psi^{(0)}\bigg(1+\dfrac{s}{\alpha}\bigg)
				\Bigg\}.
			\end{align*}
		}
	\end{itemize}

	%\item[15)] {\bf Third main tool.}
	%A simple extension of Item 7) is as follows
	%\begin{align*}
	%&\displaystyle
	%\mathbb{E}\big[X^s f^{r-1}(X;\boldsymbol{\theta})\big]
	%=
	%\dfrac{\alpha^{r-1} \beta^{s-r+1}}{r^{r+(s-r+1)/\alpha}} \,
	%\beta^s \Gamma\bigg(r+{s-r+1\over\alpha}\bigg), \quad s-r+1>-\alpha r.
	%\end{align*}

	%\item[16)] {\bf Fourth main tool.}
	%A simple extension of Item 13) is as follows
	%\begin{align*}
	%& \displaystyle
	%\mathbb{E}\big[X^s\log(X) f^{r-1}(X;\boldsymbol{\theta})\big]
	%\\[0,3cm]	
	%& =
	%{\beta^{s-r+1} \alpha^{r-1}\over r^{1+\zeta/\alpha}} \,
	%\Gamma\bigg(1+\dfrac{\zeta}{\alpha}\bigg)
	%\Bigg[\log\bigg({\beta\over r^{1/\alpha}}\bigg)
	%+
	%{1\over\alpha}\, 
	%\Psi^{(0)}\bigg(1+\dfrac{\zeta}{\alpha}\bigg)
	%\Bigg], \quad \zeta>-\alpha,
	%\end{align*}
	%where $\zeta=s+(r-1)(\alpha-1)$ and $r\in\mathbb{R}$.
	
	%\item[17)] {\bf Fifth main tool.}
	%A simple extension of Item 14) is as follows
	%	\begin{align*}
	%	\displaystyle
	%	&\hspace{-0,5cm}\mathbb{E}\big[X^s\log^2(X)  f^{r-1}(X;\boldsymbol{\theta})\big]
	%	\\[0,2cm]	
	%	&\hspace{-0,5cm}=
	%{\beta^{s-r+1} \alpha^{r-1}\over r^{1+\zeta/\alpha}} \,
	%	\Gamma\bigg(1+\dfrac{\zeta}{\alpha}\bigg)
	%	\Bigg\{
	%	\log^2\bigg({\beta\over r^{1/\alpha}}\bigg)  
	%	+
	%	{1\over\alpha^2}\, \bigg[
	%	\Psi^{(0)}\bigg(1+\dfrac{\zeta}{\alpha}\bigg)^2
	%	+
	%	\Psi^{(1)}\bigg(1+\dfrac{\zeta}{\alpha}\bigg) \bigg]
	%	\displaystyle
	%	\\[0,2cm]	
	%&\hspace{-0,5cm}
	%	+
	%	{2\over\alpha}\, \log\bigg({\beta\over r^{1/\alpha}}\bigg) 
	%	\Psi^{(0)}\bigg(1+\dfrac{\zeta}{\alpha}\bigg)
	%	\Bigg\}, \quad \zeta>-\alpha,
	%	\end{align*}
	%where $\zeta=s+(r-1)(\alpha-1)$ and $r\in\mathbb{R}$.
\end{itemize}

%\textcolor{red}{General comment for the items in the properties: I think it seems that they are enough, that is there is no need to add. Thanks. Do you agree with me? Anyway... we can think and discuss...}
\section{Inference: Estimation Methods and Fisher information}\label{section3inference}
\subsection{Maximum likelihood estimation method}
Maximum likelihood is the standard approach in parametric estimation, mainly
due to the desirable asymptotic properties of consistency, efficiency and asymptotic normality under some regularity conditions (\cite{LehmannCas98}).

\begin{itemize}
	\item %[{\scalefont{0,45}$\blacksquare$}]
	The log-likelihood function for $\boldsymbol{\theta}$ is given by
	\begin{align*}
		l(\boldsymbol{\theta}; \boldsymbol{x})
		=
		n\log\Big(\frac{\alpha}{\beta}\Big) +
		\sum_{i=1}^{n} 
		\left\{
		(\alpha-1)\log\Big(\frac{x_i}{\beta}\Big) - \Big({x_i \over\beta}\Big)^{\alpha} 
		\right\}.
	\end{align*}
	
	A standard calculation shows that the first-order partial derivatives of $l(\boldsymbol{\theta}; \boldsymbol{x})$ are
	\begin{align}
		{\partial l(\boldsymbol{\theta}; \boldsymbol{x})\over\partial\alpha}
		&=
		{n\over\alpha}-n\log(\beta)+\sum_{i=1}^{n}\log(x_i)-
		{1\over\beta^{\alpha}} \sum_{i=1}^{n}x_i^\alpha \log(x_i),
		\label{1der}
		\\[0,2cm]
		{\partial l(\boldsymbol{\theta}; \boldsymbol{x})\over\partial\beta}
		&=
		-{n\alpha\over\beta}+{\alpha\over\beta^{\alpha+1}} \sum_{i=1}^{n}x_i^\alpha.
		\label{2der}
	\end{align}

	The second-order partial derivatives of $l(\boldsymbol{\theta}; \boldsymbol{x})$ can be written as
	\begin{align*}
		{\partial^2 l(\boldsymbol{\theta}; \boldsymbol{x})\over\partial\alpha^2}
		&=
		-{n\over\alpha^2}+{\log(\beta)\over\beta^{\alpha}} \sum_{i=1}^{n}x_i^\alpha \log(x_i) - 
		{1\over\beta^{\alpha}} \sum_{i=1}^{n}x_i^\alpha \log^2(x_i),
		\\[0,2cm]
		{\partial^2 l(\boldsymbol{\theta}; \boldsymbol{x})\over\partial\beta^2}
		&=
		{n\alpha\over\beta^2}-{\alpha(\alpha+1)\over\beta^{\alpha+2}} \sum_{i=1}^{n}x_i^\alpha,
		\\[0,2cm]
		{\partial^2 l(\boldsymbol{\theta}; \boldsymbol{x})\over\partial\alpha\partial\beta}
		&=
		{\partial^2 l(\boldsymbol{\theta}; \boldsymbol{x})\over\partial\beta\partial\alpha}
		=
		-{n\over\beta}+{\alpha\over\beta^{\alpha+1}} \sum_{i=1}^{n}x_i^\alpha \log(x_i).
	\end{align*}
	%%%==============================================================
	
	\item %[{\scalefont{0,45}$\blacksquare$}]
	The $\log_q$-likelihood function for $\boldsymbol{\theta}$ is given by
	\begin{align*}
		l_q(\boldsymbol{\theta}; \boldsymbol{x})
		= \sum_{i=1}^{n} {f^{1-q}(x_i;\boldsymbol{\theta})-1 \over 1-q}.
	\end{align*}
	Note that the first-order partial derivatives of $l_q(\boldsymbol{\theta}; \boldsymbol{x})$ are
	\begin{align*}
		{\partial l_q(\boldsymbol{\theta}; \boldsymbol{x})\over\partial\alpha}
		&=
		\sum_{i=1}^{n} f^{-q}(x_i;\boldsymbol{\theta}) \, {\partial f(x_i;\boldsymbol{\theta})\over\partial\alpha}
		%\\[0,2cm]
		%&
		=
		\sum_{i=1}^{n} f^{1-q}(x_i;\boldsymbol{\theta}) \, 
		\Bigg\{\dfrac{1}{\alpha}+\Bigg[1-\bigg(\dfrac{x_i}{\beta}\bigg)^\alpha\Bigg]\log\bigg({x_i\over\beta}\bigg)\Bigg\},
		\\[0,2cm]
		{\partial l_q(\boldsymbol{\theta}; \boldsymbol{x})\over\partial\beta}
		&=
		\sum_{i=1}^{n} f^{-q}(x_i;\boldsymbol{\theta}) \, {\partial f(x_i;\boldsymbol{\theta})\over\partial\beta}
		=
		\dfrac{\alpha}{\beta}
		\sum_{i=1}^{n} f^{1-q}(x_i;\boldsymbol{\theta}) \, 
		\Bigg[\bigg(\dfrac{x_i}{\beta}\bigg)^\alpha-1\Bigg].
	\end{align*}
\end{itemize}

%\textcolor{red}{please give first derivative for $\log_q$. We do not need to second derivatives of $\log_q$ with respect to parameters. I give q-Fisher information in Eq. \eqref{qFisher}}

\subsection{Fisher information}
The Fisher information matrix is defined by
\begin{align}\label{abdF}
	F=n\begin{bmatrix} 
		E_{\alpha \alpha} & E_{\alpha \beta}  \\
		E_{\beta \alpha} & E_{\beta \beta}     
	\end{bmatrix},
\end{align}
%\noindent 
where $n$ is sample size. $E$ is integral for partial derivatives of $\log(L)$ according to parameters and it is taken by probability density function $f(x;\boldsymbol{\theta})$. The subscript in $E$ represents second-order partial derivatives of $\log(L)$ according to parameters $\alpha $ and $ \beta$. In other words, if $X\sim \text{Weibull}(\boldsymbol{\theta})$, by using Items 7), 10) and 11) of Subsection \ref{Properties}, we have 

\begin{align*}
	\begin{array}{lllll}
		E_{\alpha \alpha}
		&= \displaystyle
		-{1\over\alpha^2}+{\log(\beta)\over\beta^{\alpha}}\, \mathbb{E}\big[X^\alpha \log(X)\big] - 
		{1\over\beta^{\alpha}}\, \mathbb{E}\big[X^\alpha \log^2(X)\big]
		\\[0,3cm]
		&= \displaystyle
		- 
		{1\over\alpha^2}\, 
		\big[
		1
		+
		\Psi^{(1)}(2) 
		+
		\Psi^{(0)}(2)^2
		+
		{\alpha\log(\beta)}\, 
		\Psi^{(0)}(2)
		\big],
		\\[0,4cm]
		E_{\beta \beta}
		&= \displaystyle
		{\alpha\over\beta^2}-{\alpha(\alpha+1)\over\beta^{\alpha+2}}\, \mathbb{E}\big(X^\alpha\big)
		=
		-{\alpha^2\over\beta^2},
		\\[0,4cm]
		E_{\alpha \beta}
		&= \displaystyle
		E_{\beta \alpha}
		=
		-{1\over\beta}+{\alpha\over\beta^{\alpha+1}}\, \mathbb{E}\big[X^\alpha \log(X)\big]
		=
		-{1\over\beta}+{\alpha\over\beta}\,
		\bigg[\log(\beta)  
		+
		{1\over\alpha}\, 
		\Psi^{(0)}(2)
		\bigg].
	\end{array}
\end{align*}
Since, $\Psi^{(0)}(2)={{\rm d}\log [\Gamma(x)]\over {\rm d}x}\big\vert_{x=2}=1-\gamma$ and $\Psi^{(1)}(2)={{\rm d^2}\log [\Gamma(x)]\over {\rm d}x^2}\big\vert_{x=2}={\pi^2\over 6}-1$, where
$\gamma=-{{\rm d} \Gamma(x)\over {\rm d}x}\big\vert_{x=1}\approx 0.57721$ is the Euler-Mascheroni constant.
Then the matrix \eqref{abdF} is given by
\begin{align*}
	F=n\begin{bmatrix}  \displaystyle
		- 
		{1\over\alpha^2}\, 
		\bigg\{
		{\pi^2\over 6}
		+
		(1-\gamma)
		\big[
		{\alpha\log(\beta)}+(1-\gamma)
		\big]
		\bigg\} 
		&&& \displaystyle
		{\alpha\over\beta}\,
		\bigg\{\log(\beta)  
		-
		{\gamma\over\alpha}
		\bigg\}  
		\\[0,5cm] \displaystyle
		{\alpha\over\beta}\,
		\bigg\{\log(\beta)  
		-
		{\gamma\over\alpha}
		\bigg\} 
		&&&  \displaystyle
		-{\alpha^2\over\beta^2}     
	\end{bmatrix}.
\end{align*}

\begin{theorem}
	Let $\Theta=\{\boldsymbol{\theta}=(\alpha,\beta)\in\mathbb{R}^+\times\mathbb{R}^+ :  \alpha \ \text{is known} \ \text{and} \ \beta> 1 \}$ be the parameter space.
	Then, with probability approaching $1$, as $n\to\infty,$
	the likelihood equation
	${{\rm d}\, l(\boldsymbol{\theta}; \boldsymbol{x})\over{\rm d}\beta}=0$ has a consistent solution,
	denoted by $\widehat{\beta}$.
\end{theorem}
\begin{proof}
	If $X\sim \text{Weibull}(\boldsymbol{\theta})$, a simple calculation shows that
	\begin{enumerate}
		\item $\mathbb{E}\big[{{\rm d}\log f(X;\boldsymbol{\theta})\over{\rm d}\beta}\big]=-{\alpha\over\beta}+{\alpha\over\beta^{\alpha+1}} \, \mathbb{E}(X^\alpha) =0$
		for all $\boldsymbol{\theta}\in\Theta$;
		\item $-\infty<\mathbb{E}\big[{{\rm d}^2 \log f(X;\boldsymbol{\theta})\over{\rm d}\beta^2}\big]=-{\alpha^2\over\beta^2}<0$
		for all $\beta\in\Theta$;
		\item There exits a function $H(x)$ such that for all $\boldsymbol{\theta}\in\Theta$,
		\begin{align*}
			\biggl|{{\rm d}^3 \log f(x;\boldsymbol{\theta})\over{\rm d}\beta^3}\biggr|
			&=
			\biggl|-{2\alpha\over\beta^3}-{\alpha(\alpha+1)(\alpha+2)\over\beta^{2\alpha+4}} \, \beta^{\alpha+1} x^\alpha \biggl|
			\\[0,2cm]
			&
			<2\alpha+\alpha(\alpha+1)(\alpha+2)x^{\alpha} =H(x),
		\end{align*}
		because $\beta>1$, and
		\begin{align*}
			\mathbb{E}\big[H(X)\big]&=2\alpha+\alpha(\alpha+1)(\alpha+2)\mathbb{E}(X^{\alpha})
			\\[0,2cm]
			&= 2\alpha+\alpha(\alpha+1)(\alpha+2) \beta^{\alpha} =M(\beta)<\infty.
		\end{align*}
	\end{enumerate}
	Hence, by \cite{Cramer46} the proof follows.
\end{proof}

\subsection{Estimating equations derived by objective functions in the composite likelihood}

The maximum composite likelihood estimation (MCLE) is a generalization of maximum likelihood estimation  (MLE), MCLE is given by
\begin{equation}
	L_C\big[f(\boldsymbol{x};\boldsymbol{\theta})\big]=\prod_{i=1}^{n} f(x_i;\boldsymbol{\theta})^{w_i}
\end{equation}
\noindent $w_i \in \mathbb{R}$ is a weight function, $\boldsymbol{x}=(x_1,x_2,\ldots,x_n)$ and $f(x;\boldsymbol{\theta})$ is a p.d. function. MCLE can cover density power divergence and its generalized forms (\cite{varincomposite,alphabetadivergences}). The M-estimators are obtained by optimizing the objective function:
\begin{equation}\label{bigfamily}
	\sum_{i=1}^{n} \rho(x_i;\boldsymbol{\theta}):=\sum_{i=1}^{n}  w_i \log\big[f(x_i;\boldsymbol{\theta})\big]=\sum_{i=1}^{n}\Lambda\big[f(x_i;\boldsymbol{\theta})\big].
\end{equation}

When $w_i=1$ and $\log$ is replaced by $\log_q$, we have MLqE. $\log_q(f)={f^{1-q} -1 \over 1-q}, ~\lim_{q \to 1} \log_{q}(f)=\log(f)$. $\log_{\kappa}(f)={f^{\kappa}-f^{-\kappa} \over 2 \kappa}, ~ \lim_{\kappa \to 0} \log_{\kappa}(f)=\log(f)$ and $\log_q$ and $\log_{\kappa}$  are deformed logarithm of $\log$. $q \in \mathbb{R} \backslash \{1\},\kappa \in \mathbb{R}$ and $\gamma \geq 0$ are tuning constants used to adjust robustness and also efficiency (\cite{FerrariYang10,Wadatwopara,Vajda86}). The concavity property of $\Lambda$ is examined by \cite{CanKor18,Jan2} and references therein to use $\Lambda$ for the estimation process accurately.

The density power divergence (DPD) as an objective function  between $f(x;\boldsymbol{\theta})$ and $g(x)$ which is free from parameter was proposed and the reorganized form of DPD is given by \cite{Basuetal98}:
\begin{eqnarray}\label{DPDsimple}
	\sum_{i=1}^{n} \rho(x_i;\boldsymbol{\theta})
	:=\sum_{i=1}^{n}\Lambda\big[f(x_i;\boldsymbol{\theta})\big]
	=\int f(x;\boldsymbol{\theta})^{1+\gamma} \ud x - \bigg(1+\dfrac{1}{\gamma}\bigg){1 \over n} \sum_{i=1}^{n}f(x_i;\boldsymbol{\theta})^{\gamma}.
\end{eqnarray}

Let us try to get $\log_q$ from DPD after algebraic rearrangement and $\gamma=1-q>0$. Since $q<1$,  DPD puts a restriction for the values of tuning constant when we compare with $\log_q$. If equation \eqref{DPDsimple} is rewritten for $\gamma=1-q$, then we have the following expression: 
\begin{equation}\label{connectionqDPD}
	{2-q \over n} \sum_{i=1}^{n} \log_q\big[f(x_i;\boldsymbol{\theta})\big]+\frac{n}{1-q}+\int f(x;\boldsymbol{\theta})^{2-q} \ud x =\sigma \sum_{i=1}^{n} \log_q\big[f(x_i;\boldsymbol{\theta})\big] + \mu,
\end{equation}
\noindent where $\sigma$ and $\mu$ represent scale and location for $\sum_{i=1}^{n} \log_q\big[f(x_i;\boldsymbol{\theta})\big]$. When we consider to apply the optimization for DPD and maximum $\log_q$-likelihood,  $\sigma$ and $\mu$ will change where the optimized region is. In other words, there is an equivalence between DPD and  $\log_q$ if they are optimized according to parameters $\boldsymbol{\theta}$. Further, the integral value of $I=\int f(x;\boldsymbol{\theta}) \ud x$ in DPD depends on Gamma function if $f$ is Weibull (see \cite{Cankaya2018}). The arguman of Gamma function, i.e. $\Gamma, r>0$, has to be a positive and requires that the values of parameters remain within certain values, which is disadvantegous for a case in estimation. In addition, the computation time \cite{Muhammedbadwei} for numerical integration can be high according to the used p.d. function which is not tractable to calculate and get an expression for the result of integral. MCLE method is useful to get robust estimators for parameters. As a generalized form of MCLE, we can use MLqE which is for robust estimation, because MLqE is a simple method and does not have extra conditions. If $I$ does not exist, it is mandatory to show that $I$ is finite for the values of parameters in a p.d. function by use of tools in (\cite{Malikcondition}). Otherwise, $f$ in $I$  does not have a finite value, which shows that we will not have $\hat{\boldsymbol{\theta}}$.  DPD for estimation of parameters of Weibull distribution does not work properly, as introduced by  (\cite{Muhammedbadwei}).

Let us derive the estimating equations (EEs) to examine the role of objective functions. EEs are obtained after taking derivatives of objective functions, $\rho=\Lambda(f)$, according to parameters $\boldsymbol{\theta}$. For all of parameters, we have a system of EEs. The rearrangement forms of EEs are as follows:
\begin{equation} \label{eesq}
	\sum_{i=1}^{n} {\partial \log_q \big[f(x_i;\boldsymbol{\theta})\big] \over \partial \boldsymbol{\theta}} =\sum_{i=1}^{n} f(x;\boldsymbol{\theta})^{1-q}\mathcal{Z}(x_i;\boldsymbol{\theta}) = \sum_{i=1}^{n} w_q(x_i;\boldsymbol{\theta}) \mathcal{Z}(x_i;\boldsymbol{\theta})= \boldsymbol{0},
\end{equation}
\begin{equation} \label{eesk}
	\sum_{i=1}^{n}  {\partial \log_{\kappa} \big[f(x_i;\boldsymbol{\theta})\big] \over \partial \boldsymbol{\theta}} 
	=	
	\sum_{i=1}^{n}  \big[f(x_i;\boldsymbol{\theta})^{\kappa}+f(x_i;\boldsymbol{\theta})^{-\kappa}\big]\,  {\mathcal{Z}(x_i;\boldsymbol{\theta}) \over 2} 
	= 
	\sum_{i=1}^{n} w_{\kappa}(x_i;\boldsymbol{\theta}) \mathcal{Z}(x_i;\boldsymbol{\theta}) = \boldsymbol{0},
\end{equation}
\begin{equation} \label{eesa}
	\sum_{i=1}^{n}  	{\partial \log\big[ \propto + f(x_i;\boldsymbol{\theta}) \big] \over \partial \boldsymbol{\theta}} 
	= 	
	\sum_{i=1}^{n}   {f(x_i;\boldsymbol{\theta}) \over \propto+f(x_i;\boldsymbol{\theta})}\,
	\mathcal{Z}(x_i;\boldsymbol{\theta})
	=
	\sum_{i=1}^{n} w_{\propto}(x_i;\boldsymbol{\theta}) \mathcal{Z}(x_i;\boldsymbol{\theta}) = \boldsymbol{0},
\end{equation}
\begin{equation} \label{eesdpd}
	{1 \over n}	\sum_{i=1}^{n} \mathcal{Z}(x_i;\boldsymbol{\theta}) f(x_i;\boldsymbol{\theta})^{\gamma} - I  = \sum_{i=1}^{n} w_{\gamma}(x_i;\boldsymbol{\theta}) \mathcal{Z}(x_i;\boldsymbol{\theta}) - I  = \boldsymbol{0},
\end{equation}
\noindent where $I = \int  \mathcal{Z}(x_i;\boldsymbol{\theta}) f(x_i;\boldsymbol{\theta})^{1+\gamma} \ud x$. Eqs. \eqref{eesq}-\eqref{eesdpd} are the weighted score function with $w_{q},w_{\kappa},w_{\propto}$ and $w_{\gamma}$.
$\mathcal{Z}(x_i;\boldsymbol{\theta})={{\partial f(x_i;\boldsymbol{\theta}) \over \partial \boldsymbol{\theta}}  / f(x_i;\boldsymbol{\theta})}$ is a score function.  Note that DPD has a weighting $0 \leq w_{\gamma} \leq 1$ which is disadvantegous because of the poorness in modeling capability (\cite{Basuetal98}).  Let us rewrite the system of estimating equations (EE) given by the following form:
\begin{equation}\label{godambeeq}
	\sum_{i=1}^{n}  w(x_i;\boldsymbol{\theta}) \mathcal{Z}(x_i;\boldsymbol{\theta}) = \boldsymbol{0}.
\end{equation}

Note that $w_q$ is bigger than 1 if $q>1$, which shows an advantage for us when we compare with $ 0 \leq  w_{\propto},w_{\gamma} \leq 1$ for $\propto, \gamma \geq 0$. If  $0 \leq w \leq 1$, then a partial form of $\mathcal{Z}$ is produced. If $w >1$, then   $(1+v)  \mathcal{Z}$ as an extended form of $\mathcal{Z}$ is produced, which shows us that the information gained from the joint role of $w$ and $\mathcal{Z}$ at same time is managed by not only $\mathcal{Z}$ from $f^{'}/f$  but also the different values of $w$ from $f^{1-q}$. For this reason, $\log_q(f)$ as an objective function is used to perform an efficient fitting on a data set. Thus, it is seen that $w$ can manage to produce the efficient estimators $\hat{\boldsymbol{\theta}}$ from equation \eqref{godambeeq}. Note that the chosen $w$ affects the estimations of parameters which are shape and scale.  After solving the systems of the estimating equations according to parameters, the estimators $\hat{\boldsymbol{\theta}}$ are also obtained instead of optimizing the objective function (\cite{Udristegeo,Amari16}).
\subsection{Investigation for behaviour of objective functions in the composite likelihood}

Estimation is performed when we assume that the empirical data sets are a member of an objective function $\rho$. For this reason, we can remove integral or summation from relative entropy or divergence to have function part of analytical expression. Since $f$ is on the closed inverval $[0,1]$, we have advantage to display the behaviour of objective functions $\rho$ derived by $\Lambda(f)$ and the values of tuning parameters are chosen as the interval $[0,1]$. $\Lambda$ can be chosen as $\log_q$, $\log_{\kappa}$, $\log(\propto+f)$, etc (\cite{Lindsay94}).  Let us display their behaviours via Fig. \ref{functionscompare}: 
\begin{figure}[htbp]
	\centering
	\subfigure[]{\label{fig:logq1}\includegraphics[width=0.41\textwidth]{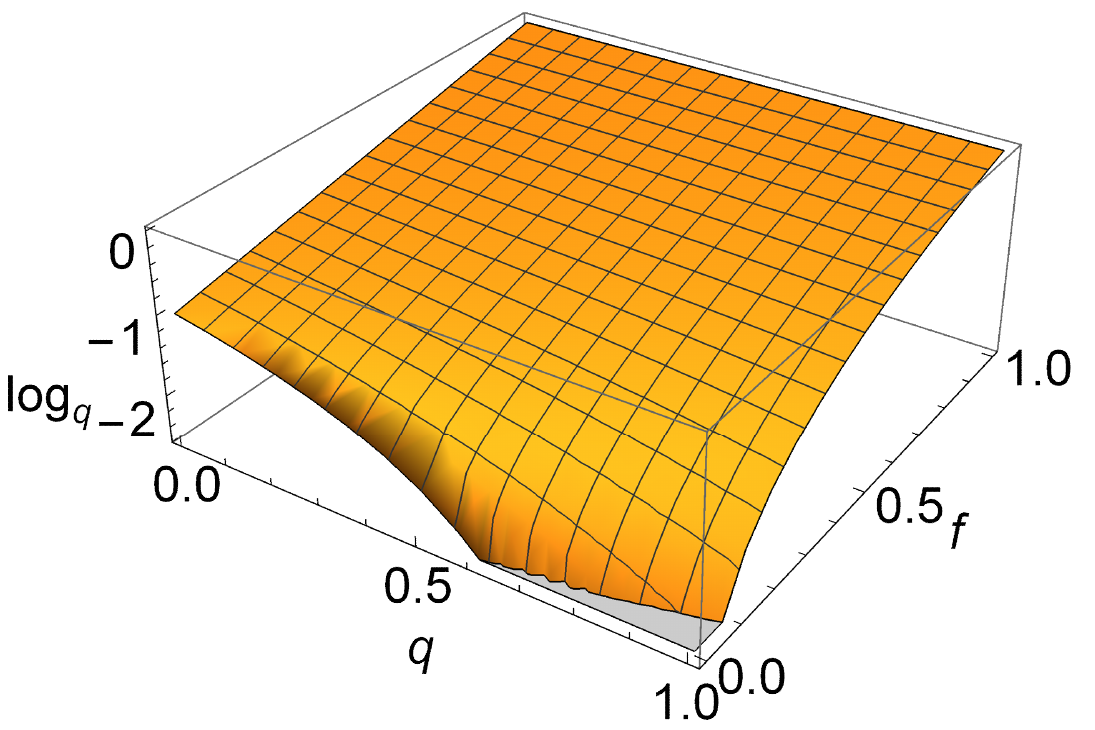}}
	\subfigure[]{\label{fig:logk}\includegraphics[width=0.41\textwidth]{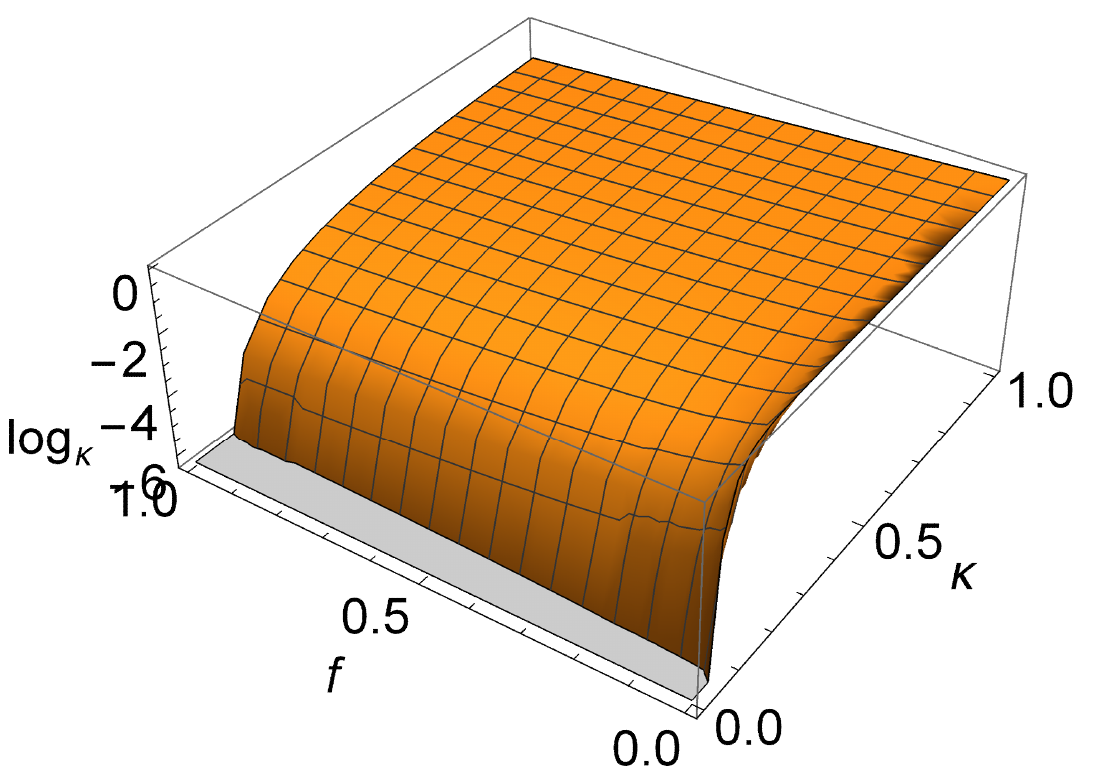}}
	\subfigure[]{\label{fig:logaf}\includegraphics[width=0.41\textwidth]{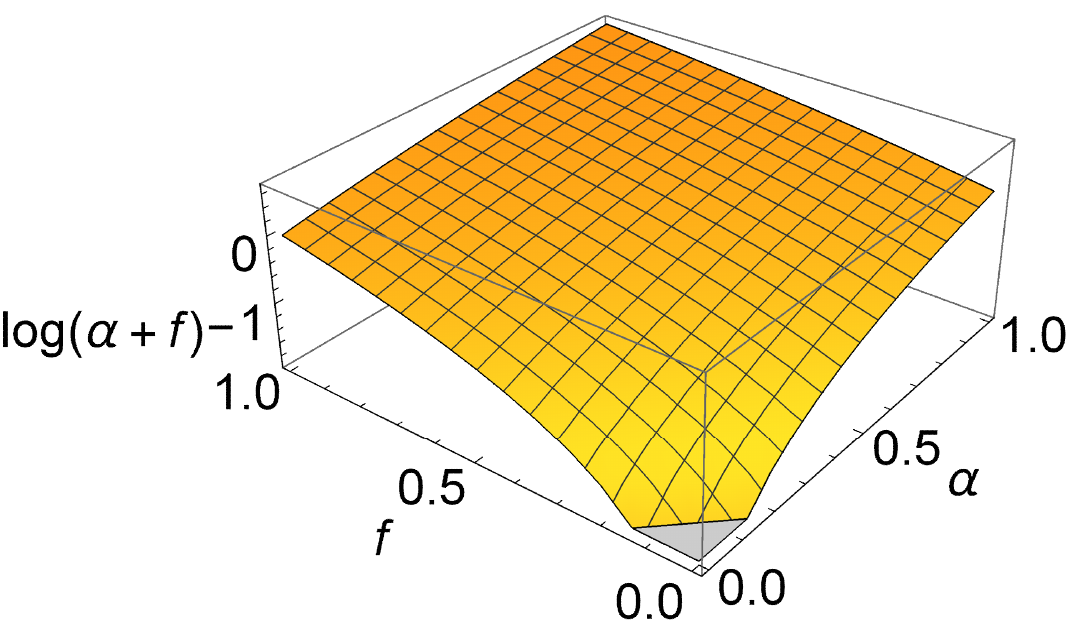}}
	\subfigure[]{\label{fig:logDPD}\includegraphics[width=0.41\textwidth]{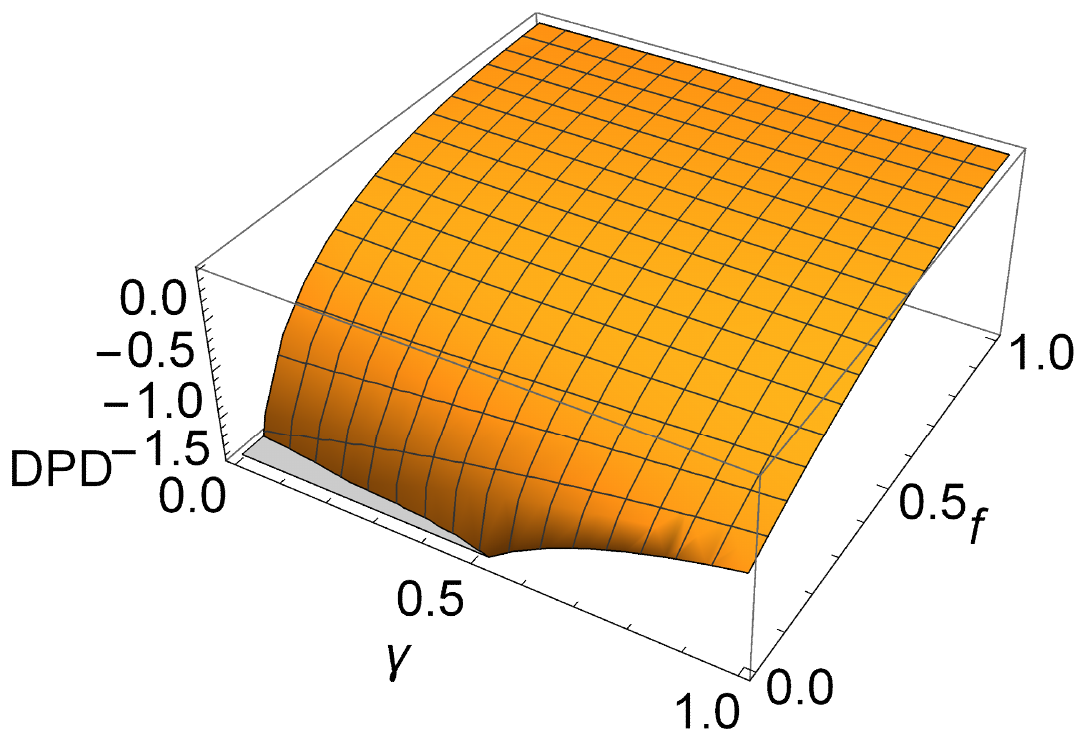}}
	\caption{$f \in [0,1]$ and $q,\kappa,\propto,\gamma \in [0,1]$}
	\label{functionscompare}
\end{figure}
Figs. \ref{fig:logq1} and \ref{fig:logk} have big range for the values of functions when compared with the Figs. \ref{fig:logaf} and \ref{fig:logDPD}.  As it seen from plots in Fig. \ref{functionscompare}, the concavity property of Figs. \ref{fig:logq1} and \ref{fig:logaf} can be better than  Figs. \ref{fig:logk} and \ref{fig:logDPD}.  When we look at plots in Fig. \ref{functionscompare}, we can observe that $\log_q$ is better than other functions from $\Lambda$ (\cite{BrondivergencesBas,Vajda86}).  Thus, we test the role of $\Lambda$ on p.d. function $f$.

%In other words, the flexibility of $\log_q$ due to the value of  $q$ and the role on $f$ can be higher than other ones.
%In second stage of our performance comparison for the values of functions in Fig. \ref{functionscompare}, let us make a comparison between Figs. \ref{fig:logq1} and \ref{fig:logk}. The behaviour of Fig. \ref{fig:logk} makes a sudden peak for the changing values of $\kappa \in [0,1]$. However,  Fig. \ref{fig:logq1} increases slowly.

\subsection{$q$-Fisher information}
Let us remind the definition of Fisher information based on $\log_q$ (\cite{CanKor18}). 
\begin{align}\label{qFisher}
	&_q \mathbb{E}
	\left\{{\partial \log f(X;\boldsymbol{\theta}) \over \partial \boldsymbol{\theta}} \cdot \left[ {\partial \log f(X;\boldsymbol{\theta})\over \partial \boldsymbol{\theta}}  \right]^{\intercal}  f^{1-q}(X;\boldsymbol{\theta})   \right\} \nonumber
	\\[0,3cm]
	&= 
	\int_{0}^\infty {\partial \log f(x;\boldsymbol{\theta})\over \partial \boldsymbol{\theta}}   \cdot 
	\left[{\partial \log f(x;\boldsymbol{\theta})\over \partial \boldsymbol{\theta}} \right]^{\intercal} 
	f^{2-q}(x;\boldsymbol{\theta})\,  \ud x.
\end{align}
Let us rewrite the Equation \eqref{qFisher} for the parameters $\alpha$ and $\beta$ as follows
%
%\begin{align*}
%_q F=n\begin{bmatrix}
%\displaystyle 
%\int \left\{{\partial \log\big[f(x;\boldsymbol{\theta})\big]\over \partial \alpha} \right\}^2 \, f^{2-q}(x;\boldsymbol{\theta}) \ud x
%& \displaystyle 
%\int {\partial \log\big[f(x;\boldsymbol{\theta})\big]\over \partial \alpha}  \, 
%{\partial \log\big[f(x;\boldsymbol{\theta})\big]\over \partial \beta}  \, f^{2-q}(x;\boldsymbol{\theta})  \ud x 
%\\[0,4cm]
%\displaystyle 
%\int {\partial \log\big[f(x;\boldsymbol{\theta})\big]\over \partial \alpha}  \, 
%{\partial \log\big[f(x;\boldsymbol{\theta})\big]\over \partial \beta}  \, f^{2-q}(x;\boldsymbol{\theta})  \ud x 
%& \displaystyle 
%\int \left\{{\partial \log\big[f(x;\boldsymbol{\theta})\big] \over \partial \beta} \right\}^2  \, f^{2-q}(x;\boldsymbol{\theta}) \ud x
%\end{bmatrix}.
%\end{align*}
%
\begin{align*}
	_q F=n\begin{bmatrix} 
		_q E_{\alpha \alpha} & _q E_{\alpha \beta}  \\
		_q E_{\beta \alpha} & _q E_{\beta \beta}     
	\end{bmatrix}.
\end{align*}
The elements of Fisher Information (FI) matrix based on $\log_q$ ($_qF$) can be written as the following form. 
\begin{itemize}
	\item 
	By taking $n=1$ in \eqref{1der}, we get
	\begin{align*}
		_q E_{\alpha \alpha}
		&= 
		_q\hspace{-0,1cm}\mathbb{E}
		\left\{\left[ {\partial \log f(X;\boldsymbol{\theta})\over \partial {\alpha}}  \right]^{2}  f^{1-q}(X;\boldsymbol{\theta})   \right\}
		\\[0,2cm]
		&=
		\bigg[{1\over\alpha}-\log(\beta)\bigg]^2 \mathbb{E}\big[f^{1-q}(X;\boldsymbol{\theta}) \big]
		\\[0,2cm]
		&
		+
		\sum_{\ell\in\{0,\alpha\}} a_\ell \,
		\mathbb{E}\big[X^\ell \log(X) f^{1-q}(X;\boldsymbol{\theta}) \big]
		+
		\sum_{\ell\in\{0,\alpha,2\alpha\}} \widetilde{a}_\ell \,
		\mathbb{E}\big[X^\ell \log^2(X) f^{1-q}(X;\boldsymbol{\theta}) \big],
		%
		%
		%
		%+\mathbb{E}\big[\log^2(X) f^{1-q}(X;\boldsymbol{\theta})\big]
		%\\[0,2cm]
		%&
		%+
		%2\bigg[{1\over\alpha}-\log(\beta)\bigg] \mathbb{E}\big[\log(X) f^{1-q}(X;\boldsymbol{\theta}) \big]
		%+
		%{1\over\beta^{2\alpha}}\, \mathbb{E}\big[X^{2\alpha} \log^2(X) f^{1-q}(X;\boldsymbol{\theta}) \big]
		%\\[0,2cm]
		%&
		%-
		%{2\over\bet}a^{\alpha}}\, \bigg[{1\over\alpha}-\log(\beta)\bigg]
		%\, \mathbb{E}\big[X^{\alpha} \log(X) f^{1-q}(X;\boldsymbol{\theta}) \big]
		%-
		%{2\over\beta^{\alpha}}\, 
		%\, \mathbb{E}\big[X^{\alpha} \log^2(X) f^{1-q}(X;\boldsymbol{\theta}) \big].
	\end{align*}
	where 
	$a_0=2\big[{1\over\alpha}-\log(\beta)\big]$, 
	$a_\alpha=-
	{2\over\beta^{\alpha}}\, \big[{1\over\alpha}-\log(\beta)\big]$, 
	$\widetilde{a}_0=1$, 
	$\widetilde{a}_\alpha=-{2\over\beta^{\alpha}}$ 
	and 
	$\widetilde{a}_{2\alpha}={1\over\beta^{2\alpha}}$. 
	By using Items 7), 10) and 11) of Subsection \ref{Properties}, we have
	\begin{align*}
		_q E_{\alpha \alpha}
		&=
		\bigg[{1\over\alpha}-\log(\beta)\bigg]^2 \bigg(\dfrac{\alpha}{\beta}\bigg)^{1-q}
		\dfrac{1}{q^{2-q+(q-1)/\alpha}}\,
		\Gamma\bigg(2-q+{q-1\over\alpha}\bigg)
		\\[0,2cm]
		&
		+
		\sum_{\ell\in\{0,\alpha\}} a_\ell \,
		{\beta^{\ell+q-1} \alpha^{1-q}\over (2-q)^{1+\zeta/\alpha}} \,
		\Gamma\bigg(1+\dfrac{\zeta}{\alpha}\bigg)
		\Bigg[\log\bigg({\beta\over \{2-q\}^{1/\alpha}}\bigg)
		+
		{1\over\alpha}\, 
		\Psi^{(0)}\bigg(1+\dfrac{\zeta}{\alpha}\bigg)
		\Bigg]
		\\[0,2cm]
		&
		+
		\sum_{\ell\in\{0,\alpha,2\alpha\}} \widetilde{a}_\ell \,
		{\beta^{\ell+q-1} \alpha^{1-q}\over (2-q)^{1+\zeta/\alpha}} \,
		\Gamma\bigg(1+\dfrac{\zeta}{\alpha}\bigg)
		\Bigg\{
		\log^2\bigg({\beta\over \{2-q\}^{1/\alpha}}\bigg)  
		+
		\displaystyle
		\\[0,2cm]	
		&
		+
		{1\over\alpha^2}\, \bigg[
		\Psi^{(0)}\bigg(1+\dfrac{\zeta}{\alpha}\bigg)^2
		+
		\Psi^{(1)}\bigg(1+\dfrac{\zeta}{\alpha}\bigg) \bigg]
		+
		{2\over\alpha}\, \log\bigg({\beta\over \{2-q\}^{1/\alpha}}\bigg) 
		\Psi^{(0)}\bigg(1+\dfrac{\zeta}{\alpha}\bigg)
		\Bigg\},
	\end{align*}
	where $\zeta=\ell+(1-q)(\alpha-1)$.
	
	\item
	Letting $n=1$ in \eqref{2der}, we get
	\begin{align*}
		_q E_{\beta \beta}
		&= 
		_q\hspace{-0,1cm}\mathbb{E}
		\left\{\left[ {\partial \log f(X;\boldsymbol{\theta})\over \partial {\beta}}  \right]^{2}  f^{1-q}(X;\boldsymbol{\theta})   \right\}
		=
		\sum_{\ell\in\{0,\alpha,2\alpha\}} b_\ell \,
		\mathbb{E}\big[X^\ell f^{1-q}(X;\boldsymbol{\theta}) \big],
		%\\[0,2cm]
		%&=
		%\bigg({\alpha\over\beta}\bigg)^2 
		%\mathbb{E}\big[f^{1-q}(X;\boldsymbol{\theta}) \big]
		%+
		%{\alpha^2\over\beta^{2(\alpha+1)}} \,
		%\mathbb{E}\big[X^{2\alpha} f^{1-q}(X;\boldsymbol{\theta}) \big]
		%-
		%{2\alpha^2\over\beta^{\alpha+2}}\, 
		%\mathbb{E}\big[X^\alpha f^{1-q}(X;\boldsymbol{\theta})\big],
	\end{align*}
	where $b_0=\big({\alpha\over\beta}\big)^2 $, 
	$b_\alpha=-{2\alpha^2\over\beta^{\alpha+2}}$ 
	and $b_{2\alpha}={\alpha^2\over\beta^{2(\alpha+1)}}$.
	By using Item 7) of Subsection \ref{Properties}, we reach
	\begin{align*}
		_q E_{\beta \beta}
		&
		=
		\sum_{\ell\in\{0,\alpha,2\alpha\}} b_\ell \,
		\dfrac{\alpha^{1-q} \beta^{\ell+q-1}}{(2-q)^{r+(\ell+q-1)/\alpha}}\, 
		\beta^\ell  \Gamma\bigg(r+{\ell+q-1\over\alpha}\bigg).
	\end{align*}
	
	\item
	Multiplying \eqref{1der} and \eqref{2der} with $n=1$, we obtain
	\begin{align*}
		_q E_{\alpha \beta}
		&= 
		_q\hspace{-0,1cm}\mathbb{E}
		\left[{\partial \log f(X;\boldsymbol{\theta})\over \partial {\alpha}}  \,
		{\partial \log f(X;\boldsymbol{\theta})\over \partial {\beta}}\, f^{1-q}(X;\boldsymbol{\theta})   \right]
		\\[0,2cm]
		&=
		\sum_{\ell\in\{0,\alpha\}} c_\ell \,
		\mathbb{E}\big[X^\ell f^{1-q}(X;\boldsymbol{\theta}) \big]
		+
		\sum_{\ell\in\{0,\alpha,2\alpha\}} d_\ell \,
		\mathbb{E}\big[X^\ell \log(X) f^{1-q}(X;\boldsymbol{\theta}) \big],
		%
		%
		%
		%-{\alpha\over\beta}\,
		%\bigg[{1\over\alpha}-\log(\beta)\bigg] \mathbb{E}\big[f^{1-q}(X;\boldsymbol{\theta}) \big]
		%-
		%{\alpha\over\beta}\, \mathbb{E}\big[\log(X) f^{1-q}(X;\boldsymbol{\theta}) \big]
		%\\[0,2cm]
		%&+
		%{\alpha\over\beta}\bigg(1+{1\over\beta^\alpha}\bigg)\, \mathbb{E}\big[X^{\alpha} \log(X) f^{1-q}(X;\boldsymbol{\theta}) \big]
		%+
		%{\alpha\over\beta^{\alpha+1}}\, \bigg[{1\over\alpha}-\log(\beta)\bigg]
		%\, \mathbb{E}\big[X^{\alpha} f^{1-q}(X;\boldsymbol{\theta}) \big]
		%\\[0,2cm]
		%&
		%-
		%{\alpha\over\beta^{2\alpha+1}}\, 
		%\, \mathbb{E}\big[X^{2\alpha} \log(X) f^{1-q}(X;\boldsymbol{\theta}) \big],
	\end{align*}
	where 
	$c_0=-{\alpha\over\beta}\,
	\big[{1\over\alpha}-\log(\beta)\big]$, 
	$c_\alpha={\alpha\over\beta^{\alpha+1}}\, \big[{1\over\alpha}-\log(\beta)\big]$, 
	$d_0=-{\alpha\over \beta}$, 
	$d_\alpha={\alpha\over\beta}\big(1+{1\over\beta^\alpha}\big)$ 
	and 
	$d_{2\alpha}=-{\alpha\over\beta^{2\alpha+1}}$.
	By using Items 7) and 10) of Subsection \ref{Properties}, we have
	\begin{align*}
		_q E_{\alpha \beta}
		&=
		\sum_{\ell\in\{0,\alpha\}} c_\ell \,
		\dfrac{\alpha^{1-q} \beta^{\ell+q-1}}{(2-q)^{r+(\ell+q-1)/\alpha}}\, 
		\beta^\ell  \Gamma\bigg(r+{\ell+q-1\over\alpha}\bigg)
		\\[0,2cm]
		&
		+
		\sum_{\ell\in\{0,\alpha,2\alpha\}} d_\ell \,
		{\beta^{\ell+q-1} \alpha^{1-q}\over (2-q)^{1+\zeta/\alpha}} \,
		\Gamma\bigg(1+\dfrac{\zeta}{\alpha}\bigg)
		\Bigg[\log\bigg({\beta\over \{2-q\}^{1/\alpha}}\bigg)
		+
		{1\over\alpha}\, 
		\Psi^{(0)}\bigg(1+\dfrac{\zeta}{\alpha}\bigg)
		\Bigg],
	\end{align*}
	where $\zeta=\ell+(1-q)(\alpha-1)$.
\end{itemize}
Note that arguments in $\log$ and $\Gamma$ functions should be positive.
%Note that the above expectations  $_q E_{\alpha \alpha}$, $_q E_{\beta \beta}$ and $_q E_{\alpha \beta}$ has a closed formula, because
%$\mathbb{E}\big[f^{1-q}(X;\boldsymbol{\theta}) \big]$ is obtained from Tsallis entropy (see Item 10 of Section \ref{Properties}), and the expected values
%$\mathbb{E}\big[X^{\alpha} f^{1-q}(X;\boldsymbol{\theta}) \big]$,
%$\mathbb{E}\big[\log(X) f^{1-q}(X;\boldsymbol{\theta}) \big]$,
%$\mathbb{E}\big[\log^2(X) f^{1-q}(X;\boldsymbol{\theta}) \big]$,
%$\mathbb{E}\big[X^{\alpha} \log(X) f^{1-q}(X;\boldsymbol{\theta}) \big]$ and $\mathbb{E}\big[X^{\alpha} \log^2(X) f^{1-q}(X;\boldsymbol{\theta}) \big]$  can be obtained by using Items 15), 16) and 17) of Section \ref{Properties}, with $r=2-q$.

%\textcolor{red}{I did q-Fisher by using Mathematica. The results of integrals from Mathematica are extremely too long. If you want, we can use the results of Mathematica, we will already give the numerical result and plot of Fisher. For this reason, it is not necessary to give mathematical expression of q-Fisher in Eq. \eqref{qFisher}. However, if you can do simple expressions for the elements of matrix in Eq. \eqref{qFisher}, you can write a simple expression. I am not sure whether or not the integral in Eq. \eqref{qFisher} can be calculated}

%%%=====================================================================
\section{Robustness}\label{section4rob}
Influence function is a measure which is used to evaluate the robustness of
M-estimators (\cite{Hampeletal86}). Let $\boldsymbol{\theta}$  be a vector for parameters $\alpha$ and
$\beta$, i.e. $\boldsymbol{\theta}=\left(\alpha,\beta\right)$, and  $f(x;\boldsymbol{\theta})$ represents p. d. function of Weibull
distribution. In this case, the objective function is defined as
\begin{align*}
	\rho(x;\boldsymbol{\theta})=\Lambda\big[f(x;\boldsymbol{\theta})\big],
	\quad 
	\psi_{\xi}\left(x\right)=
	\frac{\partial \rho(x;\boldsymbol{\theta}) }{\partial {\xi}}, 
	\quad \xi\in\{\alpha,\beta\},
	%\\[0,2cm]
	%\psi_{\beta}\left(x\right)&=
	%\frac{\partial \Lambda \big[f\left(x;\boldsymbol{\theta}\right)\big]}{ \partial {\beta}},
\end{align*}
where
$\Psi=(\psi_{\alpha},\psi_{\beta})$ is a vector of score functions $\psi_{\alpha}$ and $\psi_{\beta}$ and $\Lambda(z)=\log_q(z)$.

Robustness trusts on the finiteness of score functions from $ {\partial \rho(x;\boldsymbol{\theta})  \over \partial \boldsymbol{\theta}}$ when $x$ goes to infinity, i.e. $\lim_{x\to \infty} {\partial \rho(x;\boldsymbol{\theta})  \over \partial \boldsymbol{\theta}}$.  By using same way from robustness to outliers, the finiteness of score functions should be tested for the case in which we have $\lim_{x\to 0} {\partial \rho(x;\boldsymbol{\theta})  \over \partial \boldsymbol{\theta}}$. Thus, we will imply the robustness to inlier observations in a data set. 

\subsection{Examination of score functions of parameters in objective function $\log_q(f)$}
In order to get MLqE of parameters $\alpha$ and $\beta$, $\log_q$ is applied to $f(x;\boldsymbol{\theta})$. Thus, we have objective function $\rho(x;\boldsymbol{\theta})=\log_q\big[f(x;\boldsymbol{\theta})\big]$. The score functions derived by $\rho(x;\boldsymbol{\theta})$ for the corresponding parameters are given in the following order:

\begin{align}\label{psialphaq}
	\psi_{\alpha}\left(x\right)
	&=
	{\partial \rho(x;\boldsymbol{\theta}) \over \partial \alpha }
	=
	f^{1-q}(x;\boldsymbol{\theta}) \, 
	\Bigg\{\dfrac{1}{\alpha}+\Bigg[1-\bigg(\dfrac{x}{\beta}\bigg)^\alpha\Bigg]\log\bigg({x\over\beta}\bigg)\Bigg\},
	\\[0,3cm]
	\psi_{\beta}(x)
	&=
	{\partial \rho(x;\boldsymbol{\theta}) \over \partial \beta }
	=
	\dfrac{\alpha}{\beta}\,
	f^{1-q}(x;\boldsymbol{\theta}) \, 
	\Bigg[\bigg(\dfrac{x}{\beta}\bigg)^\alpha-1\Bigg]. \label{psibetaq}
\end{align}
%\textcolor{red}{I think that giving tables  are more descriptive, proofs should be added in the appendix}
A simple calculation shows that
{\scalefont{0,9}
	\begin{align*}
		\lim_{x \rightarrow 0} \psi_{\alpha}(x)=
		\begin{cases}
			+\infty & \text{for} \ \alpha\geq1 \ \text{and} \ q>1,
			\\
			0 & \text{for} \ \alpha \geq 1 \ \text{and} \ 0<q<1,
			\\
			0 & \text{for} \ 0<\alpha<1 \ \text{and} \ q>1,
			\\
			+\infty & \text{for} \ 0<\alpha<1 \ \text{and} \ 0<q<1,
		\end{cases}
		\quad 
		\lim_{x \rightarrow 0} \psi_{\beta}(x)=
		\begin{cases}
			-\infty & \text{for} \ \alpha\geq1 \ \text{and} \ q>1,
			\\
			0 & \text{for} \ \alpha \geq 1 \ \text{and} \ 0<q<1,
			\\
			0 & \text{for} \ 0<\alpha<1 \ \text{and} \ q>1,
			\\
			-\infty & \text{for} \ 0<\alpha<1 \ \text{and} \ 0<q<1,
		\end{cases}
	\end{align*}
	\begin{align*}
		\lim_{x \rightarrow +\infty} \psi_{\alpha}(x)=
		\begin{cases}
			-\infty & \text{for} \ \alpha\geq1 \ \text{and} \ q>1,
			\\
			0 & \text{for} \ \alpha \geq 1 \ \text{and} \ 0<q<1,
			\\
			-\infty & \text{for} \ 0<\alpha<1 \ \text{and} \ q>1,
			\\
			0 & \text{for} \ 0<\alpha<1 \ \text{and} \ 0<q<1,
		\end{cases}
		\quad 
		\lim_{x \rightarrow +\infty} \psi_{\beta}(x)=
		\begin{cases}
			+\infty & \text{for} \ \alpha\geq1 \ \text{and} \ q>1,
			\\
			0 & \text{for} \ \alpha\geq1 \ \text{and} \ 0<q<1,
			\\
			+\infty & \text{for} \ 0<\alpha<1 \ \text{and} \ q>1,
			\\
			0 & \text{for} \ 0<\alpha<1 \ \text{and} \ 0<q<1,
		\end{cases}
	\end{align*}
}

%\begin{center}
%\textcolor{red}{	
%	$\lim_{x \rightarrow 0} \psi_{\alpha}=?,$
%	$\lim_{x \rightarrow 0} \psi_{\beta}=?,$}
%\end{center}
%	\textcolor{red}{
%\begin{center}
%		$\lim_{x \rightarrow \infty} \psi_{\alpha}=?,$
%		$\lim_{x \rightarrow  \infty} \psi_{\beta}=?,$
%\end{center}
%}
%\textcolor{red}{I cannot make sure to write my proofs in word for ${\partial \over \partial \alpha}\rho(x;\alpha,\beta) $ and for beta, because my explanation is more complicated maybe reviewer cannot understand what text tells about. If you catch a simple proof for these limits, it will be good}

Briefly, the limits above can be written using the following tables:
\begin{table}[htb!]
	\centering 
	\caption{Limit values of vector $\Psi(x)=(\psi_{\alpha}(x),\psi_{\beta}(x))$} \vspace*{0,1cm}
	\begin{tabular}{c|cc} \hline
		$x\to 0$	& $q >1$ & 		$0<q<1$ \\ \hline 
		$	\alpha\geq1$		& 	$(+\infty,-\infty)$ & $(0,0)$ \\
		$	0<\alpha<1$	& $(0,0)$ & $(+\infty,-\infty)$
	\end{tabular}
	\label{limitvalueszero} 
\end{table}

%\textcolor{red}{Please show limit values of Table \ref{limitvalueszero}}
\begin{table}[htb!]
	\centering 
	\caption{Limit values of a vector $\Psi(x)=(\psi_{\alpha}(x),\psi_{\beta}(x))$} \vspace*{0,1cm}
	\begin{tabular}{c|cc} \hline
		$x\to +\infty$	& $q >1$ & 		$0<q<1$ \\  \hline
		$	\alpha\geq1$		& 	$(-\infty,\infty)$ & $(0,0)$ \\
		$	0<\alpha<1$	& $(-\infty,\infty)$ &$(0,0)$
	\end{tabular}
	\label{limitvaluesinfty}
\end{table}

$\Psi=(\psi_{\alpha},\psi_{\beta}) $ is a vector of score functions $\psi_{\alpha}$ and $\psi_{\beta}$. Influence function is proportional to score function. If score functions $\psi_{\alpha}$ and $\psi_{\beta}$ are finite, then we have finite influence function, which shows that the estimators $\hat{\alpha}$ and $\hat{\beta}$ from MLqE and MLE will be robust (\cite{Hampeletal86}). When  MLqE is used, the influence function of the estimators $\hat{\alpha}$ and $\hat{\beta}$ is finite if $q \in (0,1)$ (see Table \ref{limitvaluesinfty}). However, one can show that  there are cases for $\log(f)$ in which the score functions of parameters of $f$ are finite or infinite according to property of $f$ (\cite{Canetal19}). In addition, for an arbitrary $f$, we cannot get  FI and the definition of influence function also includes the inverse of Fisher. There can be cases in which an element of FI matrix is not defined for some values of parameters such as $\Gamma(r),r>0$, $a/b, b \neq 0$, etc and the inverse of FI matrix cannot exist. However, we can get estimates of parameters for these cases in which FI and its inverse does not exist. Instead of calling robustness and trusting on the tools in robustness (M. Thompson,  e-mail communication, March 2, 2016), the main approach should be based on the modelling competence of an arbitrary function.  Further, Table \ref{limitvalueszero} includes the infinity cases.  The modeling is carried out for the finite sample size. The robustness is for the case in which we take limit at the values which are zero and infinity. However, it is expected that the robustness should be supported by simulation. There is an open question: Even if the score functions are infinite for limit values at zero and infinity, can we perform a modeling capability on a finite sample size? Yes and  we can find a counter example against the robustness theory from simulation results (see Case 4 in Table \ref{caseswut}).

\section{Optimization and numerical experiments}\label{optimizationnumerical}

\subsection{Optimization via genetic algorithm for $\Lambda(f)$} 
The genetic algorithm (GA) can be applied to solve a variety of optimization problems that are unsuitable for standard optimization algorithms which include problems where the objective function is undifferentiated in Radon Nikodym derivative, highly nonlinear, discontinuous, (absolutely) continuous, non-smooth, even stochastic and random subsets from the real line. GA method is preferred to ensure that such objective functions converge to a global point. When $\log_q$ is used as the objective function, GA method has been used by (\cite{CanKor18}). The  codes  used to get estimates of MLqE are given by Appendix \ref{comptGA}.

\subsection{Structure of contamination: The distributions and their parameter values used at Monte Carlo simulation}
Contamination makes a disorder in the identically distributed random variables represented by $X_1,X_2,\dots ,X_n$. If these random variables are disrtibuted non-identically, then we express the structure of non-identicality from mixing of two p.d. functions $f_0$ and $f_1$, as given by following form:
\begin{equation}
	f_{\varepsilon}(x;\boldsymbol{\theta};\boldsymbol{\tau}) = (1-\varepsilon)f_0(x;\boldsymbol{\theta}) + \varepsilon f_1(x;\boldsymbol{\tau})
\end{equation}
\noindent is the contaminated distribution. The constant $\varepsilon$ is the contamination rate. $f_0$ is the underlying and $f_1$ is contamination. $f_1$ can be same distribution with $f_0$, but the parameter values of $f_1$ are different from $f_0$, i.e. $f_0=$ Weibull($\alpha_0$,$\beta_0$) and $f_1=$ Weibull ($\alpha_1$,$\beta_1$). We can select $f_1$ distribution with the given values of parameters. For example, $f_1$  = BurrIII($\alpha_1$,$\beta_1$). In the real world, after assuming that a data set is a member of $f_0$ distribution, a contamination to $f_0$ by means of $f_1$ can occur in an empirical distribution. We do not know how much rate $\varepsilon$ of contamination into the data set exist. Further, there are two types of contamination. These are inliers and outliers. A data set can include both of contaminations. The aim is to estimate robustly parameters $\boldsymbol{\theta}_1,\boldsymbol{\theta}_2,\dots,\boldsymbol{\theta}_p$ of $f_0$ under contamination. 

The distributions and parameter values selected to generate different forms of contamination are those in which both inliers and    outliers are at the same time. The plots of functions of the selected values of parameters are given in Figures \ref{wwcase}-\ref{wbcase}.  Thus, it can be observed the structure of underlying and contamination distributions. We provide three blocks for Figures \ref{wwcase}-\ref{wbcase}. Each of plots in the blocked Figures \ref{wwcase}-\ref{wbcase} are given to depict the role of underlying and contamination distributions clearly. They are depicted separately and same length $x$-axis to avoid bad illustration which can occur the different values of parameters for values of p.d. function at $y$-axis. Thus, we observe the contamination structure and so that we can test the performance of MLqE for such contaminations in all cases given by Figures \ref{wwcase}-\ref{wbcase}. Blue and red lines representing p. d. function abbverivated as PDF show the underlying $f_0$ and the contamination $f_1$ into underlying distribution respectively. Simulation in section \ref{simusect} includes the robust estimations of parameters of blue lines. As it is proven theoretically by \ref{convexconcavonemode}, Weibull distribution has one unimodal for $\alpha>1$. Note that the mode of Weibull does not exist and makes asymptotic to $y$-axis for $\alpha \leq 1$. The uniform and BurrIII \cite{Canetal19} distributions have one mode as well. 

For the estimation process, the modality of p.d. function $f(x;\boldsymbol{\theta})$ and the concavity of $\Lambda$ are taken in account to get $\Lambda(f)$ and apply for the estimation. Thus, the cooperation between $f(x;\boldsymbol{\theta})$ and $\Lambda$ for conducting an accurate modeling on a data set can be performed successfully. For example, the smoothness property of objective function $\rho(x;\boldsymbol{\theta})=\Lambda\big[f(x;\boldsymbol{\theta})\big]$ is important to apply for estimation if a data set does not have several jumpings on an interval on the real line, i.e. the smoothness of data set should be provided. The structure of inliers in Monte Carlo simulation will not be strictly existing. In other words, the frequencies of artificial data sets for a narrow interval on the real line are not extremely high degree. One can observe the schema of underlying and contamination distributions in Figures \ref{wwcase}-\ref{wbcase}.

\begin{figure}[htbp]
	\centering
	\subfigure[Case 1]{\label{fig:Case1}\includegraphics[width=0.24\textwidth]{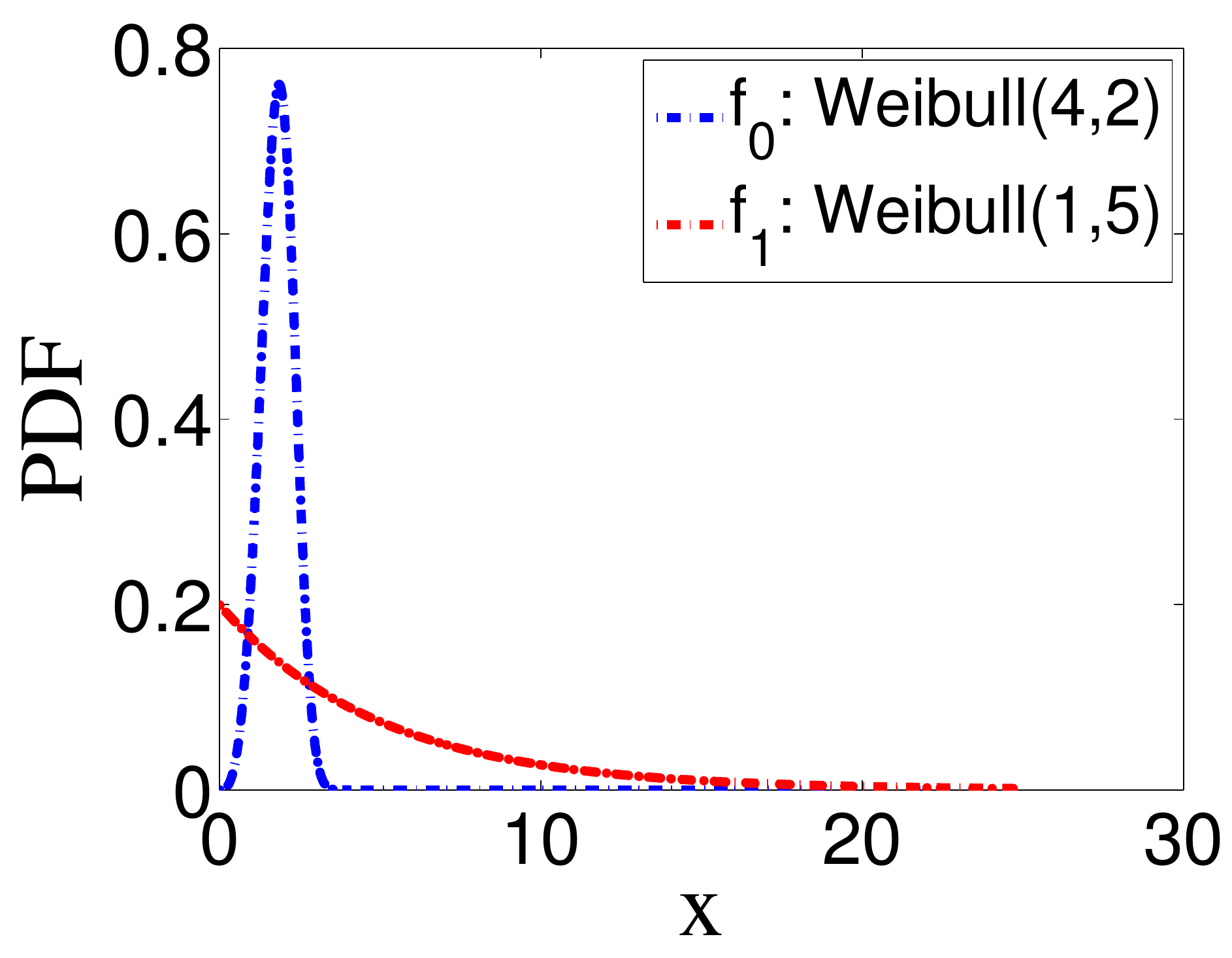}}
	\subfigure[Case 2]{\label{fig:Case2}\includegraphics[width=0.24\textwidth]{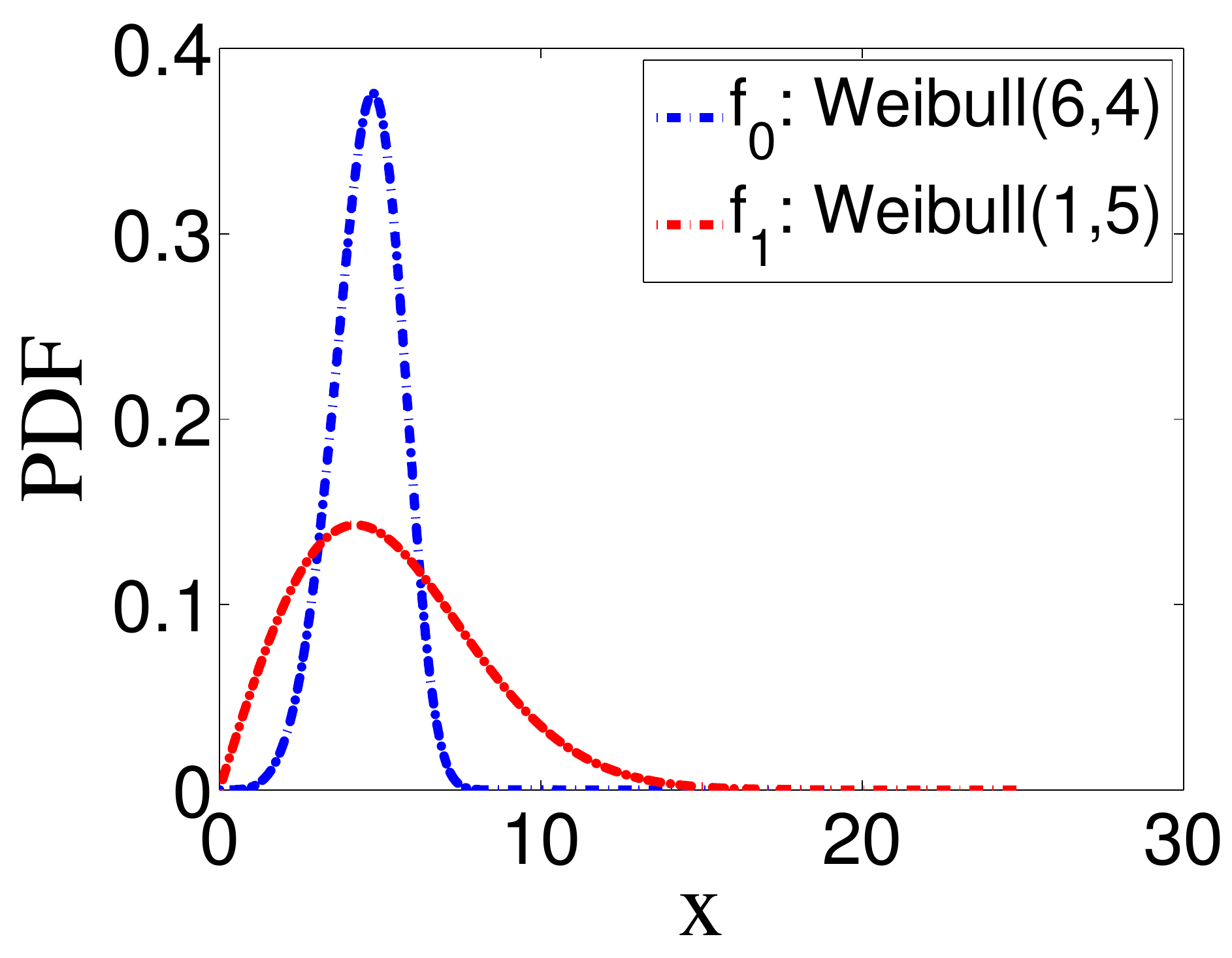}}
	\subfigure[Case 3]{\label{fig:Case3}\includegraphics[width=0.24\textwidth]{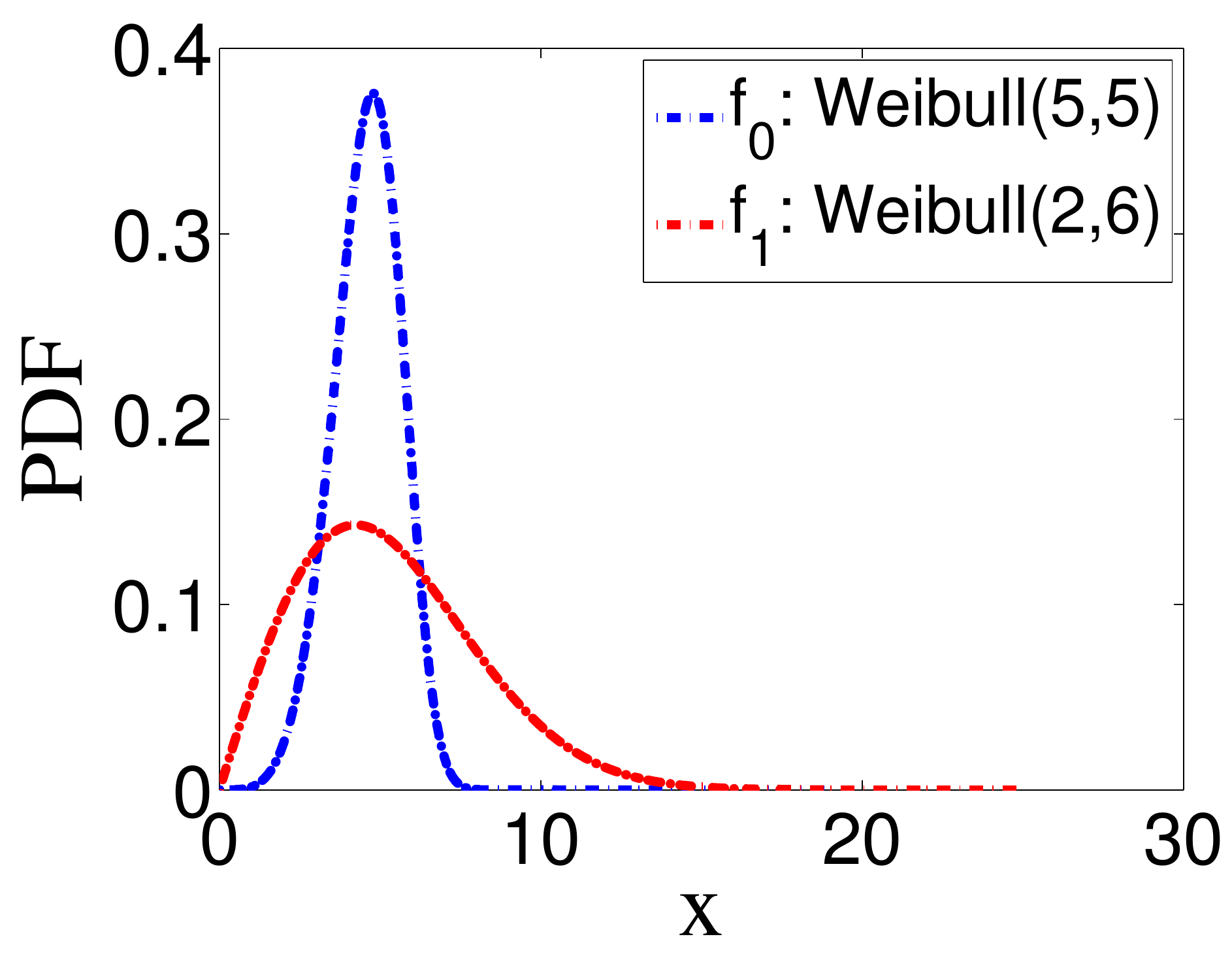}}
	\subfigure[Case 4]{\label{fig:Case4}\includegraphics[width=0.24\textwidth]{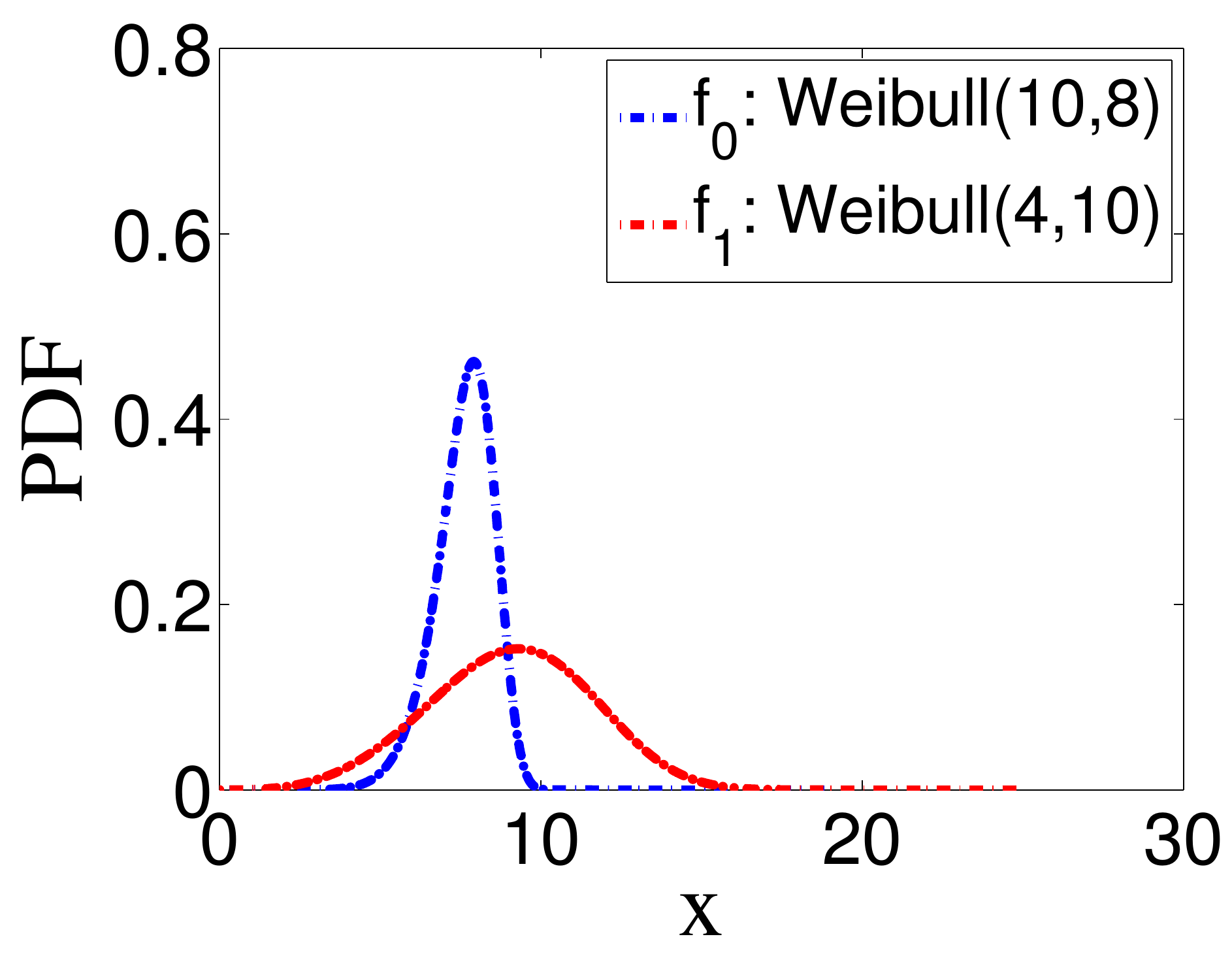}}
	\subfigure[Case 5]{\label{fig:Case5}\includegraphics[width=0.24\textwidth]{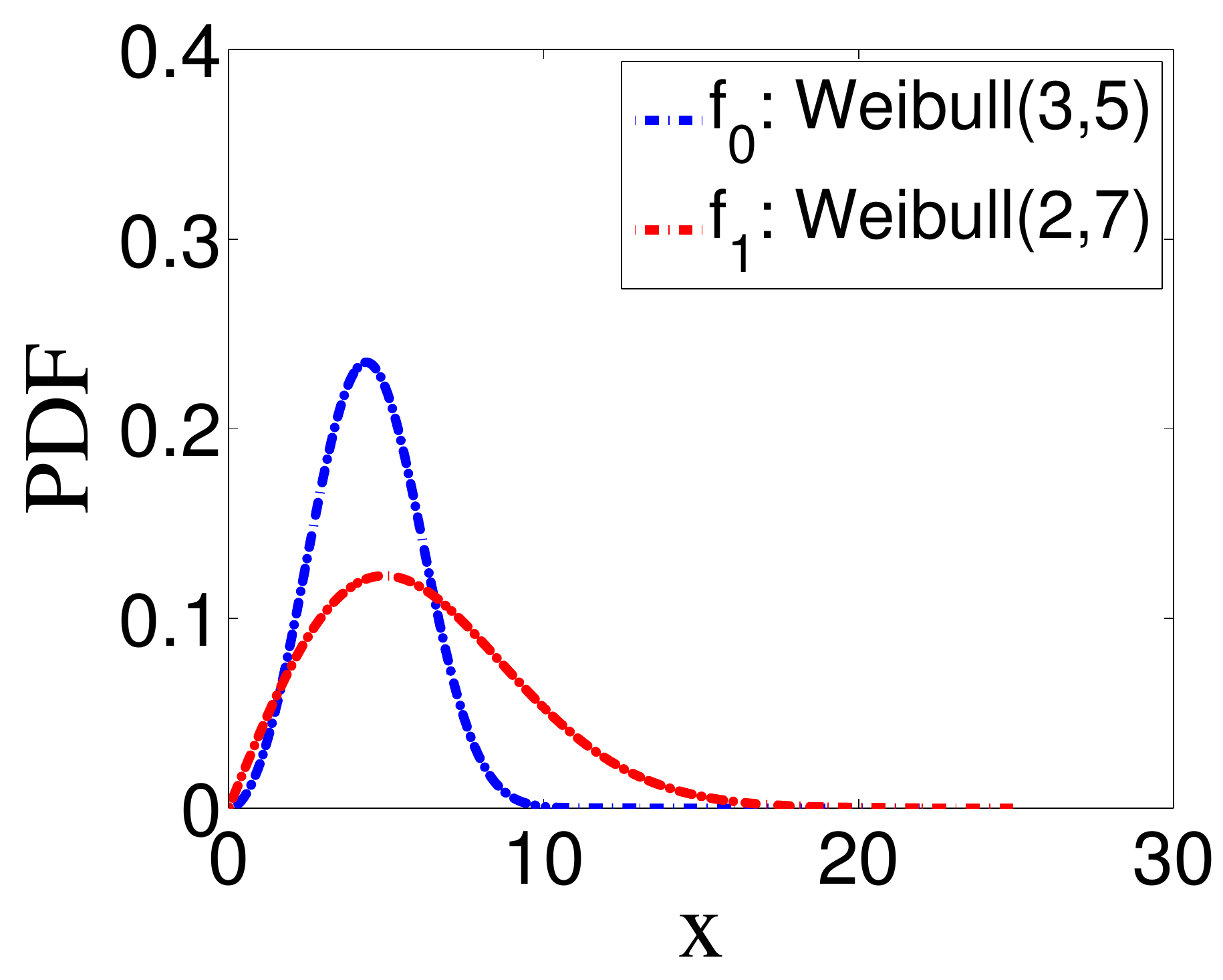}}
	\subfigure[Case 6]{\label{fig:Case6}\includegraphics[width=0.24\textwidth]{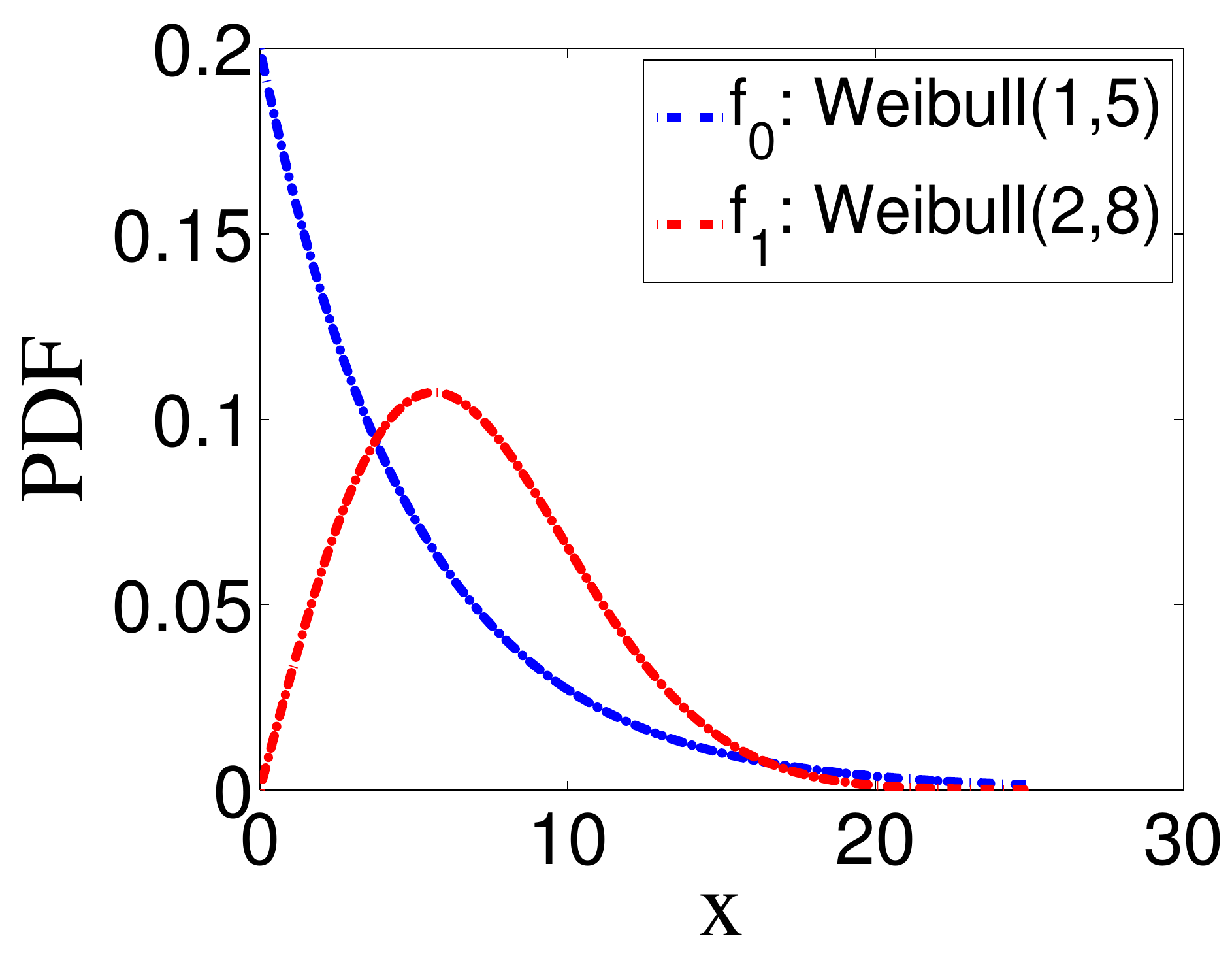}}
	\subfigure[Case 7]{\label{fig:Case7}\includegraphics[width=0.24\textwidth]{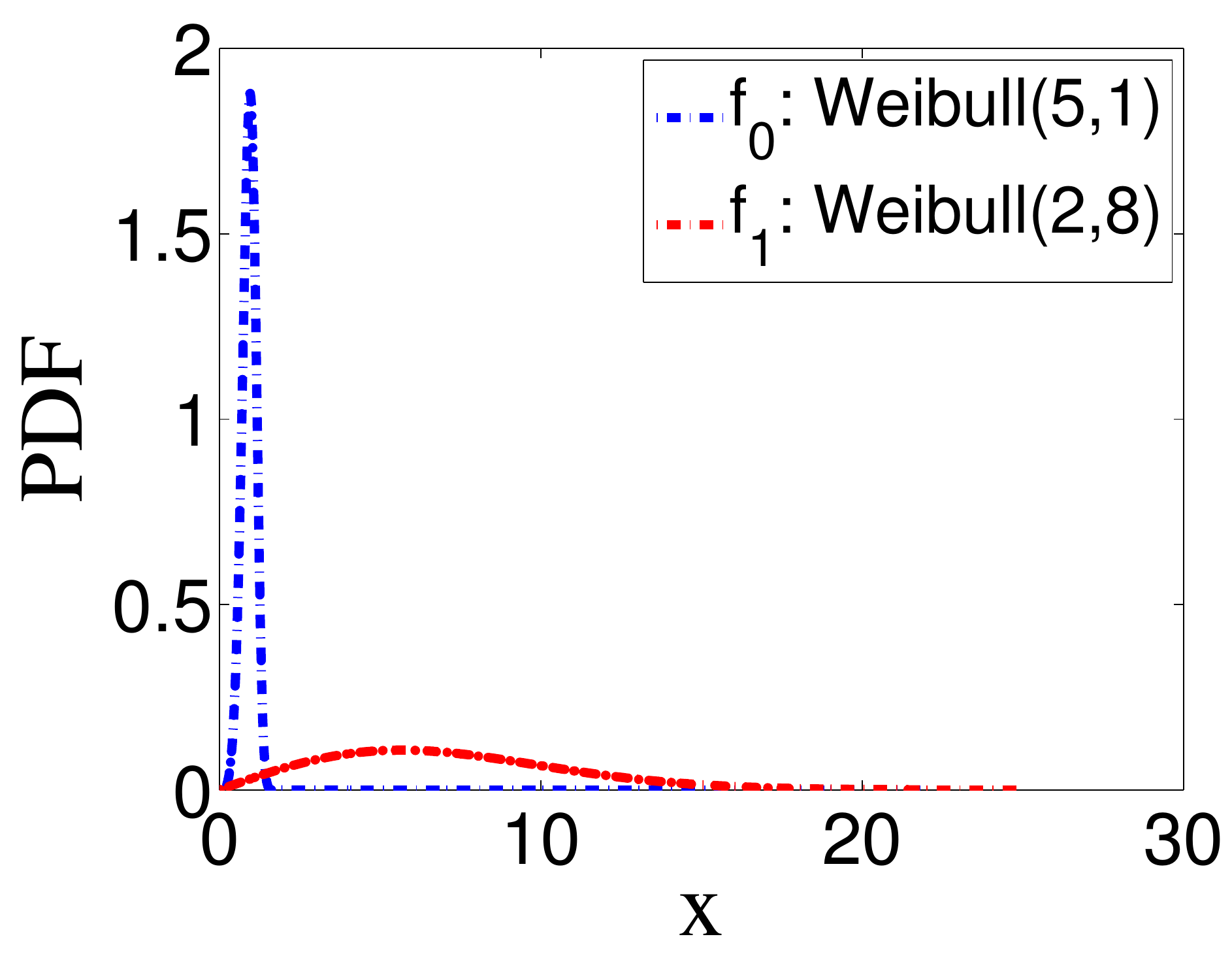}}
	\subfigure[Case 8]{\label{fig:Case8}\includegraphics[width=0.24\textwidth]{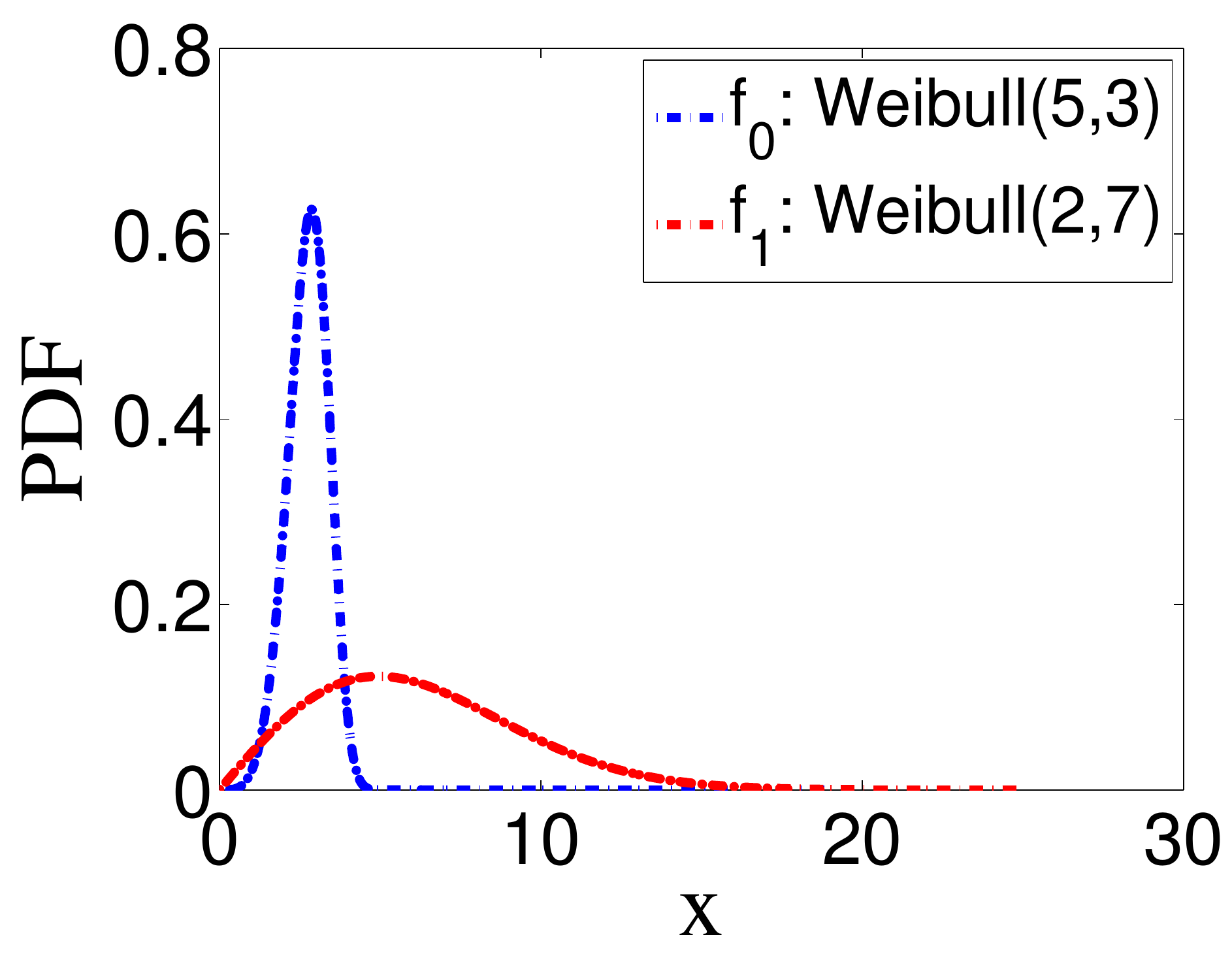}}
	%    \subfigure[Case 9]{\label{fig:Case9}\includegraphics[width=0.25\textwidth]{plotw810w410}}
	%    \subfigure[Case 10]{\label{fig:Case10}\includegraphics[width=0.25\textwidth]{plotw710w58}}
	\caption{The underlying and contamination distributions are Weibull (color online)}
	\label{wwcase}
\end{figure}
\begin{figure}[htbp]
	\centering
	\subfigure[Case 1]{\label{fig:Case1u}\includegraphics[width=0.24\textwidth]{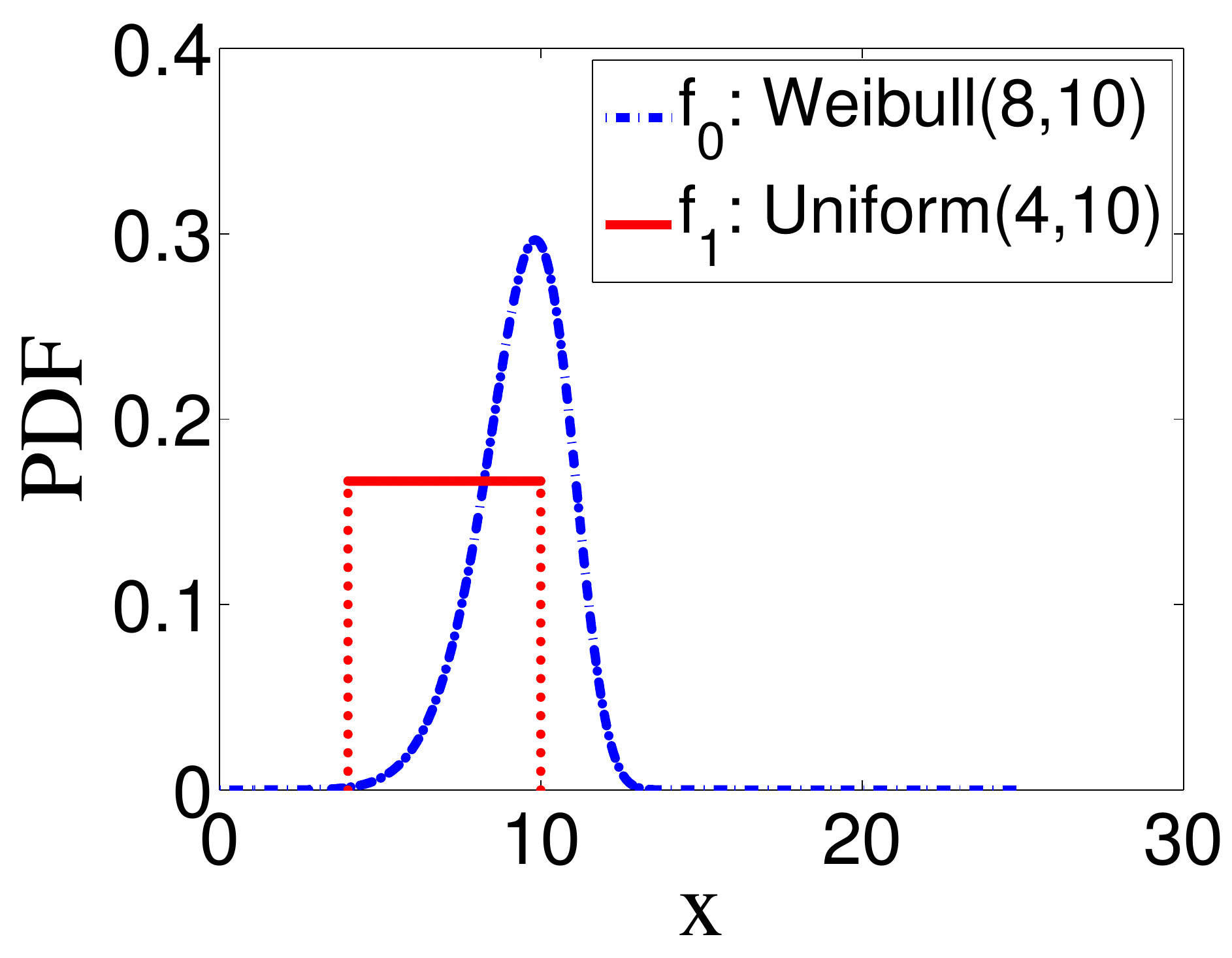}}
	\subfigure[Case 2]{\label{fig:Case2u}\includegraphics[width=0.24\textwidth]{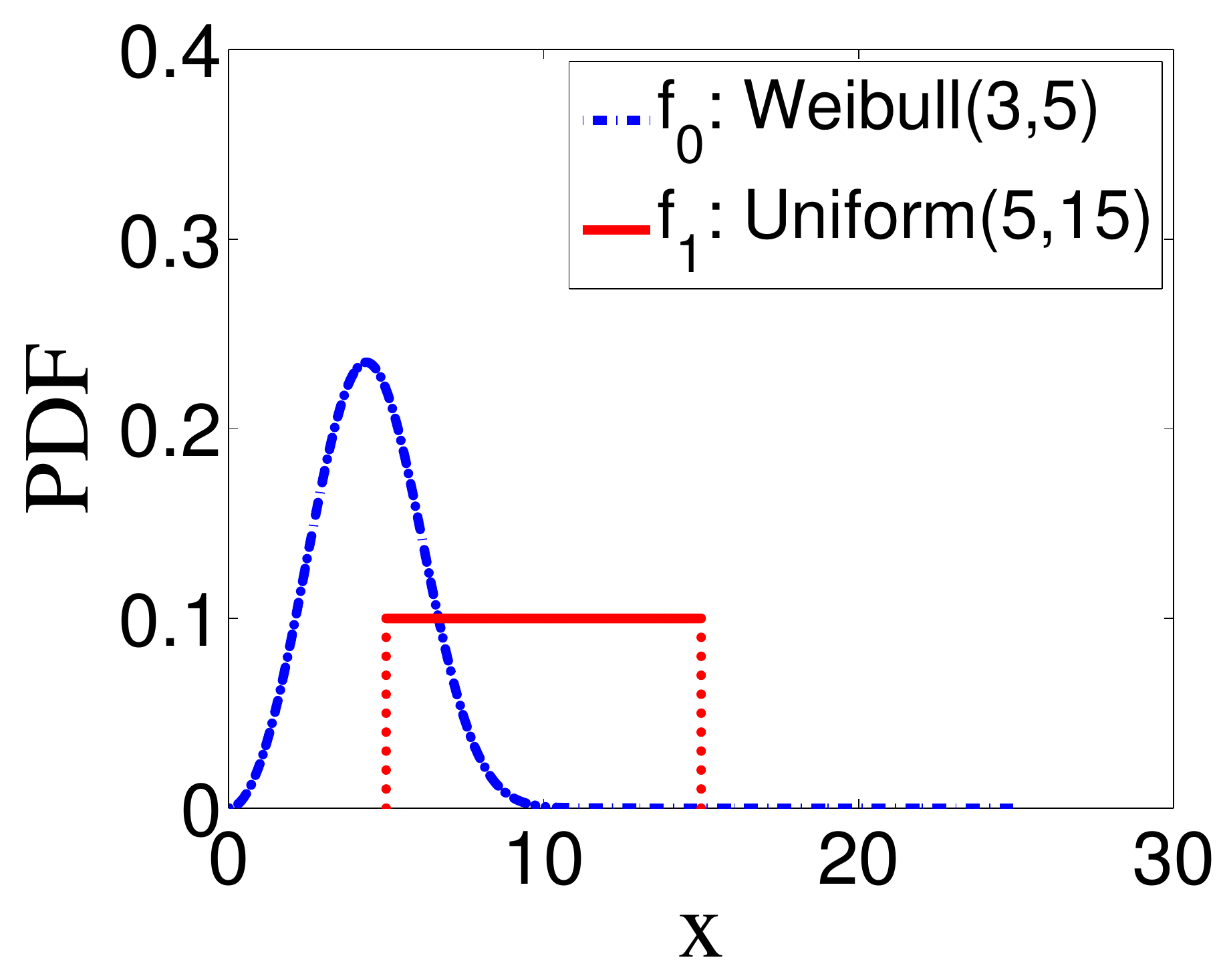}}
	\subfigure[Case 3]{\label{fig:Case3u}\includegraphics[width=0.24\textwidth]{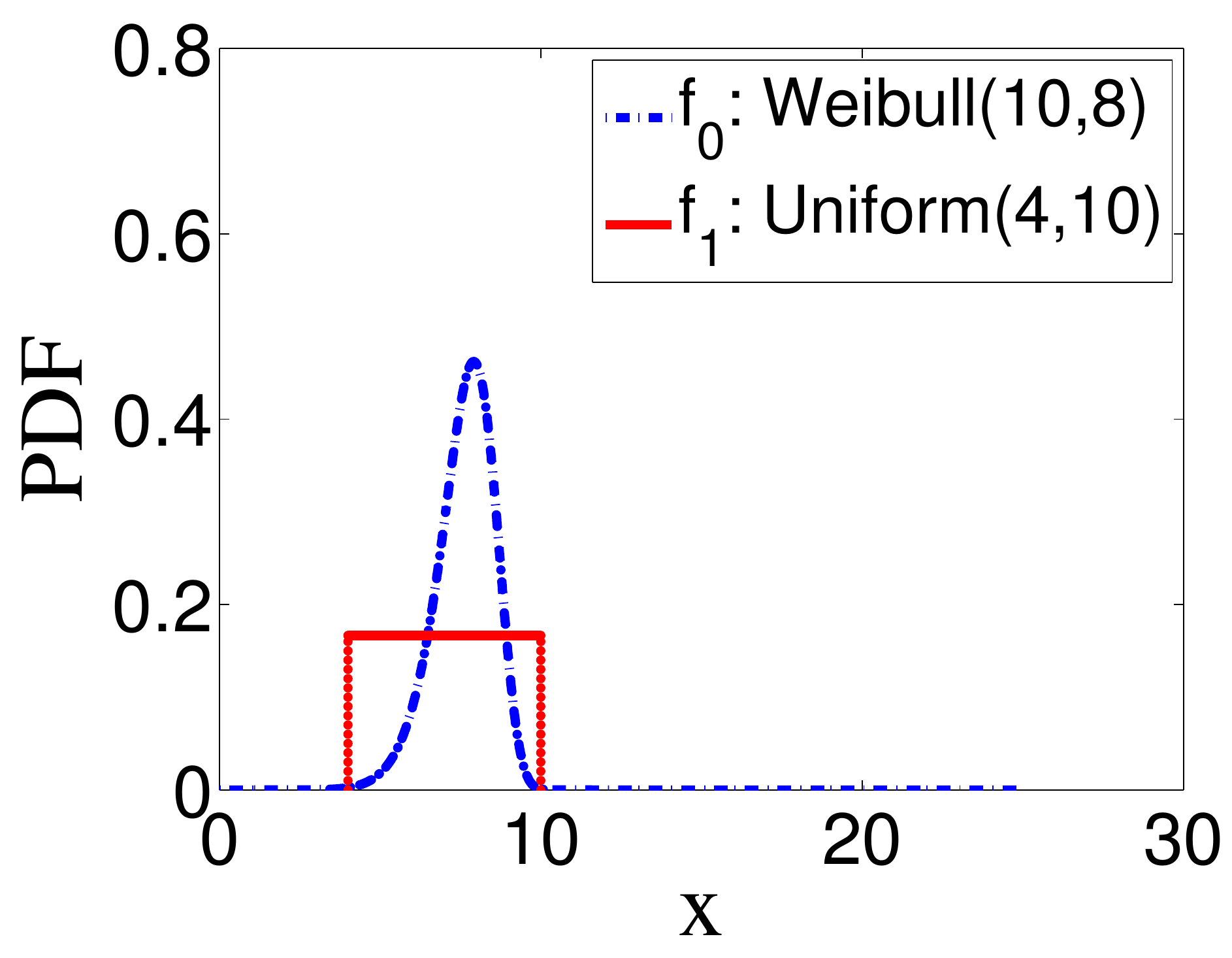}}
	\subfigure[Case 4]{\label{fig:Case4u}\includegraphics[width=0.24\textwidth]{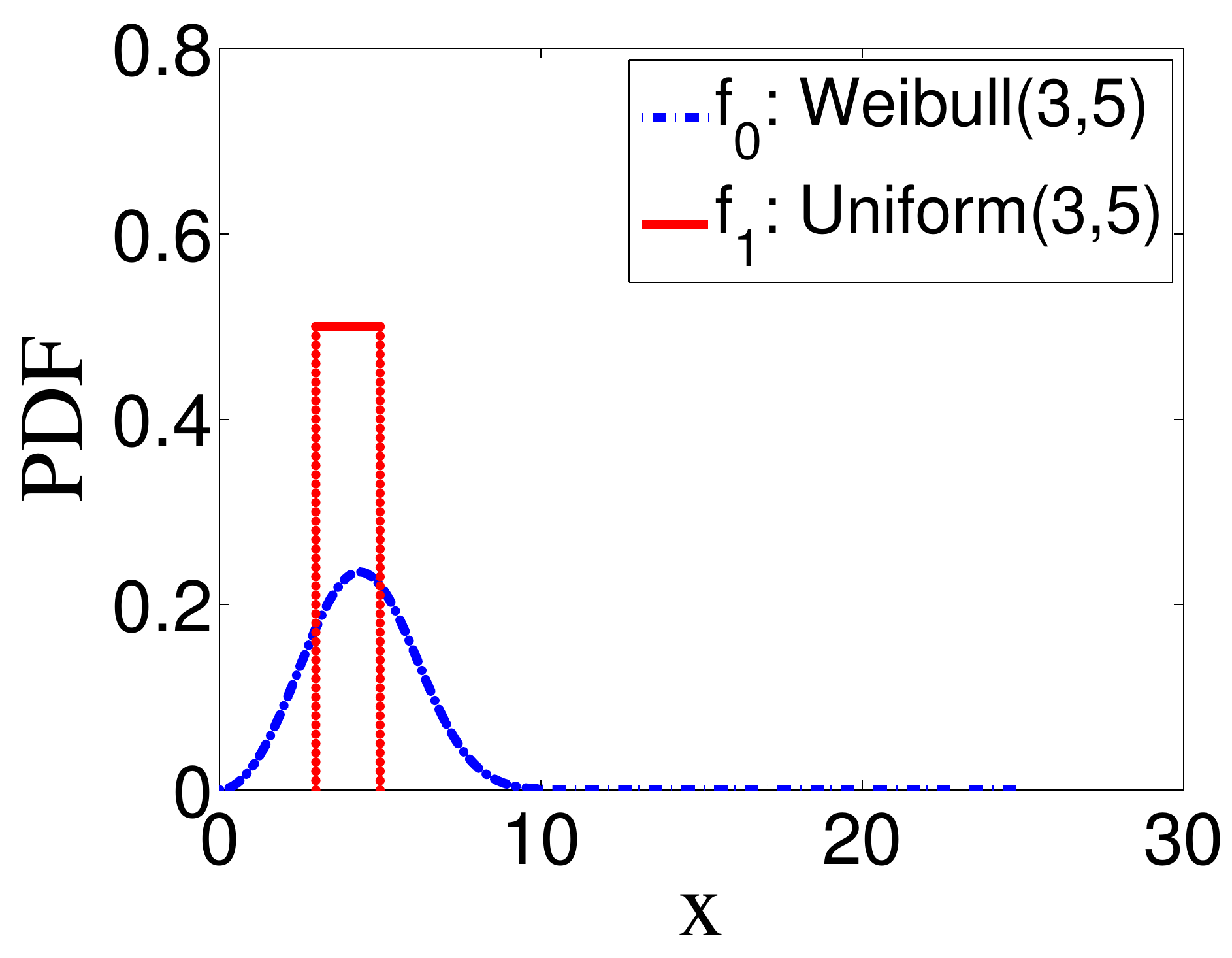}}
	\caption{The underlying is Weibull and contamination is Uniform distributions (color online)}
	\label{wucase}
\end{figure}
\begin{figure}[htbp]
	\centering
	\subfigure[Case 1]{\label{fig:Case1b}\includegraphics[width=0.24\textwidth]{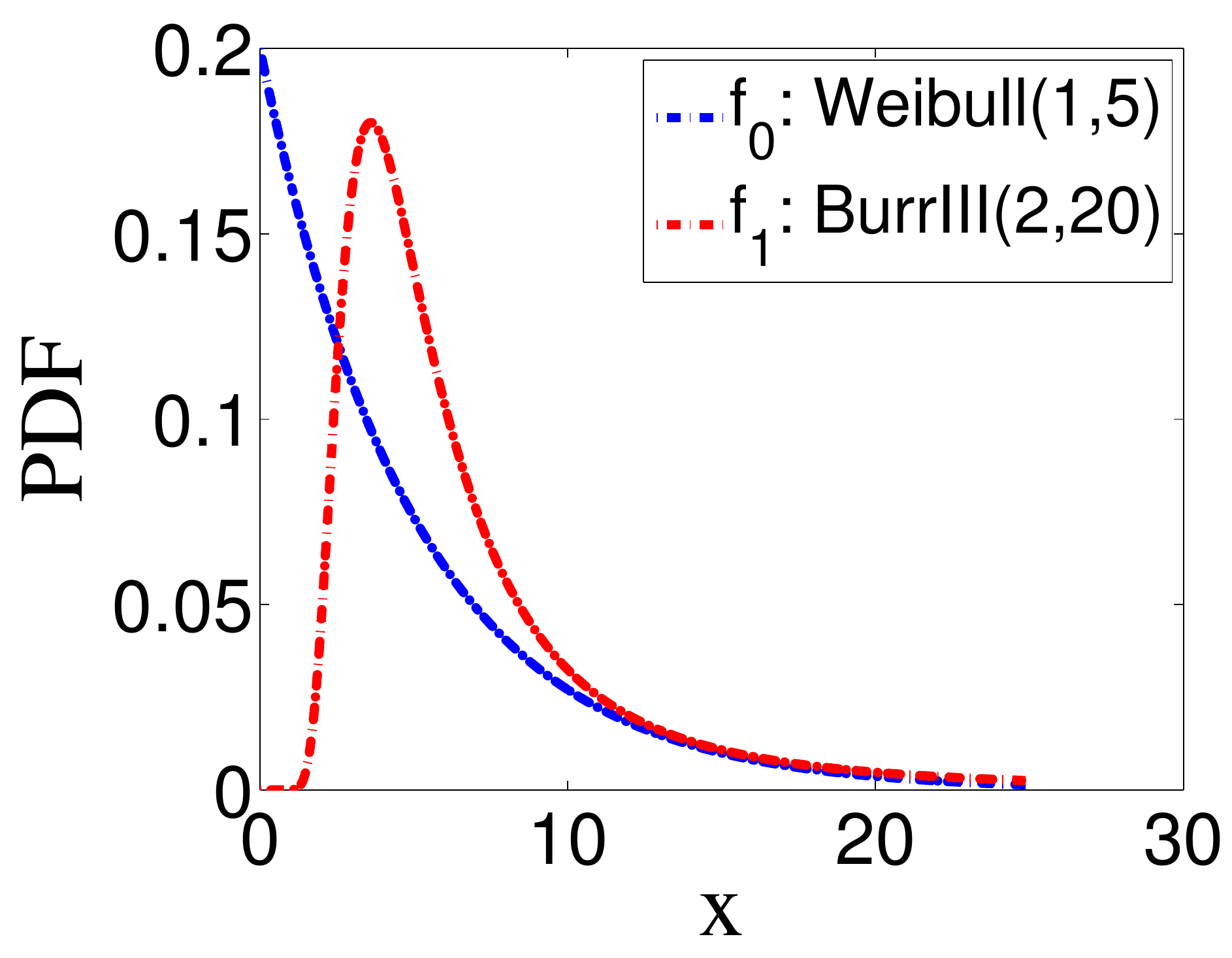}}
	\subfigure[Case 2]{\label{fig:Case2b}\includegraphics[width=0.24\textwidth]{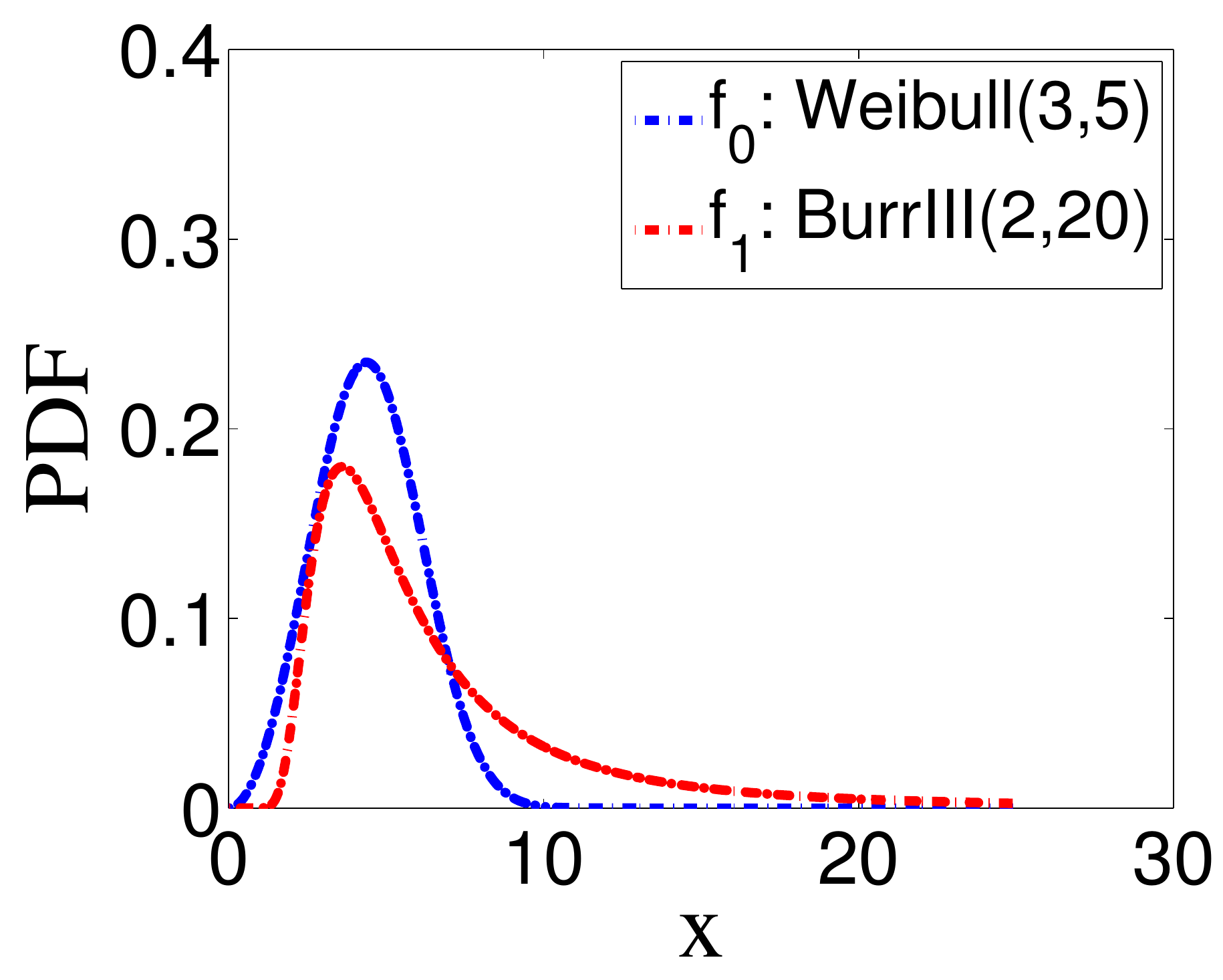}}
	\subfigure[Case 3]{\label{fig:Case3b}\includegraphics[width=0.24\textwidth]{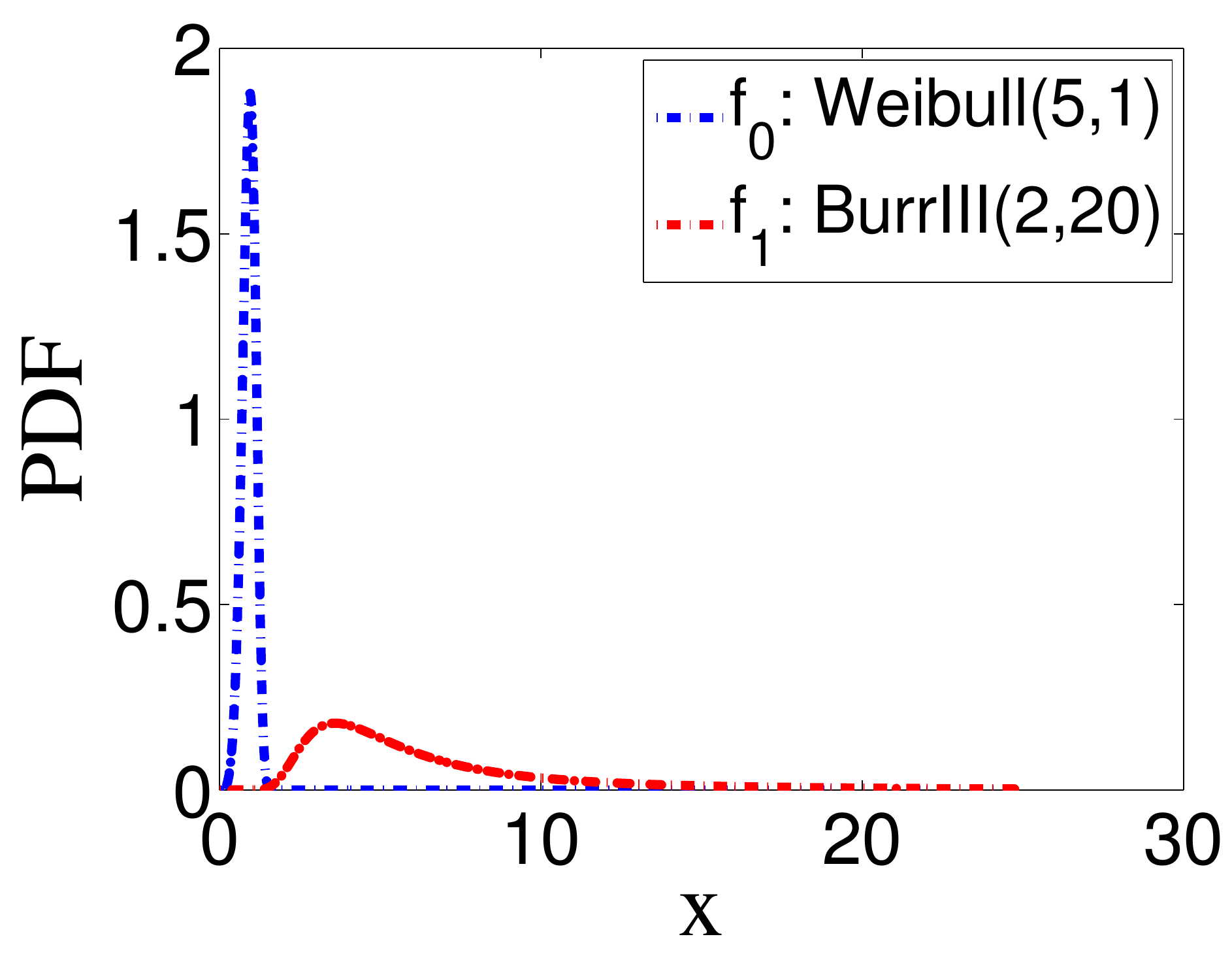}}
	\subfigure[Case 4]{\label{fig:Case4b}\includegraphics[width=0.24\textwidth]{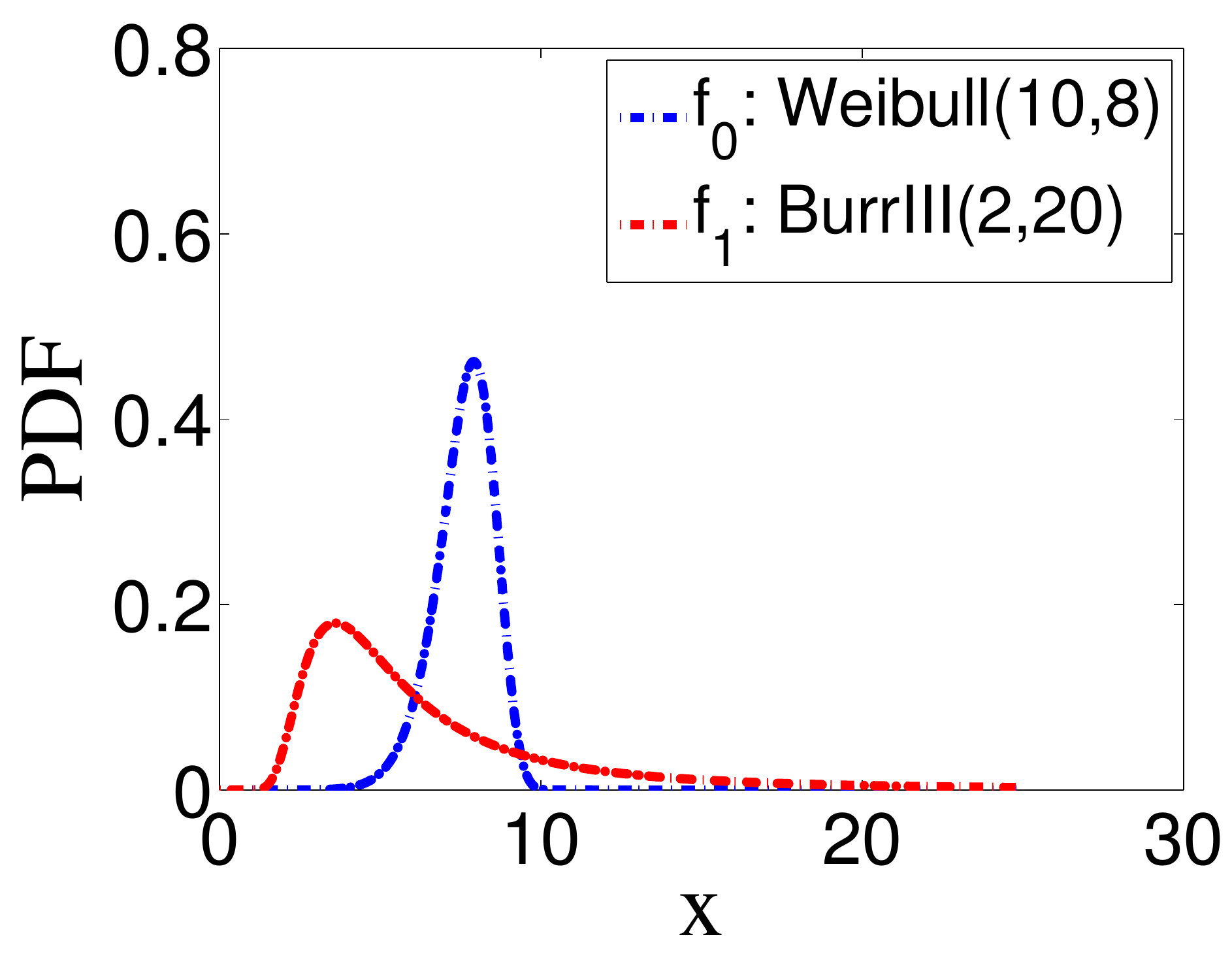}}
	\caption{The underlying is Weibull and contamination is BurrIII distributions (color online)}
	\label{wbcase}
\end{figure}
\subsection{Monte Carlo simulation: The design of artificial random numbers and results of simulation}\label{simusect}

A comprehensive simulation study is performed to make a comparison between MLqE and MLE methods. There are three contamination structures in the simulation.  Contamination makes a disorder in the identically distributed random observations represented by $x_1,x_2,\ldots ,x_n$. Contaminations are as follows:
\begin{enumerate}
	\item $(1-\varepsilon)f_0(x;\alpha_0,\beta_0) + \varepsilon f_1(x;\alpha_1,\beta_1) =(1-\varepsilon)$Weibull($x;\alpha_0$,$\beta_0$)+$\varepsilon$Weibull($x;\alpha_1$,$\beta_1$),
	\item $(1-\varepsilon)f_0(x;\alpha_0,\beta_0) + \varepsilon f_1(x;\alpha_1,\beta_1) =(1-\varepsilon)$Weibull($x;\alpha_0$,$\beta_0$)+$\varepsilon$Uniform($x;\alpha_1$,$\beta_1$),
	\item $(1-\varepsilon)f_0(x;\alpha_0,\beta_0) + \varepsilon f_1(x;\alpha_1,\beta_1) =(1-\varepsilon)$Weibull($x;\alpha_0$,$\beta_0$)+$\varepsilon$BurrIII($x;\alpha_1$,$\beta_1$).
\end{enumerate}

The contamination rate $\varepsilon$ is chosen as 10\% and 20\%. The parameters $\alpha_0$ and $\beta_0$ of underlying distribution Weibull($\alpha_0$,$\beta_0$) are estimated under contamination. $f_1$ can be chosen as Weibull($\alpha_1,\beta_1$), single-peak BurrIII($\alpha_1,\beta_1$) \cite{Canetal19} and Uniform($\alpha_1,\beta_1$) with the same mode for all values of $x \in [a,b]$. Thus, inliers and outliers can be observed through shape of function (see Figures \ref{wwcase}-\ref{wbcase}). Four different sample sizes $n=50,100,150$ and $n=200$ are used in the designs of simulation. $n_0$ and $n_1$ represent the number of random numbers from $f_0$ and $f_1$ distributions respectively. The number of replication for $n$ is $10^{4}$. 

Simulation variance $\widehat{Var}(\hat{\boldsymbol{\theta}})$ and simulation mean error squares $\widehat{MSE}(\hat{\boldsymbol{\theta}})$ are computed by using the following forms:
\begin{eqnarray}
	\widehat{MSE}(\hat{\boldsymbol{\theta}})&=&\widehat{Var}(\hat{\boldsymbol{\theta}})+\widehat{Bias}(\hat{\boldsymbol{\theta}})^2 \\ \nonumber
	\widehat{MSE}(\hat{\boldsymbol{\theta}})&=&\hat{E}[(\hat{\boldsymbol{\theta}}-\hat{E}(\hat{\boldsymbol{\theta}}))^2]+(\hat{E}(\hat{\boldsymbol{\theta}})-\boldsymbol{\theta})^2
\end{eqnarray}
$\widehat{Bias}^2$ and $\widehat{E}$ show the bias and the expected values obtained from the simulation, respectively. $\widehat{MSE}$ and $\widehat{Var}$ obtained through simulation are sampling forms of theoretical MSE and Var, respectively. If the estimators  $\hat{\boldsymbol{\theta}}$ obtained by MLqE are unbiased, then MSE($\hat{\boldsymbol{\theta}}$)=Var($\hat{\boldsymbol{\theta}}$). For comparison between MLqE and MLE methods, Tables \ref{caseswt}-\ref{caseswbtg} give the estimates of parameters $\alpha_0$ and $\beta_0$, $\widehat{Var}(\hat{\boldsymbol{\theta}}$)  and $\widehat{MSE}(\hat{\boldsymbol{\theta}}$). Note that other objective functions from $\log_{\kappa}$ and DPD did not give results which are better than MLqE for different types of contaminations. For this reason, we do not give them for the sake of not increasing the page numbers (see Figures \ref{fig:logq1}-\ref{fig:logDPD} for objective functions). 
\input{table3wp}
%%%%%%%%%%%%%%%%BREAKKBERAKKK
\input{table4wp}
\input{table5wp}
\input{table6wp}

When the Tables \ref{caseswt}-\ref{caseswbtg} are examined;
\begin{itemize}
	\item It is generally observed that the modeling capability of $\log_q$($f_0$=Weibull($\alpha_0,\beta_0$)) from MLqE is better than that of $\log$($f_0$=Weibull($\alpha_0,\beta_0$)) from MLE.
	\item In MLqE, it is observed especially that the estimated values of shape parameter $\alpha$ can be significantly small for some tried designs.
	\item  As it is logical to expect, $\widehat{MSE}(\hat{\boldsymbol{\theta}})$ of MLE cannot get the smaller values when the sample size $n$ is increased for some designs of contamination; because $n=n_0+n_1$  increases sample size $n_1$ of a contamination, which makes more contaminated data set when it is compared with small sample sizes, such as $n=50,100,$ etc.  
	In addition, $\log(f)$ cannot model well. However, for some designs of contamination, $\widehat{MSE}(\hat{\boldsymbol{\theta}})$ of $\log(f)$ can be smaller than that of $\log_q(f)$, because the used objective function is an important indicator for modeling capability. For example, $(1-\varepsilon)$Weibull($\alpha_0=1$,$\beta_0=5$) + $\varepsilon$ Weibull($\alpha_1=2$,$\beta_1=8$).
	\item  $\widehat{MSE}(\hat{\boldsymbol{\theta}})$ were used to compare the performance of MLqE and MLE. Depending on the structure of contamination, the unbiasedness of estimators $\hat{\boldsymbol{\theta}}$ obtained by MLqE were examined by simulation. However, the contamination structure, i.e. the selected $f_1$ distribution and its parameter values and contamination rate $\varepsilon$ can lead to observe the biasedness in Tables \ref{caseswt}-\ref{caseswbtg}. In addition, in the some cases of MLqE, $\widehat{MSE}(\hat{{\alpha}}) \approx \widehat{Var}(\hat{{\alpha}})$ and $\widehat{MSE}(\hat{{\beta}}) \approx \widehat{Var}(\hat{{\beta}})$ for the shape $\alpha$ and the scale $\beta$ parameters, respectively. 
	\item When the sample size is increased, the values of $\widehat{MSE}(\hat{{\alpha}})$ and $\widehat{MSE}(\hat{{\beta}})$ for MLqE decreased, as expected. However, this is not observed by the results from MLE. The structure of inliers and outliers affected partially the estimation of shape parameter $\alpha$.
	\item 	For some contamination and the values of parameters $\alpha_0$ and $\beta_0$ for Weibull, the values of $q$ can differ for different sample sizes. In addition, $q$ differs for different contamination structure and Weibull($\alpha_0 , \beta_0$), as expected. The value of $q$ is chosen until the smallest values of $\widehat{MSE}(\hat{\boldsymbol{\theta}})$ are obtained. If the constant $q$ with MLqE provides the smallest values for $\widehat{MSE}(\hat{\boldsymbol{\theta}})$, then it is not possible to find another value for $q$ which will give the smallest values for $\widehat{MSE}(\hat{\boldsymbol{\theta}})$.
\end{itemize}

In overall assesment about robustness and modeling, the modeling is a procedure applied the finite sample size. The robustness of score function is based on the limit values at point zero and inifinity.  Since the main idea is based on the modeling capability of $\log_q(f)$, using finiteness of score functions for implying robustness is not necessary (see Table \ref{limitvalueszero} for $q>1$ and $\alpha \geq 1$). When we look at the results of simulation, we have a result in which $q>1$ and $\alpha \geq 1$. When we use p.d. $f(x;\hat{\boldsymbol{\theta}})$ and $q>1$, the values of $w_q=f(x;\hat{\boldsymbol{\theta}})^{1-q}$ are bigger than 1. $\log_q(f)$ gives an advantageous for us to have $w \in [0,\infty]$ and so it is flexible to perform an efficient modeling. %$\log_{\kappa}$ was tried, but efficiency is not as good as $\log_q$, so we omit the numerical results.

\subsection{Real data application}

Real data sets are applied to test the performance of MLE and MLqE methods. The estimates of $\hat{\boldsymbol{\theta}}$ and Kolmogorov-Smirnov (KS) test statistics are provided by Tables \ref{bioconduct}-\ref{15glassfibresp}. We give the plots of c.d. and p.d. functions (CDF and PDF) in Figures \ref{withoutinoutex1f}-\ref{withoutinoutex2f} for illustrative purpose and make a comparison among the fitting competence with estimates from MLE and MLqE. The  variance-covariance of M-estimators \cite{Hub81,God60,GodTh78}  are mainly based on the Taylor expansion of score function around the true value of parameter. Even if M-estimation method provides a family for estimation methodology, a tool from information geometry  was provided by \cite{CanKor18}  as well. The Taylor expansion approach \cite{Hub81} is also used by \cite{MLqEGamma} to determine the value of $q$. Instead of using variance-covariance based on Taylor expansion as a rough approach for score function based on $\log_q$, the $p$-value of KS test statistic should be preferred.  The different $q$-values are tried until the highest $p$-value of KS test statistic from c.d. function of Weibull is reached. This approach is also supported by Figures \ref{withoutinoutex1f}-\ref{withoutinoutex2f} which depict the fitting of c.d. and p.d. functions of Weibull distribution with  $\hat{\alpha}$ and $\hat{\beta}$ (\cite{Canorder20}). 

It is difficult to know the nature of reality and also knowing modality and bimodality in an empirical distribution or a real data set. We have to assume a parametric model and estimate the parameters of underlying distribution as much as we can do. We use two real data sets which are modeled by objective functions $\log_q\big[f(x;\alpha,\beta)\big]$ and $\log\big[f(x;\alpha,\beta)\big]$. Contaminations are applied into real data sets. Thus, we will test the performance of MLE and MLqE when the different types of contamination exist in real data sets. Three types of contamination to the real data set were performed. These are given by the following items:
\begin{enumerate}
	\item Inliers: $\alpha_1=5;\beta_1=10;n_1=100$ and $\alpha_1=\min(x)+0.5$; $\beta_1=\max(x)-0.5;n_1=10;$ for examples 1 and 2, respectively. $x_1=randuni(n_1,\alpha_1,\beta_1)$; $x=[x; x_1$];
	
	randuni is a function written in MATLAB2013a to generate uniform artificial data set from the interval $[a,b]$.  $\min(x)$ and $\max(x)$ represent the minimum and maximum value of observations, respectively. 
	\item Outliers: $x=[x; 2\max(x);3\max(x);4\max(x);5\max(x)]$. 
	\item Both of inliers and outliers: $x=[x; 2\max(x);3\max(x);4\max(x);5\max(x); x_1]$;
\end{enumerate}

\subsubsection{Real data application: Example 1}
This section consists of the numerical example for the application of real data set. R Version 4.0.2 and some packages such as source("http://bioconductor.org/biocLite.R") biocLite("GEOquery") and require(GEOquery) are used to reach the real data set. We use the real data set from "test\$myMean". The sample size $n$ for this data set is 6136. The value of tuning constant is chosen until the highest $p$-value of KS test statistc is obtained when $q$ is $0.85$, which means that the best values for estimates of parameters can obtained. 

%The illustrative form of Figures \ref{withoutinoutex1f}-\ref{withinoutex1f} from CDF and PDF supports this approach \cite{Canorder20}. 
% there is some drawbacks which is necessary to improve.
\begin{table}[htbp]
	\centering
	\caption{The estimates of parameters $\alpha$ and $\beta$ via MLE and MLqE methods}
	\scalebox{0.72}
	{
		\begin{tabular}{c|cccccccc}
			& \multicolumn{2}{l}{Without contamination} & \multicolumn{2}{l}{~~~~~~~~~~~~With inliers} & \multicolumn{2}{l}{~~~~~~~~~~~~With outliers} & \multicolumn{2}{l}{~~~~~~~~~~~~With both} \\ \hline
			$\hat{\boldsymbol{\theta}}$	&   $\hat{\alpha}$   & $\hat{\beta}$    &   $\hat{\alpha}$    &    $\hat{\beta}$   &   $\hat{\alpha}$  & $\hat{\beta}$  &  $\hat{\alpha}$  & $\hat{\beta}$  \\ \hline
			MLE($\boldsymbol{\theta}$)& 3.5427(.00045)  &7.9935(.00038) &3.5623(.00045) &7.9937(.00038) &2.5596(.00065) &7.9907(.00044)&2.5869(.00032)&8.0862 (.00053)       \\
			MLqE($\boldsymbol{\theta}$)	& 3.9137(.00065)&7.9942(.00044)&3.9382(.00064)&	7.9985(.00043)&3.9133(.00033)&7.9940(.00054)&	3.9398(.00064)&	7.9961(.00043) \\ \hline                
		\end{tabular}
	}
	\label{bioconduct}
\end{table}
Table \ref{bioconduct} shows that MLqE can be robust to inliers and ouliers at each case. The estimates of scale parameter $\beta$ can be similar to each other at each case for MLE and MLqE methods. MLE and MLqE with inliers cannot be more different than that of without the contamination case. However, when MLE and MLqE methods are compared for the case of outliers, it is seen that the shape parameter $\alpha$ obtained by MLE is very sensitive outlier. For the real data set, the results show that MLqE is robust to outliers, because the score functions derived from $\log_q$  for parameters $\alpha$ and $\beta$ are finite. For this real data set, the sensivity of estimates of $\beta$ from MLE cannot be more, however when the estimates of $\beta$ is compared with that of MLqE, MLqE is insensitive to contamination in all of cases. The estimates of $\alpha$ from MLqE can be insensitive to contamination in all of cases. Especially, the estimates of $\beta$ from MLqE for both contamination are insensitivite when compared with that of MLE. 

\begin{table}[htbp]
	\centering
	\caption{The $p$-value of KS test statistics computed by the estimates of parameters $\alpha$ and $\beta$ via MLE and MLqE }
	\scalebox{0.72}
	{
		\begin{tabular}{c|cccc}
			& Without contamination &With inliers & With outliers &With both \\ \hline
			$\hat{\boldsymbol{\theta}}$ &$p$-value &$p$-value &$p$-value &$p$-value \\ \hline
			MLE($\boldsymbol{\theta}$)& 4.3225e-07&	2.3326e-07&	6.8452e-76&	2.6225e-67 \\
			MLqE($\boldsymbol{\theta}$)& 1.2516e-04	&1.2895e-04	&1.2767e-04	&1.3000e-04 \\ \hline           
		\end{tabular}
	}
	\label{bioconductp}
\end{table}

\begin{figure}[htbp]
	\centering
	\subfigure{\label{fig:cdfwithoutinlieroutlierex1}\includegraphics[width=0.42\textwidth]{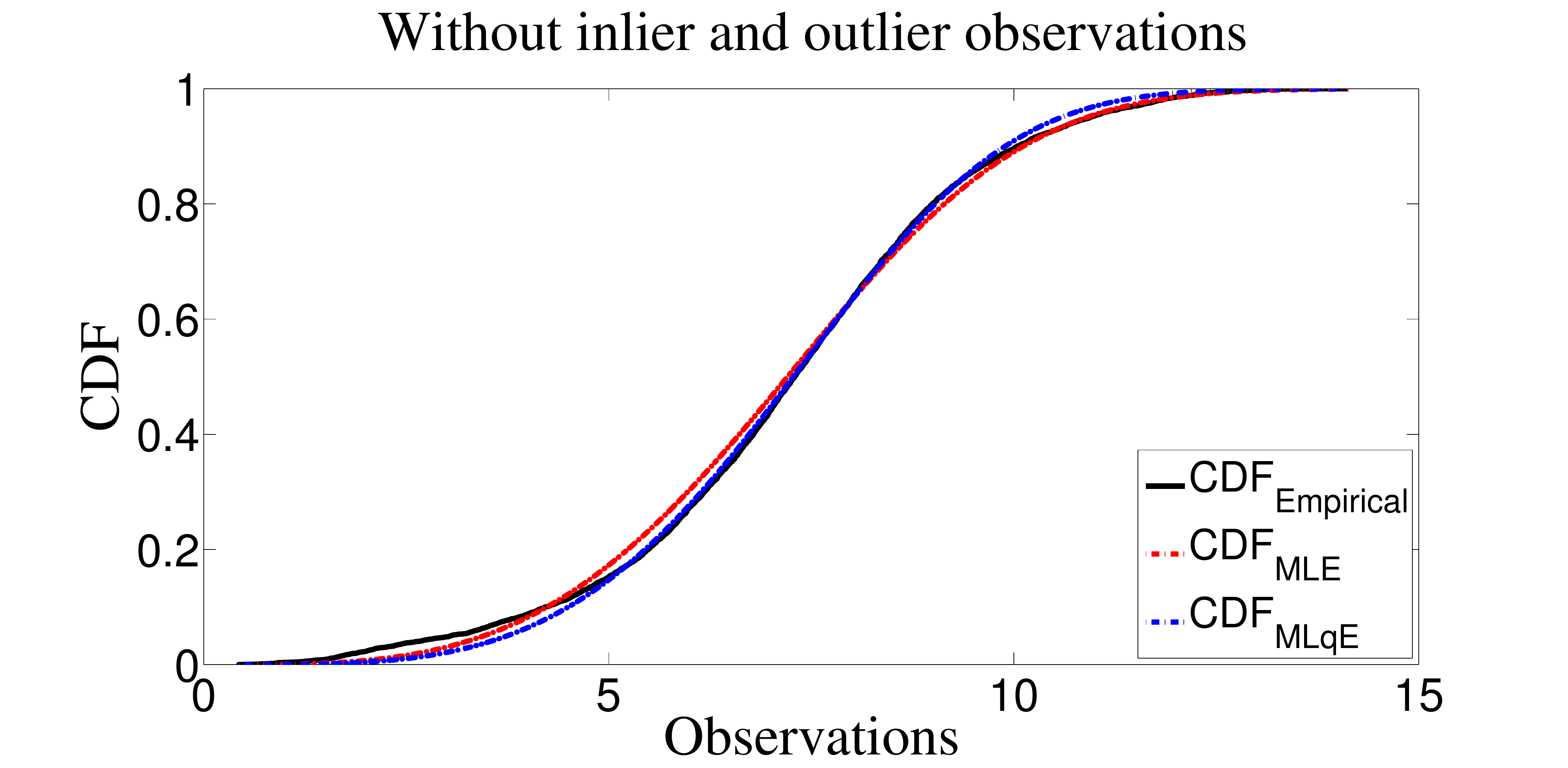}}
	\subfigure{\label{fig:pdfwithoutinlieroutlierex1}\includegraphics[width=0.42\textwidth]{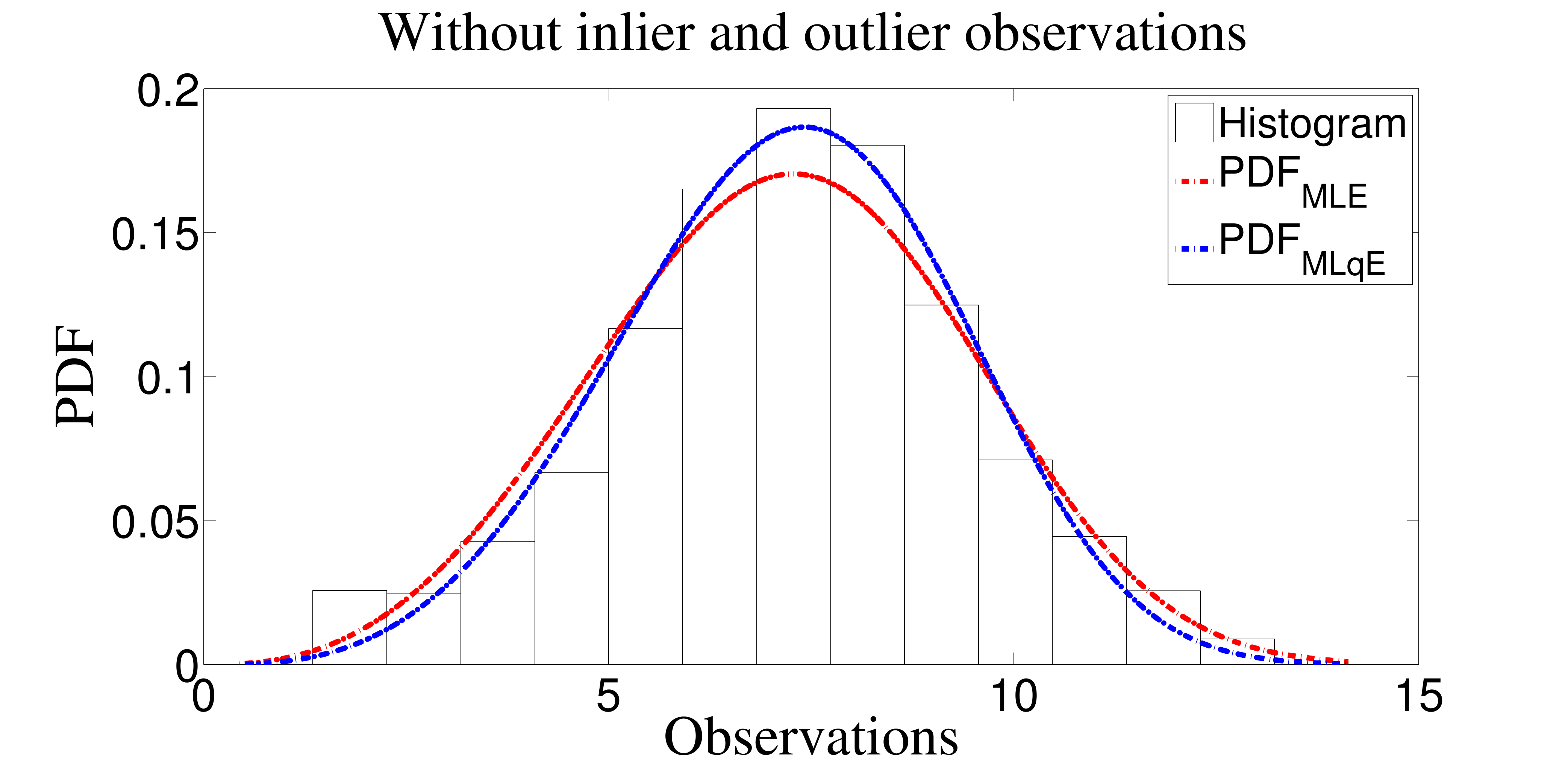}}
	\caption{CDF, PDF and histogram when real data set does not have inliers and outliers for bioconductor test data(color online)}
	\label{withoutinoutex1f}
\end{figure}
\begin{figure}[htbp]
	\centering
	\subfigure{\label{fig:cdfwithinlierex1}\includegraphics[width=0.42\textwidth]{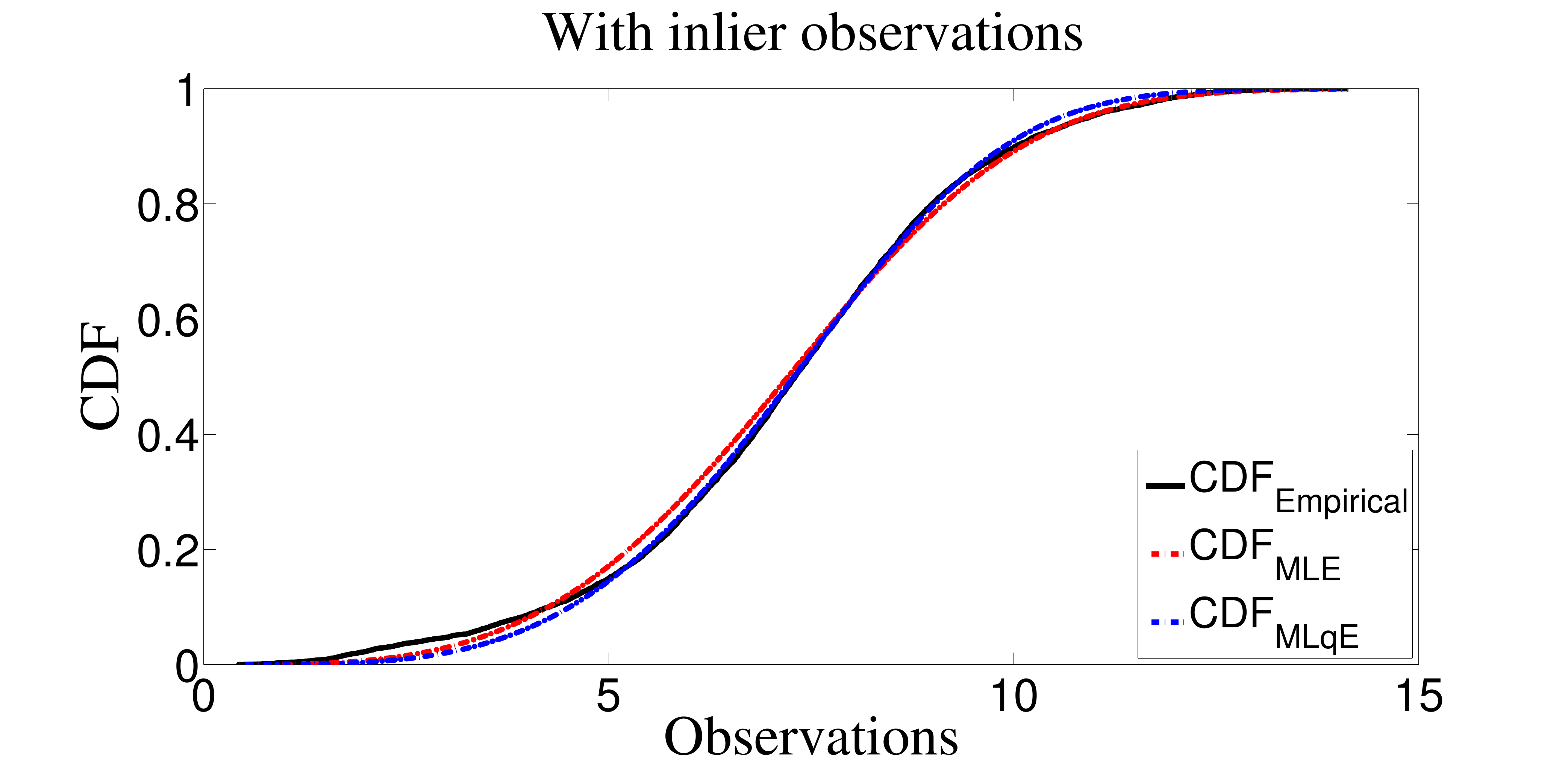}}
	\subfigure{\label{fig:pdfwithinlierex1}\includegraphics[width=0.42\textwidth]{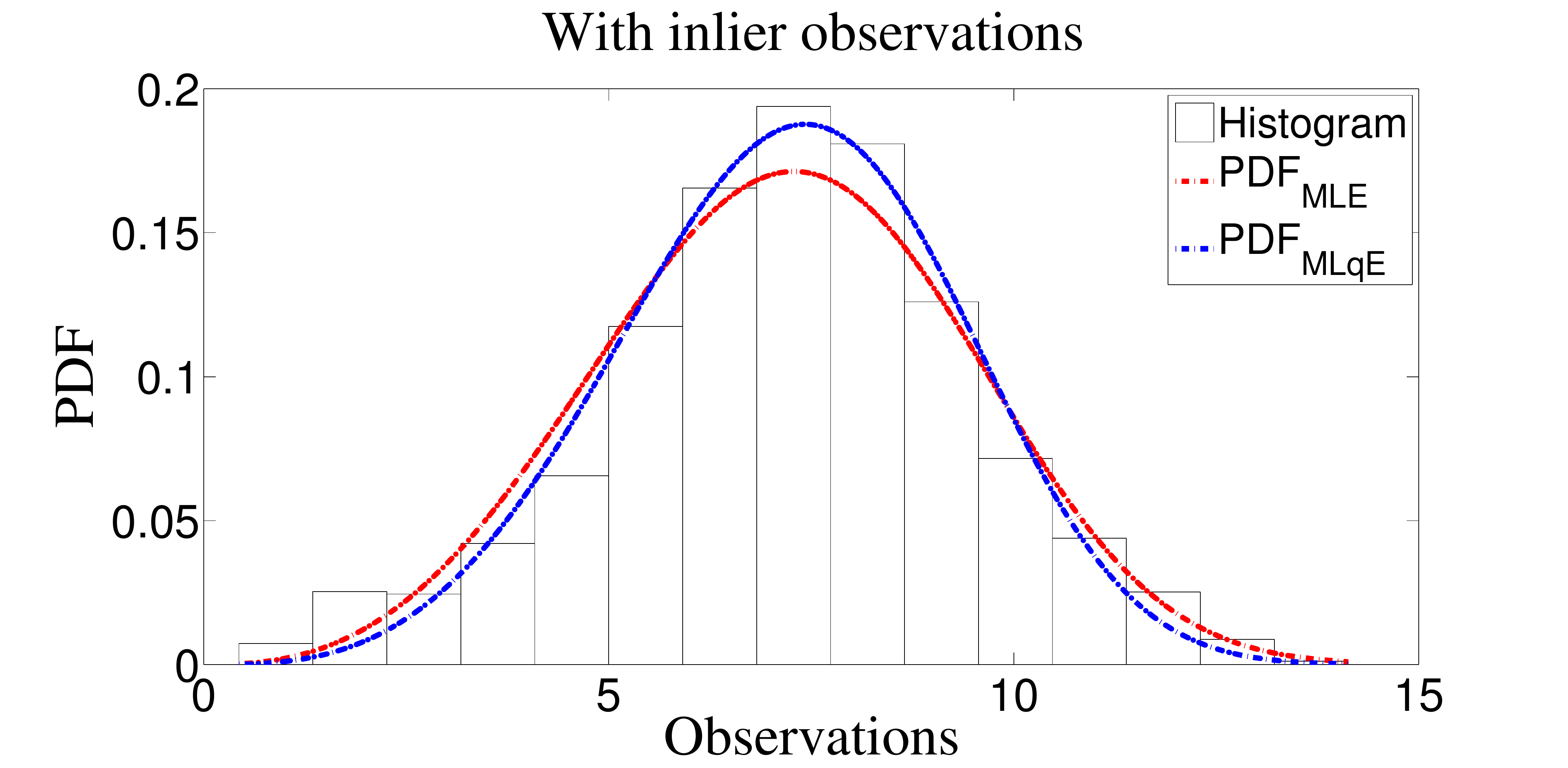}}
	\caption{CDF, PDF and histogram when real data set has inliers for bioconductor test data (color online)}
	\label{withinex1f}
\end{figure}
\begin{figure}[htbp]
	\centering
	\subfigure{\label{fig:cdfwithoutlierex1}\includegraphics[width=0.42\textwidth]{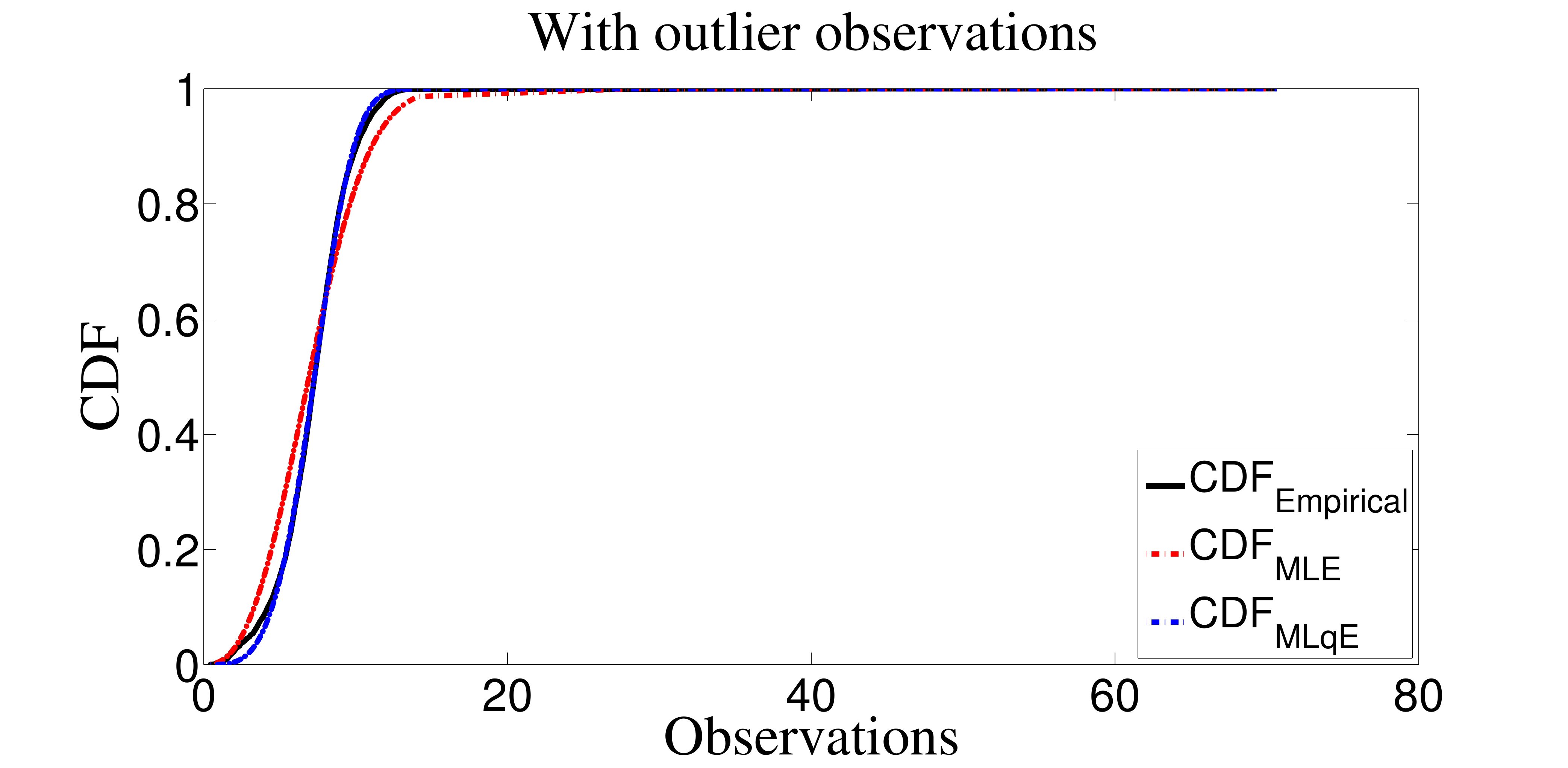}}
	\subfigure{\label{fig:pdfwithoutlierex1}\includegraphics[width=0.42\textwidth]{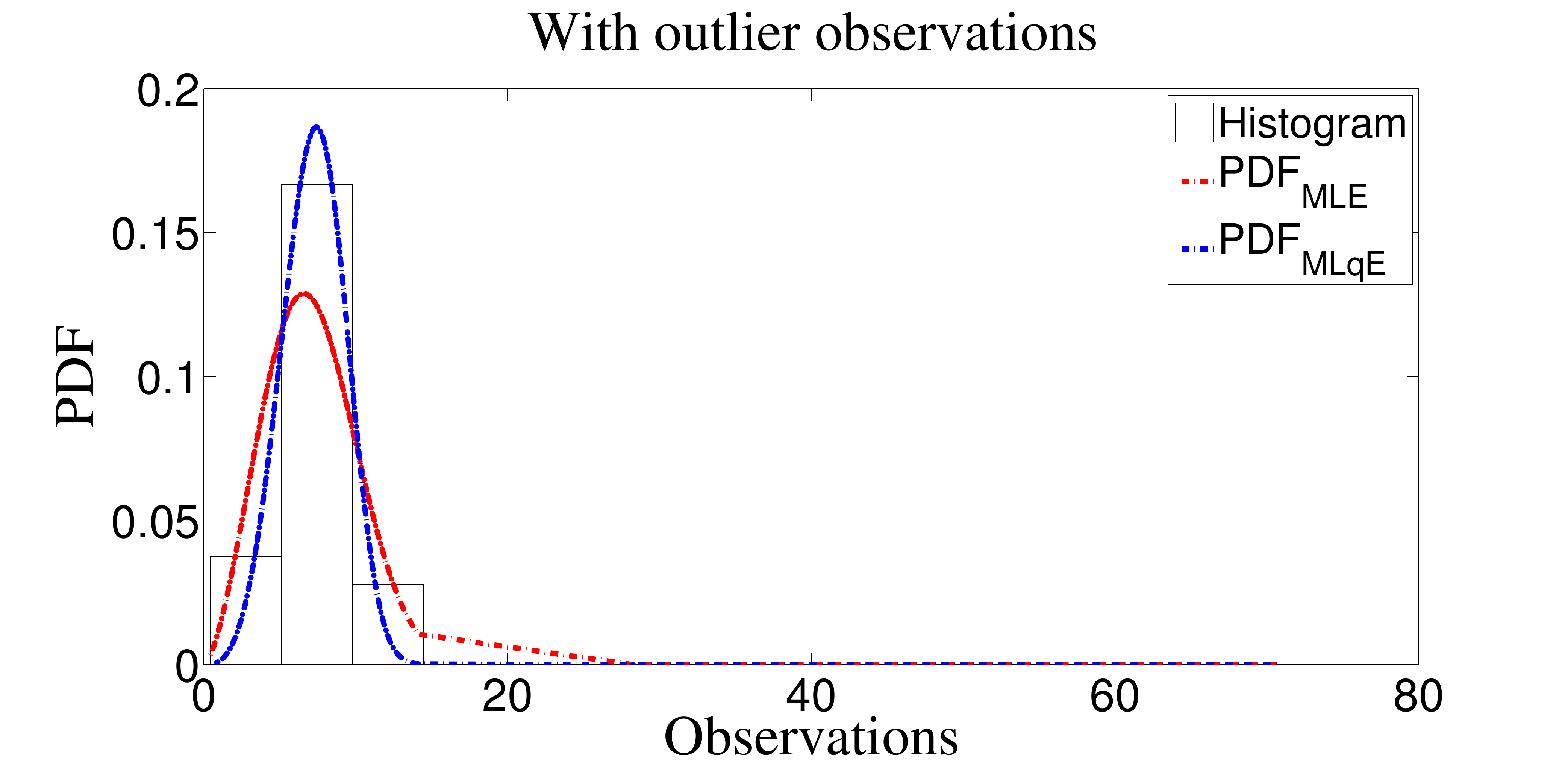}}
	\caption{CDF, PDF and histogram when real data set has outliers for bioconductor test data(color online)}
	\label{withoutex1f}
\end{figure}
\begin{figure}[htbp]
	\centering
	\subfigure{\label{fig:cdfwithinlieroutlierex1}\includegraphics[width=0.42\textwidth]{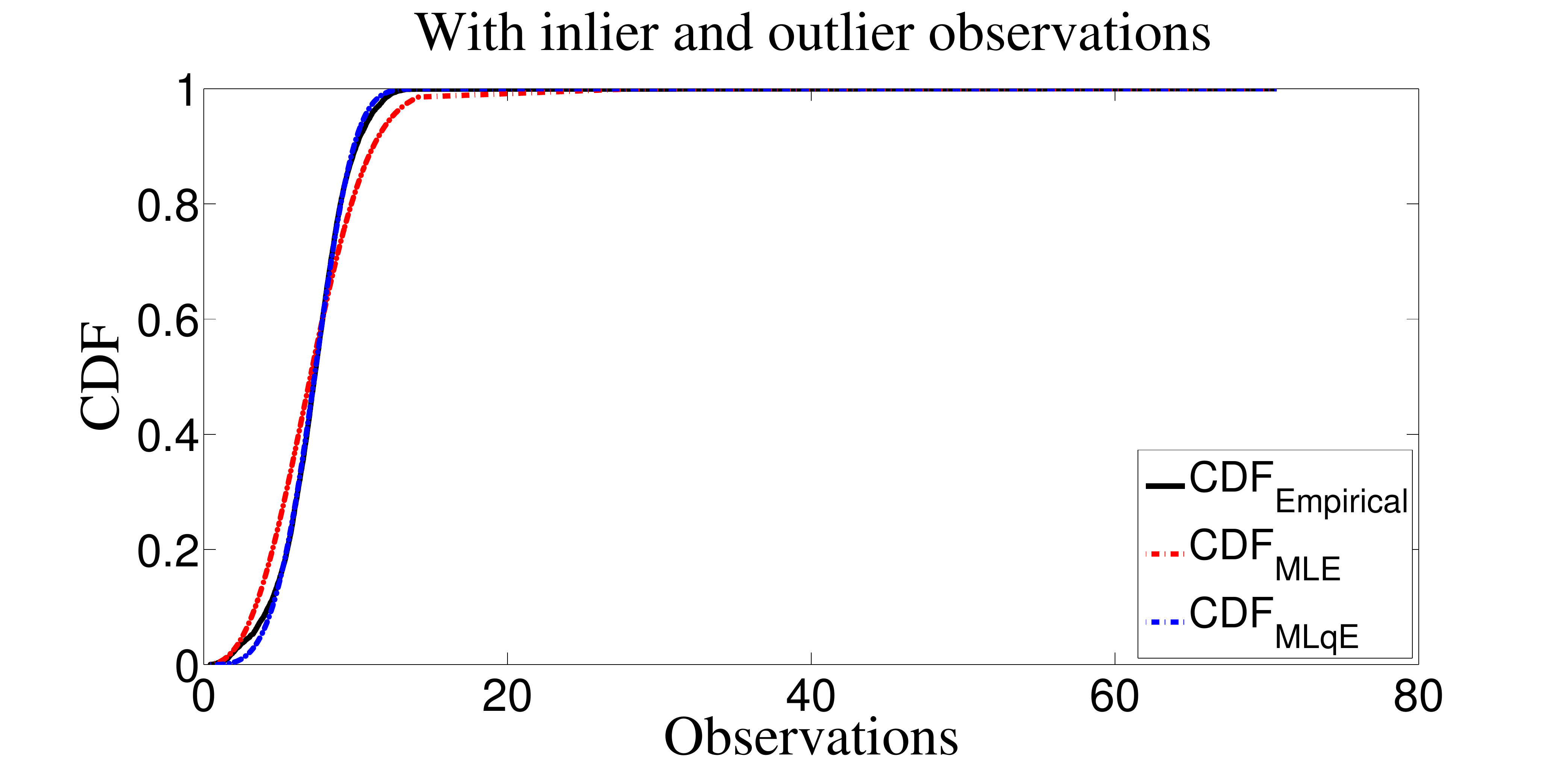}}
	\subfigure{\label{fig:pdfwithinlieroutlierex1}\includegraphics[width=0.42\textwidth]{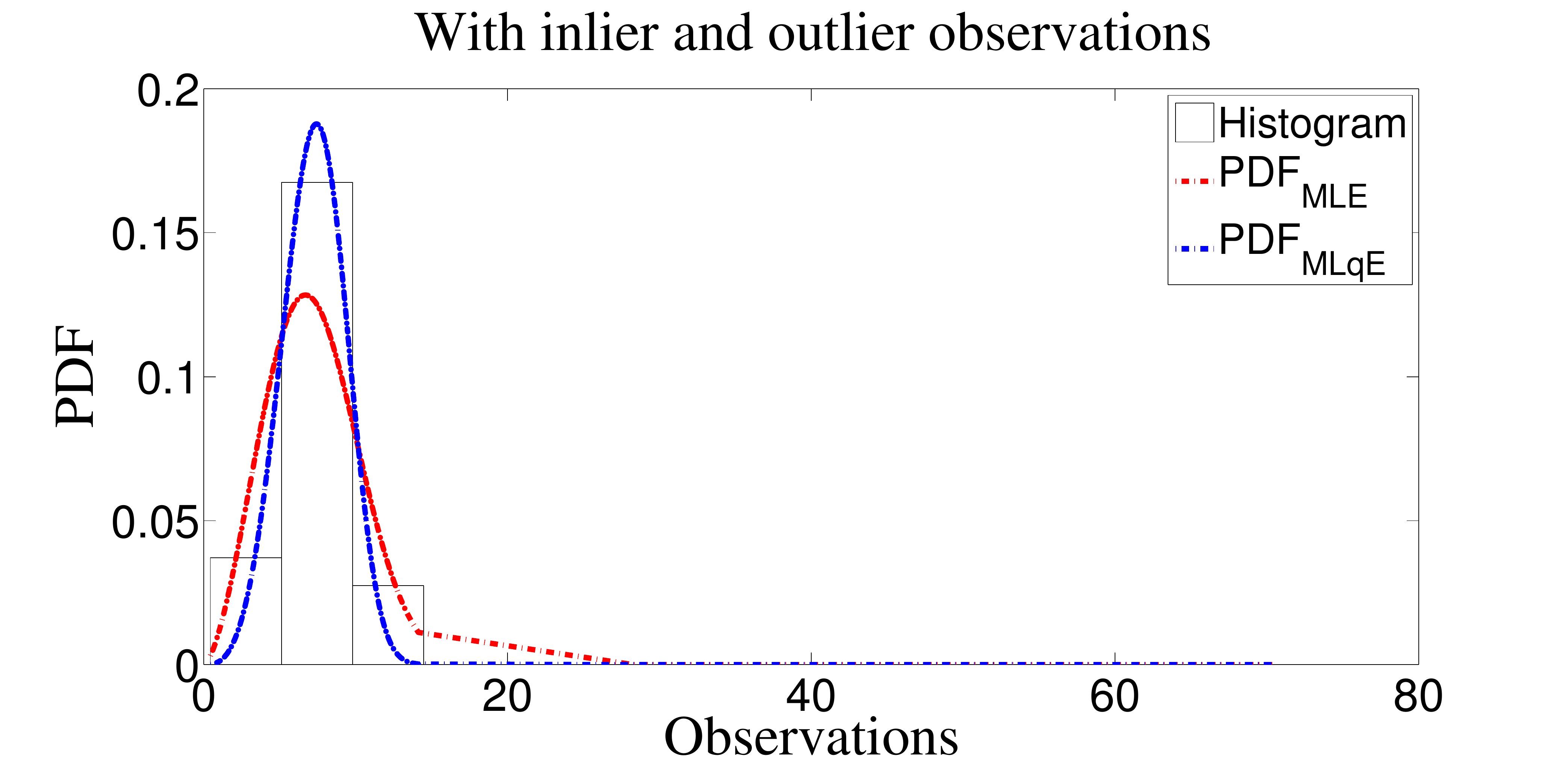}}
	\caption{CDF, PDF and histogram when real data set has inliers and outliers for bioconductor test data(color online)}
	\label{withinoutex1f}
\end{figure}
Table \ref{bioconductp} shows the $p$-value of KS test statistics. According to $p$-values of KS test statistic, there does not exist enough evidence to accept the null hypothesis $H_0$ which shows that the real data set is a member of Weibull distribution. However, note that if the significance level of test statistic is chosen to be $10^{-5}$, then Weibull with MLqE($\boldsymbol{\theta}$) provides sufficient evidence not to reject the null hypothesis $H_0$.  Let us focus on the $p$-values instead of considering whether or not the data set does really come from Weibull distribution with parameters $\alpha$ and $\beta$ which were estimated by MLqE and MLE methods. The $p$-values of KS test statistic obtained from two cases which are outliers and both contaminations are very small when they are compared with that of MLqE, which shows that adding outliers to real data set makes more far from Weibull distribution with the estimates obtained from MLE. However, the $p$-values of KS test statistic with estimates of MLqE($\boldsymbol{\theta}$) did not differ much for three cases which are with inliers, outliers and both contaminations.  Let us focus on the comparison of MLE for without contamination and with inliers cases, $p$-values from $4.3225e-07$ to $2.3326e-07$ tend to be near to zero due to the fact that adding inliers affects the estimates of MLE. In MLqE, such tendency being zero is not observed because of robustness property of MLqE.  

\subsubsection{Real data application: Example 2}
The data are the strengths of 1.5 cm glass fibres measured by the National Physical Laboratory, England. The sample size is $n=63$ (\cite{glass15data}). The distributions and references therein \cite{alphapowerglass15} are used to model the data set. The estimates from Weibull with $\log_q(f)$ as an objective function give the $p$-value which is bigger than that of distributions in (\cite{alphapowerglass15}). Note that the parametric model and objective functions should be tried for the case in which we can improve the modeling competence.

\begin{table}[htbp]
	\centering
	\caption{The estimates of parameters $\alpha$ and $\beta$ via MLE and MLqE methods}
	\scalebox{0.72}
	{
		\begin{tabular}{c|cccccccc}
			& \multicolumn{2}{l}{Without contamination} & \multicolumn{2}{l}{~~~~~~~~~~~~With inliers} & \multicolumn{2}{l}{~~~~~~~~~~~~With outliers} & \multicolumn{2}{l}{~~~~~~~~~~~~With both} \\ \hline
			$\hat{\boldsymbol{\theta}}$	&   $\hat{\alpha}$   &   $\hat{\beta}$         &    $\hat{\alpha}$    &    $\hat{\beta}$   &   $\hat{\alpha}$     &    $\hat{\beta}$  &      $\hat{\alpha}$        &    $\hat{\beta}$         \\ \hline
			MLE($\boldsymbol{\theta}$)&5.7762(.07149)&1.6275(.00471)&5.8626(.06262)&1.6164(.00398)&1.4587(.01697)&2.1153(.02278)&1.4961(.01515)&  2.0564(.01879)          \\
			MLqE($\boldsymbol{\theta}$)	&7.5423(.10755)&1.6401(.00393)&7.6146(.09359)&1.6325(.00334)&7.5136(.10077)&1.6392(.00371)&	7.5692(.08819)  &  	1.6231(.00317) \\ \hline                
		\end{tabular}
	}
	\label{15glassfibres}
\end{table}
Table \ref{15glassfibres} shows that MLqE can be robust to inliers and ouliers at each case. For all of scenarios of contaminations which are inliers, outliers and both of them, the estimates  $\hat{\alpha}$ and $\hat{\beta}$ from MLE are sensitive to contamination. However, such sensivity has not been observed at MLqE. When outliers and both contaminations are examined, it is observed that the estimates from MLE are very sensitive to contamination schemas. However, the estimates from MLqE can have similar values at which there are not contaminations into data set, which shows that MLqE are robust.  

\begin{table}[htbp]
	\centering
	\caption{The $p$-value of KS test statistics computed by the estimates of parameters $\alpha$ and $\beta$ via MLE and MLqE }
	\scalebox{0.72}
	{
		\begin{tabular}{c|cccc}
			& Without contamination &With inliers & With outliers &With both \\ \hline
			$\hat{\boldsymbol{\theta}}$ &$p$-value &$p$-value &$p$-value &$p$-value \\ \hline
			MLE($\boldsymbol{\theta}$)& 0.0936&	0.0740&	9.6151e-07&	1.1024e-07 \\
			MLqE($\boldsymbol{\theta}$)&0.7283&	0.4501&	0.4504&	0.2700 \\ \hline           
		\end{tabular}
	}
	\label{15glassfibresp}
\end{table}
Table \ref{15glassfibresp} shows the $p$-value of KS test statistics. According to $p$-values of Kolmogorov-Smirnov (KS) as a goodness of fit test, there exists enough evidence to accept the null hypothesis $H_0$ which shows that the real data set is a member of Weibull distribution. MLqE and MLE depend on  $\log_q$ and $\log$ functions respectively. So the importance of the used objective function has been observed when we make a comparsion between the $p$-values which are $0.0936$ and $0.7283$ of MLE and MLqE respectively. The value of tuning constant is chosen until the highest $p$-value of KS test statistcs is obtained when $q$ is $0.8$. Even though there is no strict changing of the values of estimates from MLqE, the $p$-values of KS test statistics from MLqE go to lower values. This is due to the definition of KS test statistic which uses values of $x$ (see codes for the computation of $p$-value in Appendix \ref{comppval}). Figures \ref{withoutinoutex2f}-\ref{withinoutex2f} illustrate that there exist a good fitting by MLqE, i.e. CDF$_{\text{MLqE}}$ and PDF$_{\text{MLqE}}$, as supported by $p$-values of KS test statistics in Table \ref{15glassfibresp}.

\begin{figure}[htbp]
	\centering
	\subfigure{\label{fig:cdfwithoutinlieroutlierex2}\includegraphics[width=0.42\textwidth]{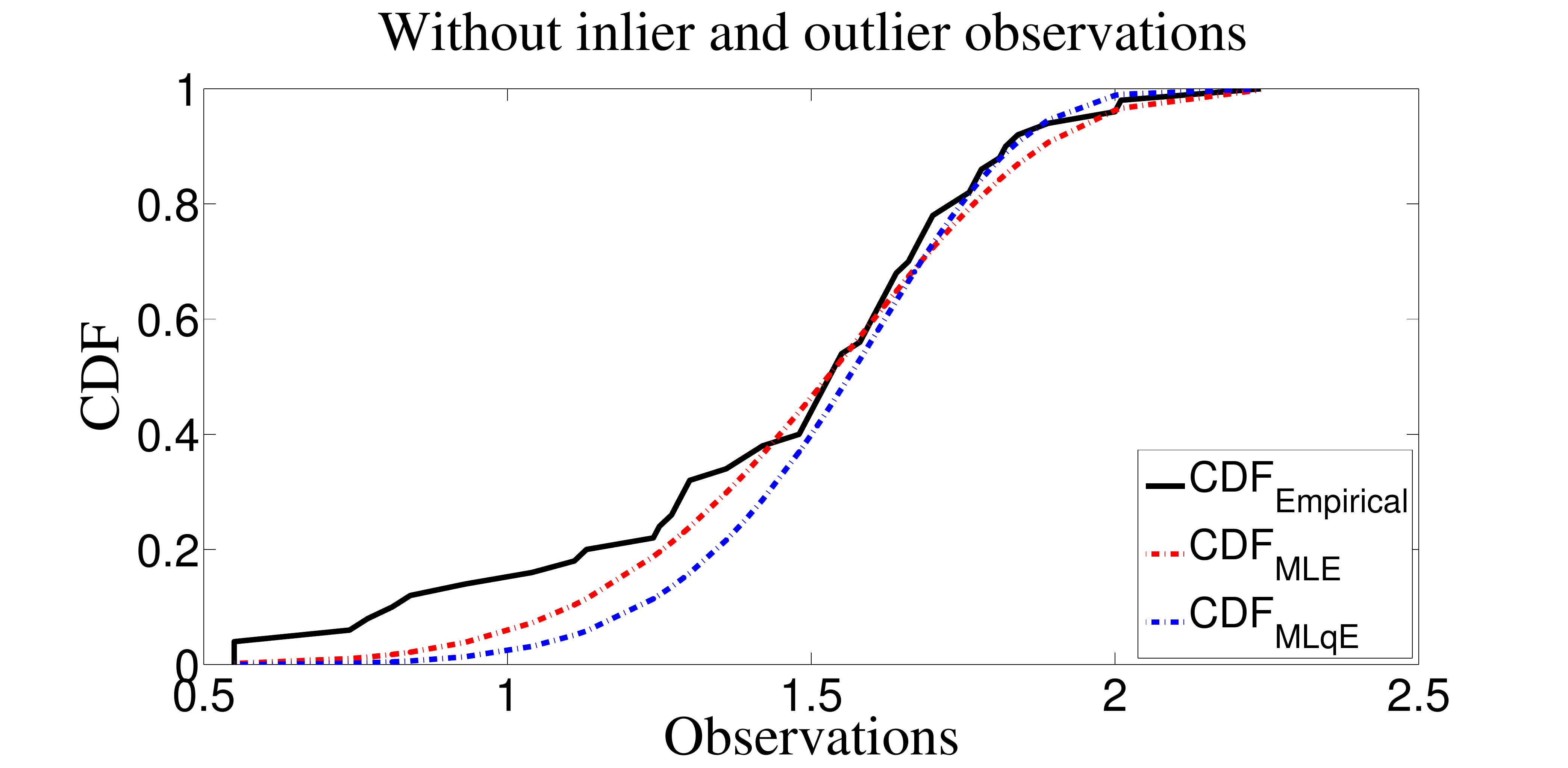}}
	\subfigure{\label{fig:pdfwithoutinlieroutlierex2}\includegraphics[width=0.42\textwidth]{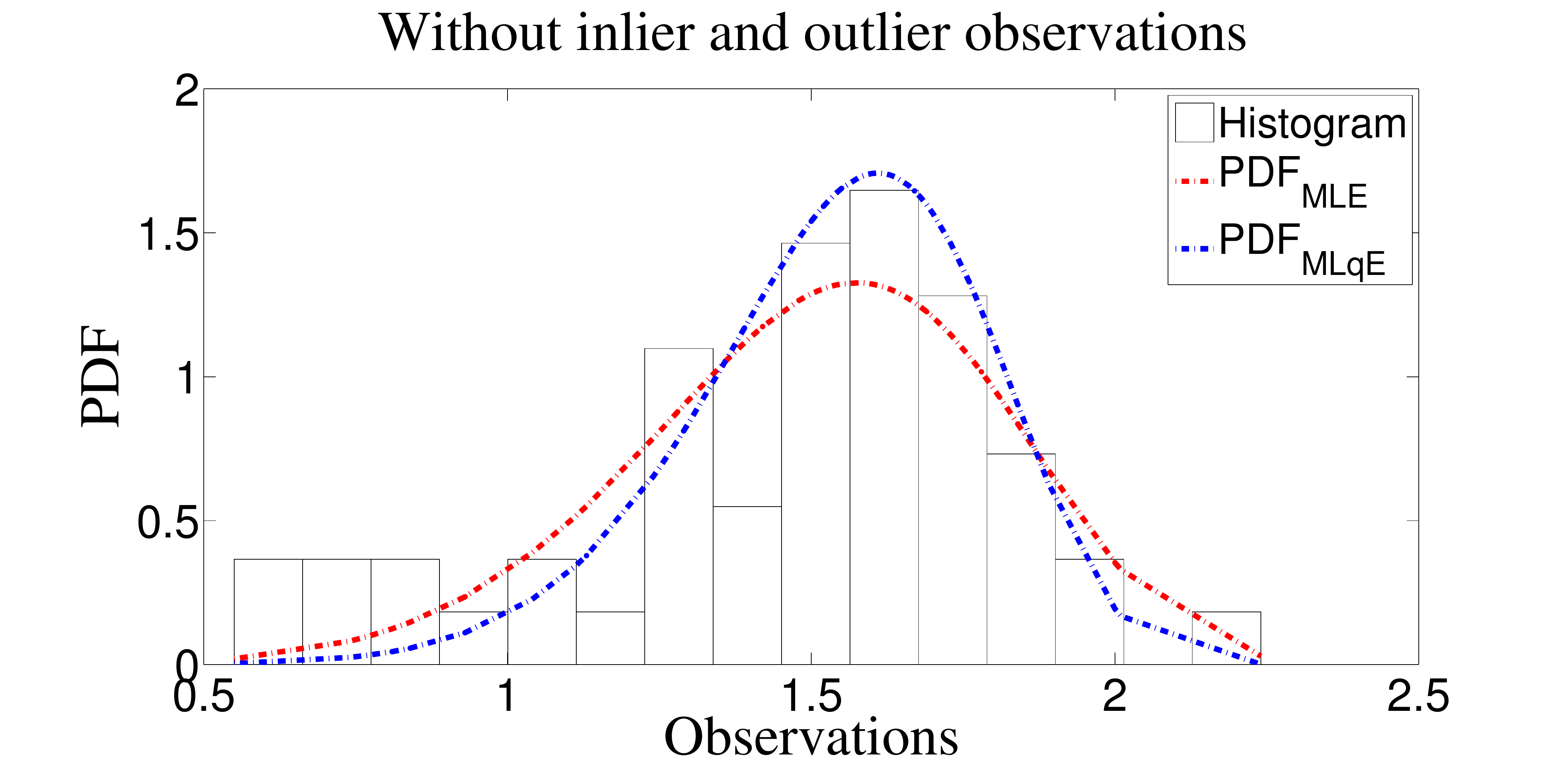}}
	\caption{CDF, PDF and histogram when real data set does not have inliers and outliers for 1.5 cm glass fibres(color online)}
	\label{withoutinoutex2f}
\end{figure}
\begin{figure}[htbp]
	\centering
	\subfigure{\label{fig:cdfwithinlierex2}\includegraphics[width=0.42\textwidth]{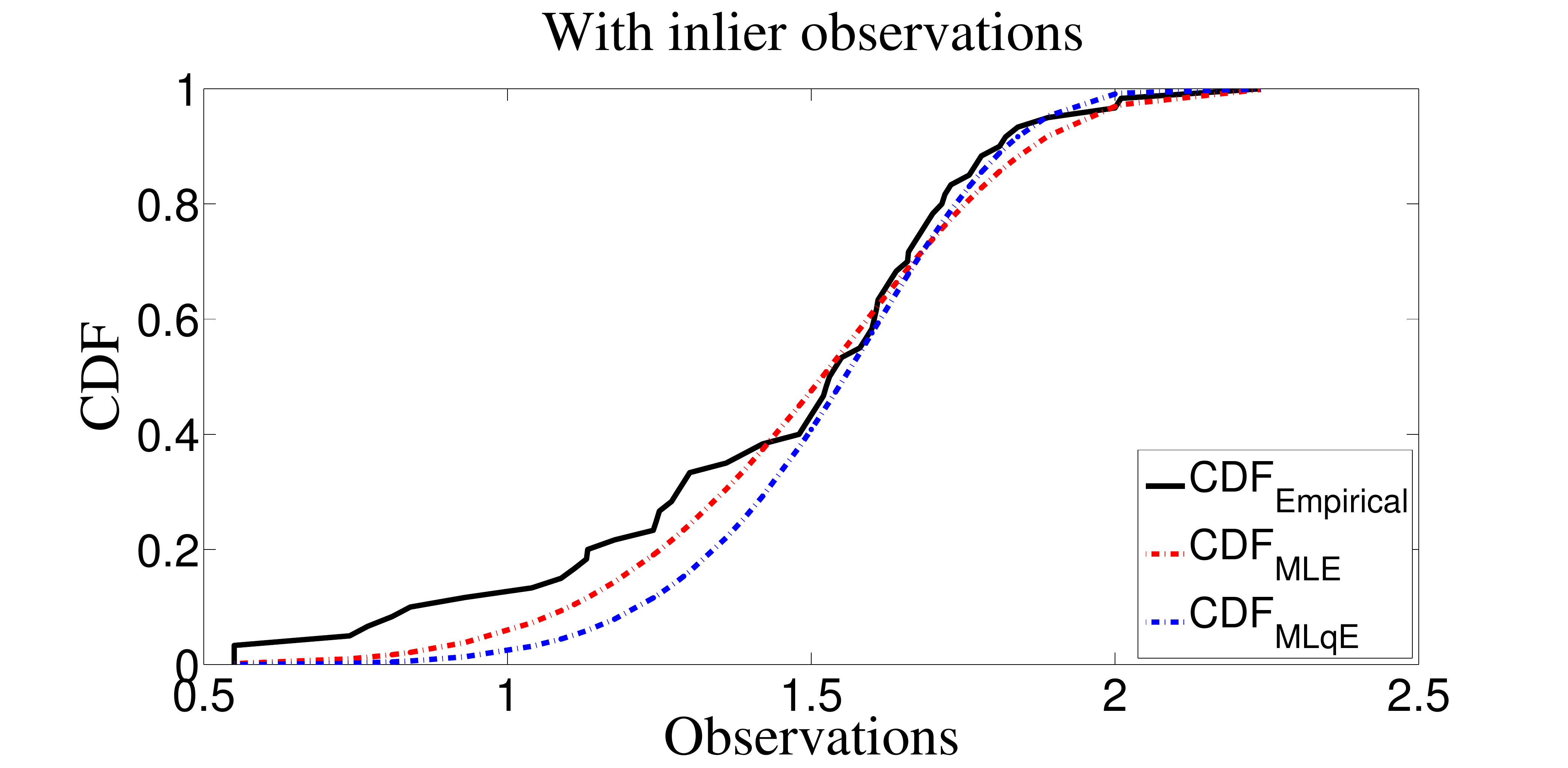}}
	\subfigure{\label{fig:pdfwithinlierex2}\includegraphics[width=0.42\textwidth]{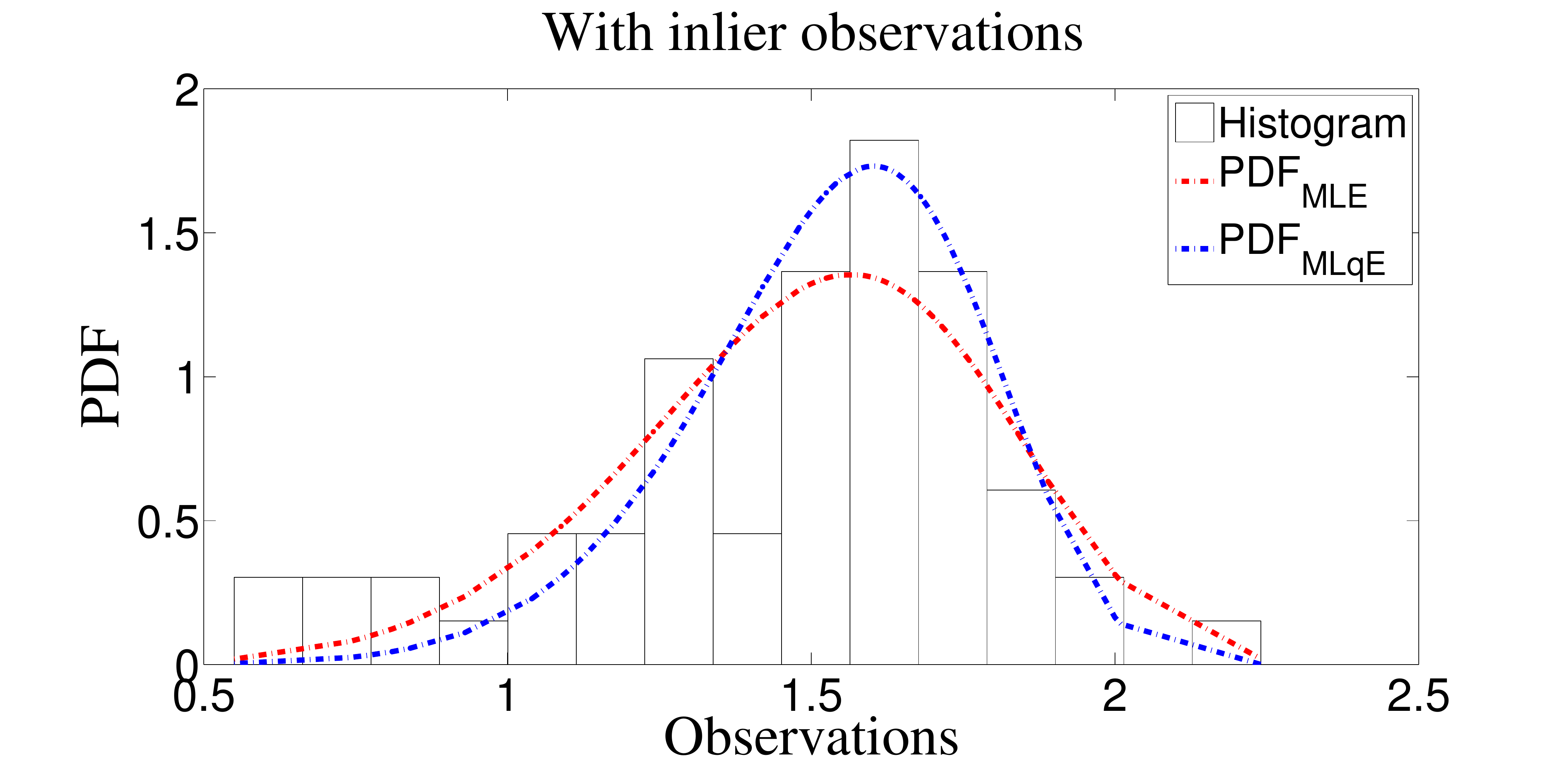}}
	\caption{CDF, PDF and histogram when real data set has inliers for 1.5 cm glass fibres(color online)}
	\label{withinex2f}
\end{figure}
\begin{figure}[htbp]
	\centering
	\subfigure{\label{fig:cdfwithoutlierex2}\includegraphics[width=0.42\textwidth]{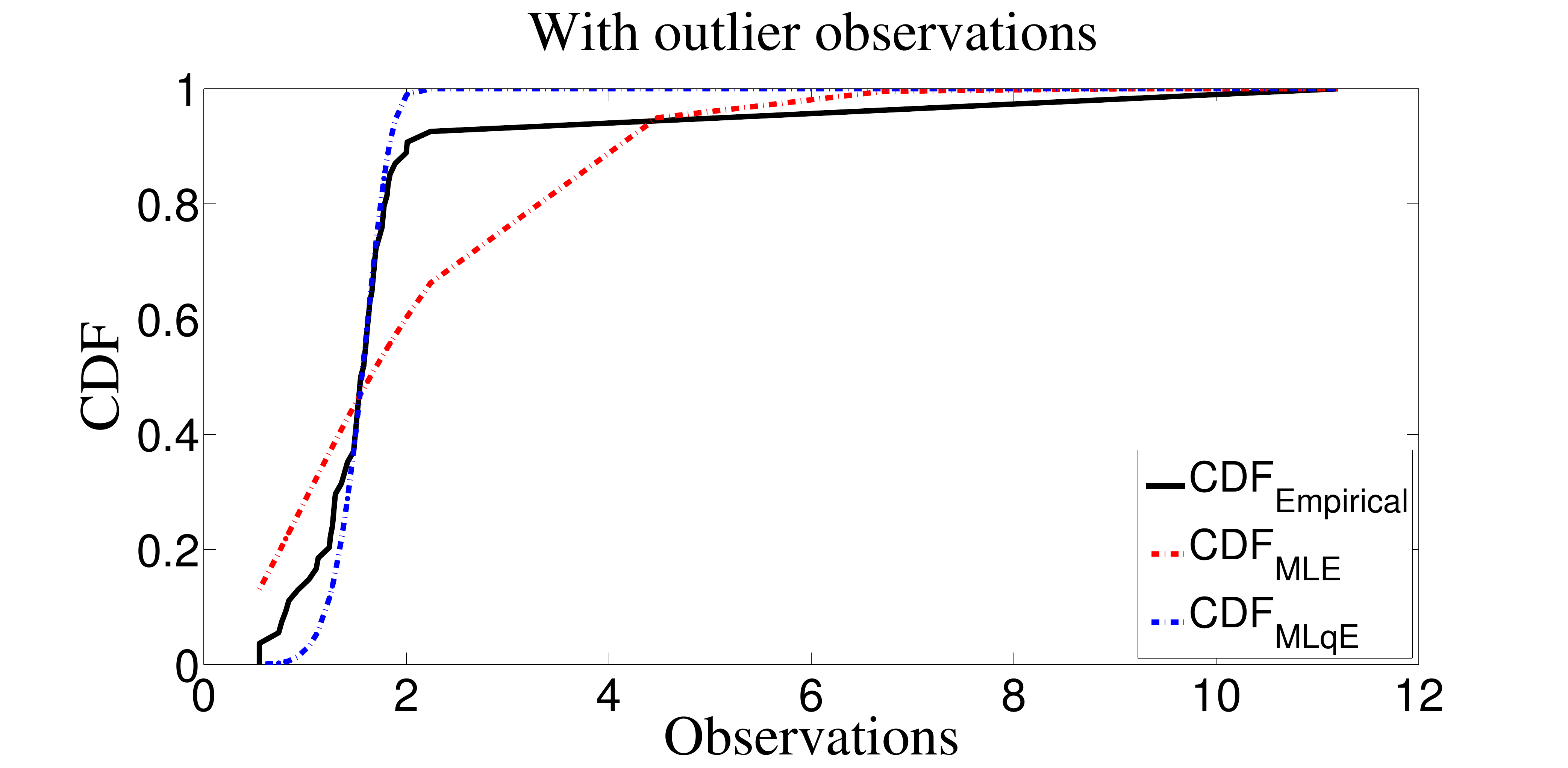}}
	\subfigure{\label{fig:pdfwithoutlierex2}\includegraphics[width=0.42\textwidth]{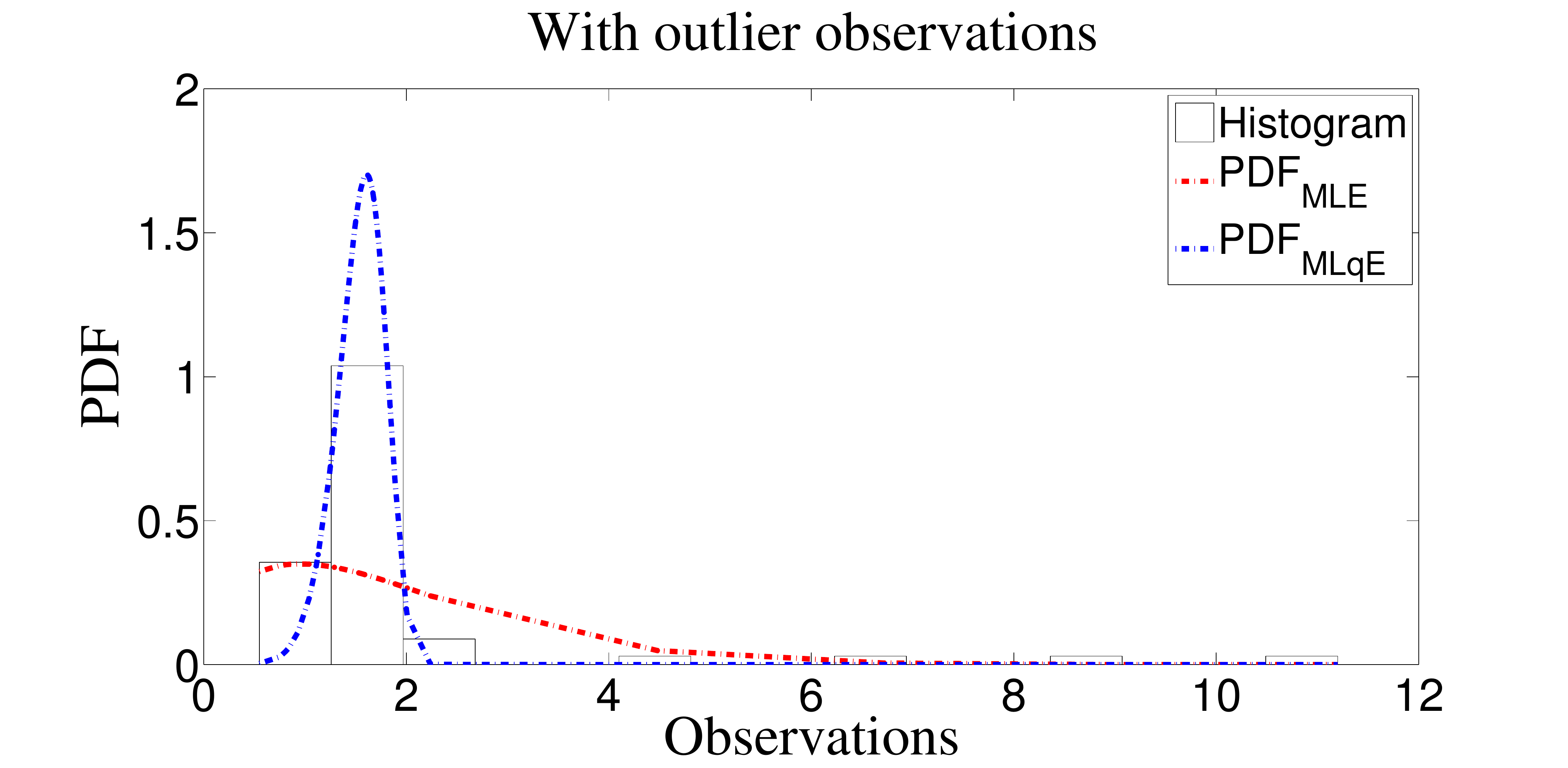}}
	\caption{CDF, PDF and histogram when real data set has outliers for 1.5 cm glass fibres(color online)}
	\label{withoutex2f}
\end{figure}
\begin{figure}[htbp]
	\centering
	\subfigure{\label{fig:cdfwithinlieroutlierex2}\includegraphics[width=0.42\textwidth]{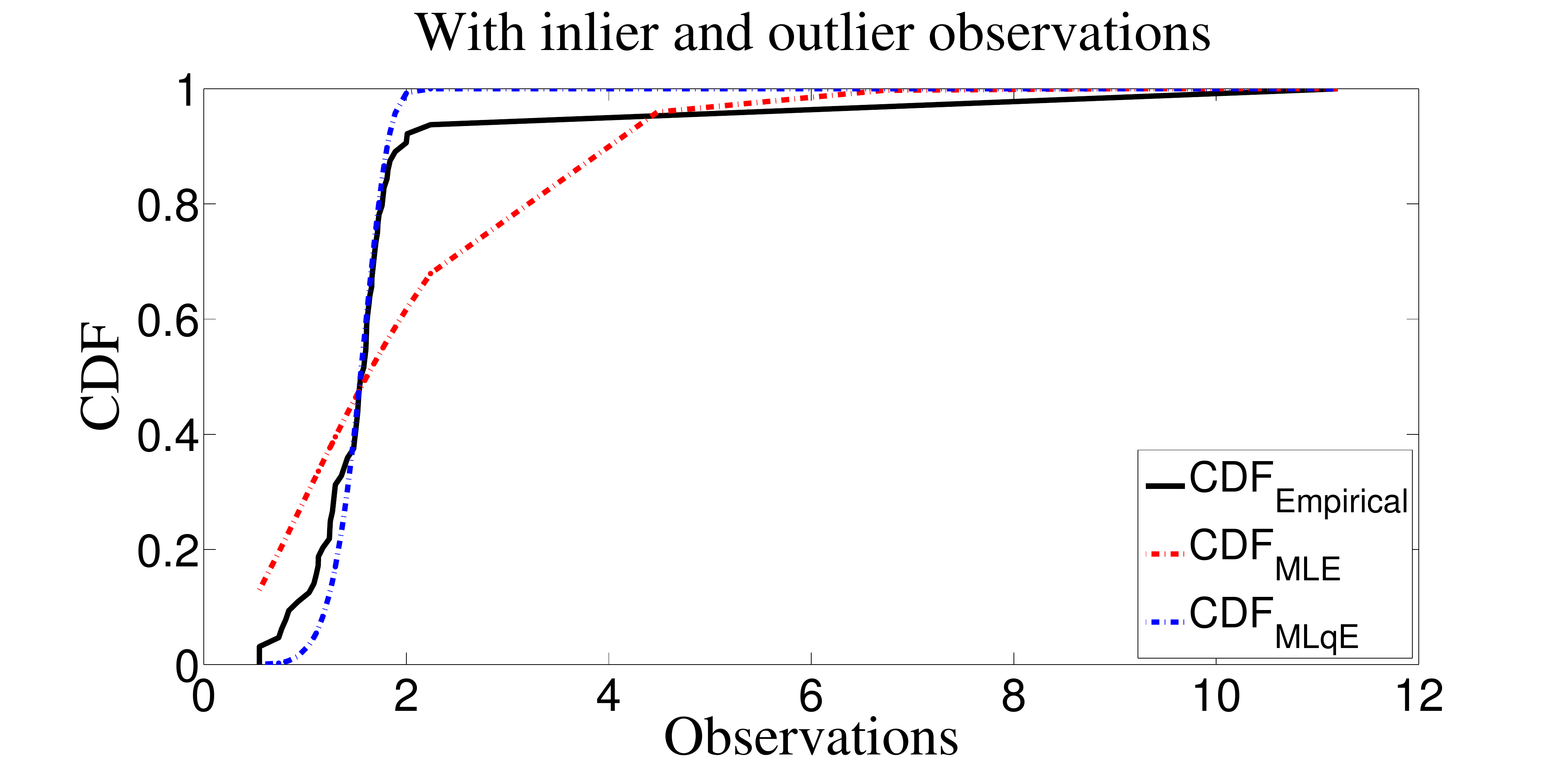}}
	\subfigure{\label{fig:pdfwithinlieroutlierex2}\includegraphics[width=0.42\textwidth]{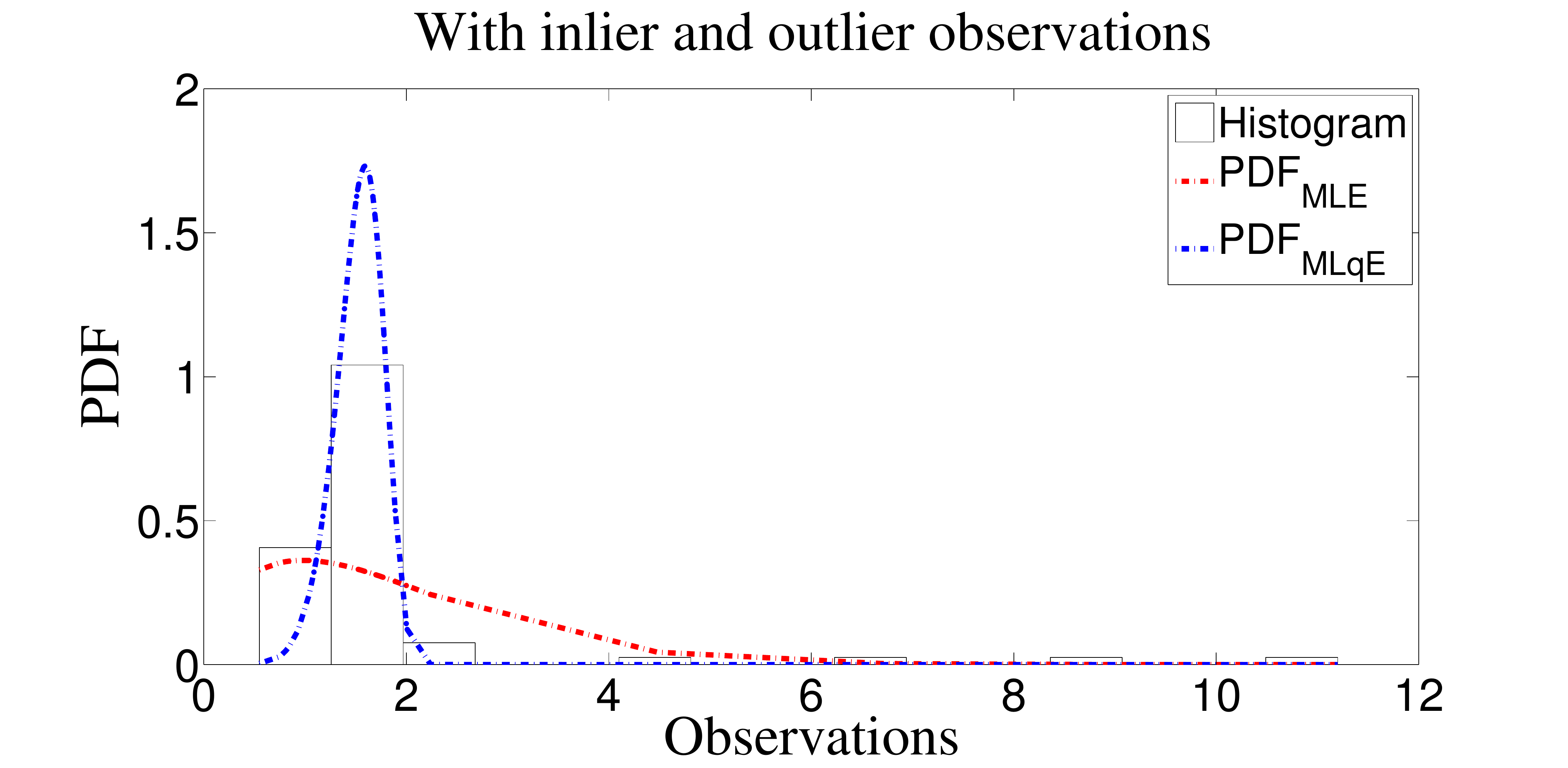}}
	\caption{CDF, PDF and histogram when real data set has inliers and outliers for 1.5 cm glass fibres(color online)}
	\label{withinoutex2f}
\end{figure}
The insensitivities for $\hat{\alpha}$ and $\hat{\beta}$ in Tables \ref{bioconduct}-\ref{15glassfibres} to contaminations show not only the robustness of MLqE but also we can conclude that the modeling competence of  $\log_q(f)$ is better than $\log(f)$. %, because the results from simulation show that even if we have unbounded score function for $q>1$, we have small values of $\widehat{MSE}(\hat{\boldsymbol{\theta}})$, which shows that the modeling competence of $\log_q(f)$ or let us restate that the function $f^{1-q} \mathcal{Z}$  from the estimating equation and $f^{1-q} \geq 0$  which can produce different values of $\mathcal{Z}$ (see Table Case 4 in \ref{caseswut}). 

\section{Conclusions and discussions}\label{conclusionslab}

If the random variables $X_1,X_2,\ldots ,X_n$ are not distributed identically, different values of parameters of Weibull distribution have been estimated robustly by using MLqE. Numerical experiments have been applied to assess the performance of MLqE.  In the simulation, Weibull with different values of parameters from underlying Weibull, BurrIII and uniform distributions have been used for schema of contamination which occurs non-identicality. The performance of MLqE was tested for the functions having one mode property. The contaminated distributions $f_1$ have one mode property as well. GA which is powerful tool to avoid the local points of the optimized function has been used to get the estimates of parameters. Researchers can use MLqE to estimate the parameters of Weibull distribution having contamination from an one mode function. The $p$-value of KS test statistic should be used to determine the value of constant $q$ for the real data set. Thus, the evaluation of fitting performance of Weibull with the estimated parameters can be tested easily. 
% We have situations in which the score functions are not finite for $q>1$; but the modeling performance of $\log_q(f)$  is good. However, the robustness did not play a role in the estimation procedure even if score functions are infinity; because, the main idea is about the fitting performance of the objective functions $\log(f)$ and $\log_q(f)$. 
%
The score functions derived by $\log(f)$  and $\log_q(f)$ are infinite and finite for $0<q<1$ and $\alpha \geq 1$, respectively. The results in real data show that if $0<q<1$, then we have estimates  which can be insensitive to the added inliers and outliers. When we consider on the results of simulation, we can have for inlier case in which $q>1$ which shows that the score functions of parameters are infinite for $q>1$ and $\alpha \geq 1$. Consequently, the numerical experiments show that robustness is not enough to imply that the best modeling is accomplished. By using the MLqE, the values of the parameters representing the majority of the distribution were estimated with small $\widehat{MSE}(\hat{\boldsymbol{\theta}})$ for different scenarios of contaminations. Many results from simulation and the application of real data sets show that MLqE is capable to model efficiently and gives an advantage to obtain the estimates for parameters of underlying distribution. 

Information geometry will be used (\cite{Amari16,Udristegeo}); and the adopted goodness of fit test will be proposed for further advance to determine the value of $q$ by means of tools in statistics. In our future works, we will try to find counter examples from deformation family and its generalization as theoretical results and  numerical experiments for different types of contamination will be used to test their modeling competence even if their score function is infinite. 

\section*{Appendix}
\subsection*{Computation of parameters via GA}\label{comptGA} 
\begin{itemize}
	\item {\bf Run}: \begin{verbatim}
		opt=gaoptimset('CrossoverFcn',{@crossoversinglepoint},'display','off'); 
	\end{verbatim} 
	\begin{enumerate}
		\item \begin{verbatim} lb=[0 0];ub=[10^10 10^10];
			MLqE(t,:)=ga(@(p)fMLqE(p,x),2,[],[],[],[],lb,ub,[],opt); 
		\end{verbatim} 
	\end{enumerate}
\end{itemize}
\begin{itemize}
	\item {\bf Apply}: 'fMLqE' is a function given by
	\begin{enumerate}
		\item  \begin{verbatim} function [sumL] = fMLqE(p,x) \end{verbatim} 
		\item \begin{verbatim}
			a=p(1);b=p(2);f=(a/b)*(x/b).^(a-1).*exp(-(x/b).^a);
		\end{verbatim} 
		\item \begin{verbatim} sumL=-sum((f.^(1-q)-1)./(1-q)); \end{verbatim} 
		When $q=1$, line 3  is replaced by $\log(f)$.
	\end{enumerate}
\end{itemize}
\subsection*{Computation of $p$-value of KS test statistic} \label{comppval}

{\bf Apply}: WeiCDF is a function  written in MATLAB 2013a to compute the CDF values. 
\begin{enumerate}
	\item \begin{verbatim} function F=WeiCDF(a,s,x) \end{verbatim}
	\item \begin{verbatim} F=1-exp(-(x/s).^a); \end{verbatim}
\end{enumerate}
%Note that we write our M.files of MATLAB2013a for random number generation and computation of PDF and CDF values to be in control for computation and the probability integral transformation is used for the procedure of random number generation. 
The $p$-value is given by
\begin{enumerate}
	\item \begin{verbatim} F=WeiCDF(alpha,beta,x); 	\end{verbatim}
	\item \begin{verbatim} test_cdf=[x,F]; \end{verbatim}
	\item \begin{verbatim} [p_value]=kstest(x,'CDF',test_cdf); \end{verbatim}
\end{enumerate}

\section*{Acknowledgements}
We would like to thank so much Editorial Board and anonymous referees to provide the invaluable comments.  
This study was financed in part by the Coordenação de Aperfeiçoamento de Pessoal de Nível Superior - Brasil (CAPES) - Finance Code 001.

\section*{Disclosure statement}
No potential conflict of interest was reported by the author(s).

\end{document}

%% file: table3wp.tex
\begin{table}[htbp]
	\caption{The underlying and contamination distributions are Weibull }
	\scalebox{0.7}
	{  
		\begin{tabular}{c|cccccccc|cccccccc}
			$\hat{\boldsymbol{\theta}}$	&   $\hat{\alpha}$   &   $\hat{\beta}$         &    $\hat{\alpha}$    &    $\hat{\beta}$   &   $\hat{\alpha}$     &    $\hat{\beta}$  &      $\hat{\alpha}$        &    $\hat{\beta}$ &  $\hat{\alpha}$   &   $\hat{\beta}$         &    $\hat{\alpha}$    &    $\hat{\beta}$   &   $\hat{\alpha}$     &    $\hat{\beta}$  &      $\hat{\alpha}$        &    $\hat{\beta}$       \\ \hline
			&	\multicolumn{8}{c}{Case 1: $(1-\varepsilon)$W(4,2) +$\varepsilon$W(1,5), $\varepsilon$=0.1} & \multicolumn{8}{c}{Case 1: $(1-\varepsilon)$W(4,2) +$\varepsilon$W(1,5), $\varepsilon$=0.2} \\ \hline
			& \multicolumn{2}{c}{\scalebox{0.72}{$n=50,q=0.84$}} & \multicolumn{2}{c}{\scalebox{0.72}{$n=100,q=0.84$}} & \multicolumn{2}{c}{\scalebox{0.72}{$n=150,q=0.84$}} & \multicolumn{2}{c}{\scalebox{0.72}{$n=200,q=0.84$}} 	& \multicolumn{2}{c}{\scalebox{0.72}{$n=50,q=0.77$}} & \multicolumn{2}{c}{\scalebox{0.72}{$n=100,q=0.77$}} & \multicolumn{2}{c}{\scalebox{0.72}{$n=150,q=0.77$}} & \multicolumn{2}{c}{\scalebox{0.72}{$n=200,q=0.77$}} \\ \hline
			MLqE($\boldsymbol{\theta}$)&4.0117&1.9984&3.9449&1.9981&3.9294&1.9967&3.9238&1.9970&3.9537&1.9951&3.8851&1.9938&3.8627&1.9939&3.8454&1.9928 \\
			$\widehat{Var}(\hat{\boldsymbol{\theta}})$&0.3883&0.0073&0.1775&0.0036&0.1161&0.0023&0.0883&0.0018&0.4907&0.0092&0.2389&0.0044&0.1516&0.0030&0.1108&0.0022\\
			$\widehat{MSE}(\hat{\boldsymbol{\theta}})$&0.3884&0.0073&0.1806&0.0036&0.1210&0.0024&0.0941&0.0018&0.4928&0.0092&0.2522&0.0045&0.1705&0.0030&0.1347&0.0023\\ \hline
			MLE($\boldsymbol{\theta}$)&1.7602&2.3742&1.5993&2.3786&1.5511&2.3800&1.5226&2.3812&1.3897&2.6791&1.3278&2.6798&1.3082&2.6802&1.2964&2.6796\\
			$\widehat{Var}(\hat{\boldsymbol{\theta}})$&0.2919&0.0518&0.1015&0.0259&0.0576&0.0174&0.0372&0.0130&0.0771&0.0869&0.0289&0.0435&0.0181&0.0284&0.0124&0.0214\\
			$\widehat{MSE}(\hat{\boldsymbol{\theta}})$&5.3086&0.1919&5.8647&0.1692&6.0546&0.1618&6.1746&0.1583&   6.8909&0.5480&7.1698&0.5056&7.2639&0.4911&7.3221&0.4832\\  \hline
			&	\multicolumn{8}{c}{Case 2: $(1-\varepsilon)$W(6,4) +$\varepsilon$W(1,5), $\varepsilon$=0.1} & \multicolumn{8}{c}{Case 2: $(1-\varepsilon)$W(6,4) +$\varepsilon$W(1,5), $\varepsilon$=0.2} \\ \hline
			& \multicolumn{2}{c}{\scalebox{0.72}{$n=50,q=0.82$}} & \multicolumn{2}{c}{\scalebox{0.72}{$n=100,q=0.82$}} & \multicolumn{2}{c}{\scalebox{0.72}{$n=150,q=0.82$}} & \multicolumn{2}{c}{\scalebox{0.72}{$n=200,q=0.82$}} 	& \multicolumn{2}{c}{\scalebox{0.72}{$n=50,q=0.76$}} & \multicolumn{2}{c}{\scalebox{0.72}{$n=100,q=0.76$}} & \multicolumn{2}{c}{\scalebox{0.72}{$n=150,q=0.76$}} & \multicolumn{2}{c}{\scalebox{0.72}{$n=200,q=0.76$}} \\ \hline
			MLqE($\boldsymbol{\theta}$)&6.0814&3.9807&5.9821&3.9823&5.9549&3.9824&5.9371&3.9828&5.8268&3.9667&5.7126&3.9646&5.6736&3.9658&5.6606&3.9650 \\
			$\widehat{Var}(\hat{\boldsymbol{\theta}})$&0.8033&0.0126&0.3677&0.0065&0.2322&0.0043&0.1738&0.0032&1.0473&0.0159&0.4823&0.0079&0.3153&0.0054&0.2322&0.0040\\
			$\widehat{MSE}(\hat{\boldsymbol{\theta}})$&0.8099&0.0130&0.3680&0.0068&0.2342&0.0046&0.1777&0.0035&1.0773&0.0170&0.5648&0.0092&0.4219&0.0066&0.3474&0.0053\\ \hline
			MLE($\boldsymbol{\theta}$)&2.9198&4.2649&2.5628&4.2746&2.4039&4.2797&2.3270&4.2847&2.1074&4.4416&1.9521&4.4407&1.8948&4.4488&1.8674&4.4494\\
			$\widehat{Var}(\hat{\boldsymbol{\theta}})$&1.0803&0.0952&0.5275&0.0509&0.2924&0.0334&0.1991&0.0254&0.3171&0.1437&0.1272&0.0721&0.0753&0.0490&0.0539&0.0364\\
			$\widehat{MSE}(\hat{\boldsymbol{\theta}})$&10.5682&0.1654&12.3418&0.1263&13.2245&0.1116&13.6903&0.1064&15.4692&0.3388&16.5125&0.2663&16.9276&0.2504&17.1320&0.2384\\  \hline
			&	\multicolumn{8}{c}{Case 3: $(1-\varepsilon)$W(5,5) +$\varepsilon$W(2,6), $\varepsilon$=0.1} & \multicolumn{8}{c}{Case 3: $(1-\varepsilon)$W(5,5) +$\varepsilon$W(2,6), $\varepsilon$=0.2} \\ \hline
			& \multicolumn{2}{c}{\scalebox{0.72}{$n=50,q=0.83$}} & \multicolumn{2}{c}{\scalebox{0.72}{$n=100,q=0.83$}} & \multicolumn{2}{c}{\scalebox{0.72}{$n=150,q=0.83$}} & \multicolumn{2}{c}{\scalebox{0.72}{$n=200,q=0.83$}} 	& \multicolumn{2}{c}{\scalebox{0.72}{$n=50,q=0.73$}} & \multicolumn{2}{c}{\scalebox{0.72}{$n=100,q=0.73$}} & \multicolumn{2}{c}{\scalebox{0.72}{$n=150,q=0.73$}} & \multicolumn{2}{c}{\scalebox{0.72}{$n=200,q=0.73$}} \\ \hline
			MLqE($\boldsymbol{\theta}$)&5.0053&5.0198&4.9333&5.0161&4.9117&5.0177&4.9002&5.0177&5.0820&5.0168&4.9997&5.0160&4.9665&5.0140&4.9475&5.0166\\
			$\widehat{Var}(\hat{\boldsymbol{\theta}})$&0.5207&0.0289&0.2328&0.0144&0.1590&0.0097&0.1163&0.0072&0.6637&0.0365&0.3162&0.0178&0.1996&0.0117&0.1496&0.0090\\
			$\widehat{MSE}(\hat{\boldsymbol{\theta}})$&0.5207&0.0293&0.2372&0.0146&0.1667&0.0100&0.1262&0.0075&0.6704&0.0368&0.3162&0.0181&0.2008&0.0119&0.1524&0.0093\\ \hline
			MLE($\boldsymbol{\theta}$)&3.8745&5.1391&3.6735&5.1411&3.6032&5.1439&3.5537&5.1446&3.2677&5.2637&3.1351&5.2671&3.0810&5.2674&3.0652&5.2669\\
			$\widehat{Var}(\hat{\boldsymbol{\theta}})$&0.5877&0.0465&0.3243&0.0237&0.2239&0.0158&0.1695&0.0118&0.3543&0.0665&0.1678&0.0333&0.1046&0.0223&0.0796&0.0170\\
			$\widehat{MSE}(\hat{\boldsymbol{\theta}})$&1.8545&0.0659&2.0839&0.0436&2.1750&0.0365&2.2613&0.0327&3.3551&0.1360&3.6455&0.1046&3.7871&0.0938&3.8231&0.0883\\ \hline
&	\multicolumn{8}{c}{Case 4: $(1-\varepsilon)$W(10,8) +$\varepsilon$W(4,10), $\varepsilon$=0.1} & \multicolumn{8}{c}{Case 4: $(1-\varepsilon)$W(10,8) +$\varepsilon$W(4,10), $\varepsilon$=0.2} \\ \hline
& \multicolumn{2}{c}{\scalebox{0.72}{$n=50,q=0.82$}} & \multicolumn{2}{c}{\scalebox{0.72}{$n=100,q=0.82$}} & \multicolumn{2}{c}{\scalebox{0.72}{$n=150,q=0.82$}} & \multicolumn{2}{c}{\scalebox{0.72}{$n=200,q=0.82$}} 	& \multicolumn{2}{c}{\scalebox{0.72}{$n=50,q=0.74$}} & \multicolumn{2}{c}{\scalebox{0.72}{$n=100,q=0.73$}} & \multicolumn{2}{c}{\scalebox{0.72}{$n=150,q=0.73$}} & \multicolumn{2}{c}{\scalebox{0.72}{$n=200,q=0.73$}} \\ \hline
MLqE($\boldsymbol{\theta}$)&10.0432&8.0523&9.9416&8.0521&9.9305&8.0521&9.8742&8.0522&10.0133&8.0936&10.0648&8.0798&10.0319&8.0819&10.0015&8.0829\\
$\widehat{Var}(\hat{\boldsymbol{\theta}})$&2.4484&0.0201&1.1634&0.0096&0.8164&0.0065&0.6065&0.0051&3.1352&0.0264&1.6086&0.0126&1.1025&0.0085&0.8469&0.0067\\
$\widehat{MSE}(\hat{\boldsymbol{\theta}})$&2.4502&0.0229&1.1668&0.0123&0.8212&0.0092&0.6223&0.0078&3.1354&0.0352&1.6128&0.0190&1.1035&0.0152&0.8469&0.0135\\ \hline
MLE($\boldsymbol{\theta}$)&6.3489&8.2799&5.9287&8.2881&5.8087&8.2890&5.7108&8.2928&5.2740&8.5174&5.0506&8.5238&4.9863&8.5280&4.9650&8.5302\\
$\widehat{Var}(\hat{\boldsymbol{\theta}})$&2.2001&0.0424&0.9516&0.0206&0.5923&0.0145&0.4180&0.0109&0.9024&0.0577&0.3525&0.0292&0.2149&0.0190&0.1645&0.0143\\
$\widehat{MSE}(\hat{\boldsymbol{\theta}})$&15.5306&0.1207&17.5273&0.1036&18.1596&0.0980&18.8150&0.0966&23.2372&0.3253&24.8492&0.3035&25.3524&0.2977&25.5160&0.2954\\ \hline
&	\multicolumn{8}{c}{Case 5: $(1-\varepsilon)$W(3,5) +$\varepsilon$W(2,7), $\varepsilon$=0.1} & \multicolumn{8}{c}{Case 5: $(1-\varepsilon)$W(3,5) +$\varepsilon$W(2,7), $\varepsilon$=0.2} \\ \hline
& \multicolumn{2}{c}{\scalebox{0.72}{$n=50,q=0.84$}} & \multicolumn{2}{c}{\scalebox{0.72}{$n=100,q=0.82$}} & \multicolumn{2}{c}{\scalebox{0.72}{$n=150,q=0.82$}} & \multicolumn{2}{c}{\scalebox{0.72}{$n=200,q=0.82$}} 	& \multicolumn{2}{c}{\scalebox{0.72}{$n=50,q=0.75$}} & \multicolumn{2}{c}{\scalebox{0.72}{$n=100,q=0.75$}} & \multicolumn{2}{c}{\scalebox{0.72}{$n=150,q=0.74$}} & \multicolumn{2}{c}{\scalebox{0.72}{$n=200,q=0.74$}} \\ \hline
MLqE($\boldsymbol{\theta}$)&3.0231&5.0735&3.0327&5.0527&3.0191&5.0543&3.0108&5.0532&3.0520&5.1156&3.0003&5.1164&3.0117&5.1089&3.0024&5.1063\\
$\widehat{Var}(\hat{\boldsymbol{\theta}})$&0.1642&0.0781&0.0782&0.0382&0.0519&0.0253&0.0381&0.0192&0.1995&0.0964&0.0931&0.0468&0.0624&0.0308&0.0456&0.0231\\
$\widehat{MSE}(\hat{\boldsymbol{\theta}})$&0.1647&0.0835&0.0793&0.0410&0.0523&0.0283&0.0383&0.0221&0.2022&0.1098&0.0931&0.0604&0.0625&0.0426&0.0456&0.0344\\ \hline
MLE($\boldsymbol{\theta}$)&2.6699&5.2183&2.6007&5.2154&2.5756&5.2189&2.5600&5.2188&2.4435&5.4241&2.3859&5.4267&2.3642&5.4295&2.3518&5.4286\\
$\widehat{Var}(\hat{\boldsymbol{\theta}})$&0.1476&0.0888&0.0733&0.0424&0.0513&0.0285&0.0371&0.0216&0.1148&0.1111&0.0552&0.0542&0.0348&0.0360&0.0261&0.0269\\
$\widehat{MSE}(\hat{\boldsymbol{\theta}})$&0.2565&0.1364&0.2327&0.0888&0.2315&0.0764&0.2307&0.0695&0.4245&0.2910&0.4324&0.2363&0.4391&0.2204&0.4462&0.2106\\ \hline
		\end{tabular}
	}
	\label{caseswt}
\end{table}

%% file: table4wp.tex
\begin{table}[htbp]
	\caption{The underlying and contamination distributions are Weibull (continuation of the Table \ref{caseswt})}
	\scalebox{0.7}
	{  
		\begin{tabular}{c|cccccccc|cccccccc}
			&	\multicolumn{8}{c}{Case 6: $(1-\varepsilon)$W(1,5) +$\varepsilon$W(2,8), $\varepsilon$=0.1} & \multicolumn{8}{c}{Case 6: $(1-\varepsilon)$W(1,5) +$\varepsilon$W(2,8), $\varepsilon$=0.2} \\ \hline
			& \multicolumn{2}{c}{\scalebox{0.72}{$n=50,q=0.97$}} & \multicolumn{2}{c}{\scalebox{0.72}{$n=100,q=0.97$}} & \multicolumn{2}{c}{\scalebox{0.72}{$n=150,q=0.97$}} & \multicolumn{2}{c}{\scalebox{0.72}{$n=200,q=0.97$}} 	& \multicolumn{2}{c}{\scalebox{0.72}{$n=50,q=0.87$}} & \multicolumn{2}{c}{\scalebox{0.72}{$n=100,q=0.87$}} & \multicolumn{2}{c}{\scalebox{0.72}{$n=150,q=0.87$}} & \multicolumn{2}{c}{\scalebox{0.72}{$n=200,q=0.87$}} \\ \hline
MLqE($\boldsymbol{\theta}$)&1.0720&5.1854&1.0565&5.1629&1.0519&5.1597&1.0513&5.1636&1.1229&5.0455&1.1103&5.0387&1.1053&5.0389&1.1020&5.0298\\
$\widehat{Var}(\hat{\boldsymbol{\theta}})$&0.0157&0.5523&0.0074&0.2762&0.0049&0.1820&0.0037&0.1358&0.0217&0.5414&0.0106&0.2750&0.0070&0.1796&0.0051&0.1326\\
$\widehat{MSE}(\hat{\boldsymbol{\theta}})$&0.0209&0.5866&0.0106&0.3027&0.0076&0.2075&0.0063&0.1626&0.0368&0.5435&0.0228&0.2765&0.0181&0.1811&0.0155&0.1335\\ \hline
MLE($\boldsymbol{\theta}$)&1.0727&5.3231&1.0568&5.3022&1.0521&5.2999&1.0514&5.3035&1.1238&5.6089&1.1099&5.6057&1.1044&5.6091&1.1008&5.6013\\
$\widehat{Var}(\hat{\boldsymbol{\theta}})$&0.0151&0.5622&0.0071&0.2814&0.0047&0.1853&0.0035&0.1390&0.0173&0.5579&0.0084&0.2808&0.0055&0.1824&0.0041&0.1360\\
$\widehat{MSE}(\hat{\boldsymbol{\theta}})$&0.0204&0.6666&0.0103&0.3727&0.0074&0.2752&0.0062&0.2311&0.0326&0.9287&0.0205&0.6477&0.0164&0.5534&0.0142&0.4976\\ \hline
			&	\multicolumn{8}{c}{Case 7: $(1-\varepsilon)$W(5,1) +$\varepsilon$W(2,8), $\varepsilon$=0.1} & \multicolumn{8}{c}{Case 7: $(1-\varepsilon)$W(5,1) +$\varepsilon$W(2,8), $\varepsilon$=0.2} \\ \hline
& \multicolumn{2}{c}{\scalebox{0.72}{$n=50,q=0.87$}} & \multicolumn{2}{c}{\scalebox{0.72}{$n=100,q=0.87$}} & \multicolumn{2}{c}{\scalebox{0.72}{$n=150,q=0.87$}} & \multicolumn{2}{c}{\scalebox{0.72}{$n=200,q=0.87$}} 	& \multicolumn{2}{c}{\scalebox{0.72}{$n=50,q=0.8$}} & \multicolumn{2}{c}{\scalebox{0.72}{$n=100,q=0.8$}} & \multicolumn{2}{c}{\scalebox{0.72}{$n=150,q=0.8$}} & \multicolumn{2}{c}{\scalebox{0.72}{$n=200,q=0.8$}} \\ \hline
MLqE($\boldsymbol{\theta}$)&5.3941&1.0002&5.3295&1.0007&5.3101&1.0006&5.2996&1.0009&5.6228&1.0009&5.5442&1.0007&5.5227&1.0006&5.5089&1.0010\\
$\widehat{Var}(\hat{\boldsymbol{\theta}})$&0.5300&0.0011&0.2401&0.0005&0.1615&0.0003&0.1201&0.0003&0.6812&0.0015&0.3046&0.0006&0.2071&0.0004&0.1522&0.0003\\
$\widehat{MSE}(\hat{\boldsymbol{\theta}})$&0.6854&0.0011&0.3487&0.0005&0.2577&0.0003&0.2098&0.0003&1.0691&0.0015&0.6008&0.0006&0.4803&0.0004&0.4113&0.0003\\ \hline
MLE($\boldsymbol{\theta}$)&1.0943&1.5828&1.0772&1.5837&1.0698&1.5858&1.0685&1.5851&0.9719&2.0993&0.9644&2.1000&0.9620&2.1006&0.9612&2.1001\\
$\widehat{Var}(\hat{\boldsymbol{\theta}})$&0.0156&0.0129&0.0062&0.0064&0.0037&0.0042&0.0027&0.0031&0.0048&0.0249&0.0021&0.0123&0.0014&0.0084&0.0011&0.0063\\
$\widehat{MSE}(\hat{\boldsymbol{\theta}})$&15.2700&0.3525&15.3944&0.3472&15.4501&0.3474&15.4596&0.3455&16.2302&1.2334&16.2880&1.2223&16.3068&1.2197&16.3127&1.2166\\ \hline
			&	\multicolumn{8}{c}{Case 8: $(1-\varepsilon)$W(5,3) +$\varepsilon$W(2,7), $\varepsilon$=0.1} & \multicolumn{8}{c}{Case 8: $(1-\varepsilon)$W(5,3) +$\varepsilon$W(2,7), $\varepsilon$=0.2} \\ \hline
& \multicolumn{2}{c}{\scalebox{0.72}{$n=50,q=0.84$}} & \multicolumn{2}{c}{\scalebox{0.72}{$n=100,q=0.84$}} & \multicolumn{2}{c}{\scalebox{0.72}{$n=150,q=0.84$}} & \multicolumn{2}{c}{\scalebox{0.72}{$n=200,q=0.83$}} 	& \multicolumn{2}{c}{\scalebox{0.72}{$n=50,q=0.77$}} & \multicolumn{2}{c}{\scalebox{0.72}{$n=100,q=0.77$}} & \multicolumn{2}{c}{\scalebox{0.72}{$n=150,q=0.76$}} & \multicolumn{2}{c}{\scalebox{0.72}{$n=200,q=0.76$}} \\ \hline
MLqE($\boldsymbol{\theta}$)&5.0450&3.0315&4.9924&3.0291&4.9679&3.0253&5.0395&3.0237&5.0236&3.0612&4.9861&3.0476&5.0784&3.0395&5.0649&3.0376\\
$\widehat{Var}(\hat{\boldsymbol{\theta}})$&0.7555&0.0112&0.3561&0.0053&0.2309&0.0035&0.1613&0.0026&1.1168&0.0181&0.4973&0.0076&0.3107&0.0046&0.2178&0.0033\\
$\widehat{MSE}(\hat{\boldsymbol{\theta}})$&0.7575&0.0121&0.3561&0.0062&0.2320&0.0042&0.1629&0.0032&1.1174&0.0218&0.4975&0.0099&0.3169&0.0061&0.2221&0.0047\\ \hline
MLE($\boldsymbol{\theta}$)&2.2789&3.4869&2.1464&3.4979&2.1086&3.4970&2.0953&3.4989&1.9144&3.9000&1.8628&3.9065&1.8508&3.9044&1.8441&3.9029\\
$\widehat{Var}(\hat{\boldsymbol{\theta}})$&0.2579&0.0414&0.0829&0.0205&0.0486&0.0141&0.0333&0.0103&0.0700&0.0636&0.0286&0.0316&0.0175&0.0201&0.0129&0.0156\\
$\widehat{MSE}(\hat{\boldsymbol{\theta}})$&7.6625&0.2785&8.2258&0.2684&8.4086&0.2611&8.4705&0.2592&9.5907&0.8736&9.8707&0.8533&9.9347&0.8381&9.9725&0.8308\\ \hline
		\end{tabular}
	}
	\label{caseswtcont}
\end{table}

%% file: table5wp.tex
\begin{table}[htbp]
	\caption{The underlying is Weibull and contamination is Uniform distributions}
	\scalebox{0.72}
	{  
		\begin{tabular}{c|cccccccc|cccccccc}
$\hat{\boldsymbol{\theta}}$	&   $\hat{\alpha}$   &   $\hat{\beta}$         &    $\hat{\alpha}$    &    $\hat{\beta}$   &   $\hat{\alpha}$     &    $\hat{\beta}$  &      $\hat{\alpha}$        &    $\hat{\beta}$ &  $\hat{\alpha}$   &   $\hat{\beta}$         &    $\hat{\alpha}$    &    $\hat{\beta}$   &   $\hat{\alpha}$     &    $\hat{\beta}$  &      $\hat{\alpha}$        &    $\hat{\beta}$       \\ \hline
&	\multicolumn{8}{c}{Case 1: $(1-\varepsilon)$W(8,10) +$\varepsilon$U(4,10), $\varepsilon$=0.1} & \multicolumn{8}{c}{Case 1: $(1-\varepsilon)$W(8,10) +$\varepsilon$U(4,10), $\varepsilon$=0.2} \\ \hline
& \multicolumn{2}{c}{\scalebox{0.72}{$n=50,q=0.84$}} & \multicolumn{2}{c}{\scalebox{0.72}{$n=100,q=0.84$}} & \multicolumn{2}{c}{\scalebox{0.72}{$n=150,q=0.82$}} & \multicolumn{2}{c}{\scalebox{0.72}{$n=200,q=0.82$}} 	& \multicolumn{2}{c}{\scalebox{0.72}{$n=50,q=0.7$}} & \multicolumn{2}{c}{\scalebox{0.72}{$n=100,q=0.7$}} & \multicolumn{2}{c}{\scalebox{0.72}{$n=150,q=0.7$}} & \multicolumn{2}{c}{\scalebox{0.72}{$n=200,q=0.7$}} \\ \hline
MLqE($\boldsymbol{\theta}$)&8.0064&9.8590&7.9204&9.8638&8.0010&9.8689&7.9822&9.8692&8.1291&9.7391&8.0101&9.7455&7.9531&9.7452&7.9363&9.7461\\ 
$\widehat{Var}(\hat{\boldsymbol{\theta}})$&0.9125&0.0385&0.4443&0.0204&0.3111&0.0138&0.2331&0.0100&1.3036&0.0470&0.6139&0.0238&0.4175&0.0161&0.3002&0.0122\\
$\widehat{MSE}(\hat{\boldsymbol{\theta}})$&0.9125&0.0584&0.4507&0.0389&0.3111&0.0310&0.2334&0.0271&1.3202&0.1151&0.6140&0.0886&0.4197&0.0810&0.3043&0.0766 \\ \hline	MLE($\boldsymbol{\theta}$)&7.1776&9.8106&7.1036&9.8159&7.0731&9.8169&7.0499&9.8152&6.4031&9.6178&6.3368&9.6245&6.3101&9.6246&6.3005&9.6251\\
$\widehat{Var}(\hat{\boldsymbol{\theta}})$&0.6376&0.0365&0.3069&0.0191&0.2049&0.0127&0.1524&0.0093&0.4950&0.0386&0.2386&0.0197&0.1605&0.0133&0.1140&0.0098\\
$\widehat{MSE}(\hat{\boldsymbol{\theta}})$&1.3140&0.0724&1.1105&0.0530&1.0639&0.0462&1.0552&0.0435&3.0452&0.1848&3.0048&0.1607&3.0162&0.1542&3.0022&0.1503\\  \hline
&	\multicolumn{8}{c}{Case 2: $(1-\varepsilon)$W(3,5) +$\varepsilon$U(5,15), $\varepsilon$=0.1} & \multicolumn{8}{c}{Case 2: $(1-\varepsilon)$W(3,5) +$\varepsilon$U(5,15), $\varepsilon$=0.2} \\ \hline
& \multicolumn{2}{c}{\scalebox{0.72}{$n=50,q=0.75$}} & \multicolumn{2}{c}{\scalebox{0.72}{$n=100,q=0.75$}} & \multicolumn{2}{c}{\scalebox{0.72}{$n=150,q=0.74$}} & \multicolumn{2}{c}{\scalebox{0.72}{$n=200,q=0.74$}} 	& \multicolumn{2}{c}{\scalebox{0.72}{$n=50,q=0.65$}} & \multicolumn{2}{c}{\scalebox{0.72}{$n=100,q=0.65$}} & \multicolumn{2}{c}{\scalebox{0.72}{$n=150,q=0.65$}} & \multicolumn{2}{c}{\scalebox{0.72}{$n=200,q=0.65$}} \\ \hline
MLqE($\boldsymbol{\theta}$)&3.1173&5.1332&3.0509&5.1345&3.0231&5.1392&3.0161&5.1385&3.2591&5.2930&3.1292&5.2903&3.0989&5.2890&3.0789&5.2872\\
$\widehat{Var}(\hat{\boldsymbol{\theta}})$&0.2868&0.0827&0.1362&0.0412&0.0871&0.0277&0.0665&0.0209&0.5459&0.1214&0.2484&0.0618&0.1641&0.0415&0.1236&0.0311\\
$\widehat{MSE}(\hat{\boldsymbol{\theta}})$&0.3006&0.1004&0.1388&0.0593&0.0876&0.0471&0.0668&0.0400&0.6130&0.2072&0.2650&0.1461&0.1738&0.1250&0.1299&0.1135\\ \hline
MLE($\boldsymbol{\theta}$)&2.2236&5.6668&2.1925&5.6698&2.1816&5.6751&2.1756&5.6775&2.0557&6.3144&2.0339&6.3136&2.0278&6.3161&2.0252&6.3139\\
$\widehat{Var}(\hat{\boldsymbol{\theta}})$&0.0656&0.0843&0.0280&0.0415&0.0173&0.0274&0.0125&0.0207&0.0363&0.0990&0.0165&0.0492&0.0107&0.0336&0.0079&0.0249\\
$\widehat{MSE}(\hat{\boldsymbol{\theta}})$&0.6684&0.5288&0.6800&0.4901&0.6870&0.4832&0.6921&0.4797&0.9280&1.8265&0.9499&1.7746&0.9560&1.7656&0.9580&1.7513\\ \hline
&	\multicolumn{8}{c}{Case 3: $(1-\varepsilon)$W(10,8) +$\varepsilon$U(4,10), $\varepsilon$=0.1} & \multicolumn{8}{c}{Case 3: $(1-\varepsilon)$W(10,8) +$\varepsilon$U(4,10), $\varepsilon$=0.2} \\ \hline
& \multicolumn{2}{c}{\scalebox{0.72}{$n=50,q=0.85$}} & \multicolumn{2}{c}{\scalebox{0.72}{$n=100,q=0.85$}} & \multicolumn{2}{c}{\scalebox{0.72}{$n=150,q=0.85$}} & \multicolumn{2}{c}{\scalebox{0.72}{$n=200,q=0.85$}} 	& \multicolumn{2}{c}{\scalebox{0.72}{$n=50,q=0.74$}} & \multicolumn{2}{c}{\scalebox{0.72}{$n=100,q=0.73$}} & \multicolumn{2}{c}{\scalebox{0.72}{$n=150,q=0.7$}} & \multicolumn{2}{c}{\scalebox{0.72}{$n=200,q=0.7$}} \\ \hline
MLqE($\boldsymbol{\theta}$)&9.9399&7.9939&9.8051&7.9961&9.7881&7.9996&9.7611&7.9982&9.7928&7.9929&9.7810&7.9984&10.0043&8.0030&9.9722&8.0032\\
$\widehat{Var}(\hat{\boldsymbol{\theta}})$&1.3651&0.0185&0.6472&0.0092&0.4316&0.0062&0.3394&0.0049&1.6467&0.0226&0.8393&0.0119&0.6388&0.0080&0.4839&0.0062\\
$\widehat{MSE}(\hat{\boldsymbol{\theta}})$&1.3687&0.0186&0.6852&0.0092&0.4765&0.0062&0.3964&0.0049&1.6896&0.0226&0.8873&0.0119&0.6388&0.0080&0.4847&0.0062\\ \hline
MLE($\boldsymbol{\theta}$)&9.0213&7.9756&8.8908&7.9779&8.8674&7.9813&8.8412&7.9802&8.0909&7.9542&8.0147&7.9587&7.9808&7.9593&7.9593&7.9605\\
$\widehat{Var}(\hat{\boldsymbol{\theta}})$&1.0064&0.0186&0.4656&0.0091&0.3062&0.0062&0.2331&0.0048&0.7783&0.0221&0.3634&0.0114&0.2341&0.0074&0.1771&0.0058\\
$\widehat{MSE}(\hat{\boldsymbol{\theta}})$&1.9643&0.0192&1.6960&0.0096&1.5889&0.0065&1.5759&0.0052&4.4230&0.0242&4.3048&0.0131&4.3113&0.0091&4.3414&0.0074\\ \hline
&	\multicolumn{8}{c}{Case 4: $(1-\varepsilon)$W(3,5) +$\varepsilon$U(3,5), $\varepsilon$=0.1} & \multicolumn{8}{c}{Case 4: $(1-\varepsilon)$W(3,5) +$\varepsilon$U(3,5), $\varepsilon$=0.2} \\ \hline
& \multicolumn{2}{c}{\scalebox{0.72}{$n=50,q=1.07$}} & \multicolumn{2}{c}{\scalebox{0.72}{$n=100,q=1.09$}} & \multicolumn{2}{c}{\scalebox{0.72}{$n=150,q=1.09$}} & \multicolumn{2}{c}{\scalebox{0.72}{$n=200,q=1.09$}} 	& \multicolumn{2}{c}{\scalebox{0.72}{$n=50,q=1.09$}} & \multicolumn{2}{c}{\scalebox{0.72}{$n=100,q=1.1$}} & \multicolumn{2}{c}{\scalebox{0.72}{$n=150,q=1.09$}} & \multicolumn{2}{c}{\scalebox{0.72}{$n=200,q=1.1$}} \\ \hline
MLqE($\boldsymbol{\theta}$)&3.0854&4.9569&3.0185&4.9681&2.9972&4.9665&2.9902&4.9669&3.1509&4.9048&3.0867&4.9119&3.0773&4.9113&3.0555&4.9142\\
$\widehat{Var}(\hat{\boldsymbol{\theta}})$&0.1239&0.0569&0.0556&0.0282&0.0362&0.0191&0.0272&0.0143&0.1287&0.0536&0.0587&0.0261&0.0391&0.0179& 0.0289&0.0136\\
$\widehat{MSE}(\hat{\boldsymbol{\theta}})$&0.1312&0.0588&0.0559&0.0292&0.0362&0.0203& 0.0272&0.0154&0.1514&0.0626&0.0662&0.0338&0.0451&0.0258&0.0319&0.0210\\ \hline
MLE($\boldsymbol{\theta}$)&3.1676&4.9365&3.1250&3.1250&3.1052 &4.9384&3.0986&4.9386&3.2627&4.8741&3.2150&4.8761& 3.1958&4.8784&3.1864&4.8775\\
$\widehat{Var}(\hat{\boldsymbol{\theta}})$&0.1317&0.0567&0.0599&0.0279&0.0392&0.0189& 0.0292&0.0141&0.1372&0.0521&0.0630&0.0251&0.0414&0.0174& 0.0308&0.0131\\
$\widehat{MSE}(\hat{\boldsymbol{\theta}})$&0.1598&0.0607&0.0755&0.0315&0.0503&0.0227& 0.0390&0.0178&0.2062&0.0680&0.1092&0.0405&0.0797&0.0322&0.0655&0.0281\\ \hline
			\end{tabular}
	}
	\label{caseswut}
\end{table}

%% file: table6wp.tex
\begin{table}[htb!]
		\caption{The underlying is Weibull and contamination is BurrIII distributions}
	\scalebox{0.7}
	{  
		\begin{tabular}{c|cccccccc|cccccccc}
			$\hat{\boldsymbol{\theta}}$	&   $\hat{\alpha}$   &   $\hat{\beta}$         &    $\hat{\alpha}$    &    $\hat{\beta}$   &   $\hat{\alpha}$     &    $\hat{\beta}$  &      $\hat{\alpha}$        &    $\hat{\beta}$ &  $\hat{\alpha}$   &   $\hat{\beta}$         &    $\hat{\alpha}$    &    $\hat{\beta}$   &   $\hat{\alpha}$     &    $\hat{\beta}$  &      $\hat{\alpha}$        &    $\hat{\beta}$       \\ \hline
			&	\multicolumn{8}{c}{Case 1: $(1-\varepsilon)$W(1,5) +$\varepsilon$B(2,20), $\varepsilon$=0.1} & \multicolumn{8}{c}{Case 1: $(1-\varepsilon)$W(1,5) +$\varepsilon$B(2,20), $\varepsilon$=0.2} \\ \hline
			& \multicolumn{2}{c}{\scalebox{0.72}{$n=50,q=0.94$}} & \multicolumn{2}{c}{\scalebox{0.72}{$n=100,q=0.94$}} & \multicolumn{2}{c}{\scalebox{0.72}{$n=150,q=0.94$}} & \multicolumn{2}{c}{\scalebox{0.72}{$n=200,q=0.94$}} 	& \multicolumn{2}{c}{\scalebox{0.72}{$n=50,q=0.89$}} & \multicolumn{2}{c}{\scalebox{0.72}{$n=100,q=0.89$}} & \multicolumn{2}{c}{\scalebox{0.72}{$n=150,q=0.89$}} & \multicolumn{2}{c}{\scalebox{0.72}{$n=200,q=0.89$}} \\ \hline
MLqE($\boldsymbol{\theta}$)&1.0576&5.0057&1.0422&5.0046&1.0374&4.9953&1.0347&4.9909&1.1055&5.0499&1.0895&5.0407&1.0855&5.0421&1.0832&5.0372\\
$\widehat{Var}(\hat{\boldsymbol{\theta}})$&0.0179&0.5367&0.0083&0.2677&0.0054&0.1777&0.0041&0.1363&0.0221&0.5177&0.0103&0.2601&0.0069&0.1736&0.0050&0.1310\\
$\widehat{MSE}(\hat{\boldsymbol{\theta}})$&0.0212&0.5367&0.0100&0.2677&0.0068&0.1777&0.0053&0.1363&0.0332&0.5202&0.0183&0.2617&0.0142&0.1753&0.0119&0.1324\\ \hline
MLE($\boldsymbol{\theta}$)&1.0530&5.3062&1.0344&5.3124&1.0274&5.3053&1.0236&5.3019&1.0795&5.6179&1.0560&5.6162&1.0484&5.6207&1.0424&5.6168\\
$\widehat{Var}(\hat{\boldsymbol{\theta}})$&0.0174&0.5776&0.0089&0.2903&0.0061&0.1932&0.0050&0.1495&0.0218&0.6135&0.0116&0.3085&0.0088&0.2059&0.0070&0.1576\\
$\widehat{MSE}(\hat{\boldsymbol{\theta}})$&0.0202&0.6714&0.0100&0.3879&0.0069&0.2864&0.0055&0.2407&0.0281&0.9953&0.0148&0.6882&0.0111&0.5911&0.0088&0.5381\\ \hline
			$\hat{\boldsymbol{\theta}}$	&   $\hat{\alpha}$   &   $\hat{\beta}$         &    $\hat{\alpha}$    &    $\hat{\beta}$   &   $\hat{\alpha}$     &    $\hat{\beta}$  &      $\hat{\alpha}$        &    $\hat{\beta}$ &  $\hat{\alpha}$   &   $\hat{\beta}$         &    $\hat{\alpha}$    &    $\hat{\beta}$   &   $\hat{\alpha}$     &    $\hat{\beta}$  &      $\hat{\alpha}$        &    $\hat{\beta}$       \\ \hline
&	\multicolumn{8}{c}{Case 2: $(1-\varepsilon)$W(3,5) +$\varepsilon$B(2,20), $\varepsilon$=0.1} & \multicolumn{8}{c}{Case 2: $(1-\varepsilon)$W(3,5) +$\varepsilon$B(2,20), $\varepsilon$=0.2} \\ \hline
& \multicolumn{2}{c}{\scalebox{0.72}{$n=50,q=0.89$}} & \multicolumn{2}{c}{\scalebox{0.72}{$n=100,q=0.87$}} & \multicolumn{2}{c}{\scalebox{0.72}{$n=150,q=0.87$}} & \multicolumn{2}{c}{\scalebox{0.72}{$n=200,q=0.87$}} 	& \multicolumn{2}{c}{\scalebox{0.72}{$n=50,q=0.81$}} & \multicolumn{2}{c}{\scalebox{0.72}{$n=100,q=0.81$}} & \multicolumn{2}{c}{\scalebox{0.72}{$n=150,q=0.81$}} & \multicolumn{2}{c}{\scalebox{0.72}{$n=200,q=0.81$}} \\ \hline
MLqE($\boldsymbol{\theta}$)&3.0249&5.0598&3.0102&5.0515&2.9923&5.0422&2.9935&5.0425&3.0853&5.0898&3.0407&5.0752&3.0137&5.0756&3.0157&5.0760\\
$\widehat{Var}(\hat{\boldsymbol{\theta}})$&0.1834&0.0740&0.0874&0.0351&0.0562&0.0227&0.0429&0.0183&0.2157&0.0858&0.0967&0.0417&0.0636&0.0272&0.0491&0.0209\\
$\widehat{MSE}(\hat{\boldsymbol{\theta}})$&0.1840&0.0776&0.0875&0.0378&0.0562&0.0245&0.0430&0.0201&0.2230&0.0939&0.0984&0.0473&0.0638&0.0329&0.0494&0.0266\\ \hline
MLE($\boldsymbol{\theta}$)&2.4092&5.3579&2.1905&5.3762&2.0784&5.3782&2.0086&5.3820&2.0350&5.7129&1.8471&5.7146&1.7609&5.7196&1.7135&5.7233\\
$\widehat{Var}(\hat{\boldsymbol{\theta}})$&0.4280&0.1964&0.2898&0.1031&0.2456&0.0725&0.2043&0.0569&0.3500&0.3050&0.2018&0.1572&0.1518&0.1118&0.1191&0.0862\\
$\widehat{MSE}(\hat{\boldsymbol{\theta}})$&0.7771&0.3245&0.9450&0.2446&1.0950&0.2155&1.1872&0.2028&1.2812&0.8132&1.5309&0.6679&1.6872&0.6296&1.7742&0.6094\\ \hline
&	\multicolumn{8}{c}{Case 3: $(1-\varepsilon)$W(5,1) +$\varepsilon$B(2,20), $\varepsilon$=0.1} & \multicolumn{8}{c}{Case 3: $(1-\varepsilon)$W(5,1) +$\varepsilon$B(2,20), $\varepsilon$=0.2} \\ \hline
& \multicolumn{2}{c}{\scalebox{0.72}{$n=50,q=0.86$}} & \multicolumn{2}{c}{\scalebox{0.72}{$n=100,q=0.88$}} & \multicolumn{2}{c}{\scalebox{0.72}{$n=150,q=0.88$}} & \multicolumn{2}{c}{\scalebox{0.72}{$n=200,q=0.88$}} 	& \multicolumn{2}{c}{\scalebox{0.72}{$n=50,q=0.8$}} & \multicolumn{2}{c}{\scalebox{0.72}{$n=100,q=0.8$}} & \multicolumn{2}{c}{\scalebox{0.72}{$n=150,q=0.8$}} & \multicolumn{2}{c}{\scalebox{0.72}{$n=200,q=0.8$}} \\ \hline
MLqE($\boldsymbol{\theta}$)&5.4852&0.9989&5.3394&1.0000&5.3218&0.9995&5.3078&0.9992&5.6303&1.0037&5.5981&0.9994&5.5816&0.9987&5.5781&0.9993\\
$\widehat{Var}(\hat{\boldsymbol{\theta}})$&0.5120&0.0010&0.2345&0.0005&0.1585&0.0004&0.1169&0.0003&0.9365&0.0028&0.3323&0.0009&0.1932&0.0004&0.1482&0.0003\\
$\widehat{MSE}(\hat{\boldsymbol{\theta}})$&0.7474&0.0010&0.3497&0.0005&0.2621&0.0004&0.2117&0.0003&1.3338&0.0028&0.6900&0.0009&0.5315&0.0004&0.4824&0.0003\\ \hline
MLE($\boldsymbol{\theta}$)&1.1026&1.5800&1.0635&1.5789&1.0379&1.5827&1.0278&1.5812&0.9505&2.1014&0.9278&2.0925&0.9142&2.0963&0.9075&2.0951\\
$\widehat{Var}(\hat{\boldsymbol{\theta}})$&0.0410&0.0223&0.0243&0.0121&0.0174&0.0082&0.0134&0.0060&0.0208&0.0498&0.0110&0.0220&0.0089&0.0160&0.0071&0.0114\\
$\widehat{MSE}(\hat{\boldsymbol{\theta}})$&15.2304&0.3587&15.5201&0.3473&15.7159&0.3478&15.7914&0.3437&16.4190&1.2628&16.5937&1.2155&16.7024&1.2177&16.7557&1.2105\\ \hline
&	\multicolumn{8}{c}{Case 4: $(1-\varepsilon)$W(10,8) +$\varepsilon$B(2,20), $\varepsilon$=0.1} & \multicolumn{8}{c}{Case 4: $(1-\varepsilon)$W(10,8) +$\varepsilon$B(2,20), $\varepsilon$=0.2} \\ \hline
& \multicolumn{2}{c}{\scalebox{0.72}{$n=50,q=0.82$}} & \multicolumn{2}{c}{\scalebox{0.72}{$n=100,q=0.82$}} & \multicolumn{2}{c}{\scalebox{0.72}{$n=150,q=0.82$}} & \multicolumn{2}{c}{\scalebox{0.72}{$n=200,q=0.8$}} 	& \multicolumn{2}{c}{\scalebox{0.72}{$n=50,q=0.7$}} & \multicolumn{2}{c}{\scalebox{0.72}{$n=100,q=0.7$}} & \multicolumn{2}{c}{\scalebox{0.72}{$n=150,q=0.69$}} & \multicolumn{2}{c}{\scalebox{0.72}{$n=200,q=0.69$}} \\ \hline
MLqE($\boldsymbol{\theta}$)&9.8204&7.9524&9.6892&7.9526&9.6657&7.9558&9.8622&7.9581&10.0918&7.9249&9.8777&7.9264&9.9913&7.9367&9.9574&7.9362\\
$\widehat{Var}(\hat{\boldsymbol{\theta}})$&1.7467&0.0190&0.8708&0.0094&0.5688&0.0064&0.4775&0.0048&2.8331&0.0239&1.3950&0.0122&0.9690&0.0082&0.7346&0.0067\\
$\widehat{MSE}(\hat{\boldsymbol{\theta}})$&1.7789&0.0212&0.9673&0.0116&0.6806&0.0084&0.4965&0.0066&2.8415&0.0295&1.4100&0.0176&0.9691&0.0123&0.7364&0.0108\\ \hline
MLE($\boldsymbol{\theta}$)&5.5193&8.2448&4.4450&8.2947&3.8883&8.3273&3.5565&8.3300&3.5681&8.4284&2.9810&8.4351&2.6693&8.4701&2.5040&8.4876\\
$\widehat{Var}(\hat{\boldsymbol{\theta}})$&5.4952&0.3408&4.0032&0.2034&3.0931&0.1523&2.2446&0.1183&2.5827&0.4989&1.4444&0.2649&0.9102&0.1950&0.6641&0.1617\\
$\widehat{MSE}(\hat{\boldsymbol{\theta}})$&25.5720&0.4007&34.8611&0.2903&40.4455&0.2594&43.7628&0.2272&43.9517&0.6824&50.7104&0.4542&54.6498&0.4160&56.8545&0.3994\\ \hline
			\end{tabular}
	}
	\label{caseswbtg}
\end{table}

%% file: Weiproparxiv2adopted.bbl
\section*{References}
\begin{thebibliography}{99}
	\bibitem{tikuinlier} Tiku, M. L. (1975). "A new statistic for testing suspected outliers," Communications in Statistics-Theory and Methods, 4(8), 737-752.

\bibitem{LehmannCas98}  E.L. Lehmann, and G. Casella, {\it Theory of point estimation},  Wadsworth \& Brooks/Cole. Pacific Grove, CA, 589, USA, 1998.

\bibitem{Hampeletal86} F.R. Hampel, E. M. Ronchetti, P. J. Rousseeuw, and W. A. Stahel, {\it Robust statistics: The approach based on influence functions}, Wiley Series in Probability and Statistics, New York, 1986.

\bibitem{Wadatwopara}  Wada, T., and Suyari, H. (2007). "A two-parameter generalization of Shannon–Khinchin axioms and the uniqueness theorem," Physics Letters A, 368(3-4), 199-205.

\bibitem{Bercher10} Bercher, J. F. (2010). "On escort distributions, q-gaussians and Fisher information," 30th International Workshop on Bayesian Inference and Maximum Entropy Methods in Science and Engineering, Jul 2010, Chamonix, France. pp.208-215, ff10.1063/1.3573618ff.

\bibitem{Bercher12a} Bercher, J. F. (2012). "A simple probabilistic construction yielding generalized entropies and divergences, escort distributions and q-Gaussians," Physica A: Statistical Mechanics and its Applications, 391(19), 4460-4469.
\bibitem{God60}  Godambe, V. P. (1960). "An optimum property of regular maximum likelihood estimation," The Annals of Mathematical Statistics, 31(4), 1208-1211.
\bibitem{GodTh78} Godambe, V.P., and Thompson, M.E. 1978. "Some aspects of the theory of estimating equations,"  Journal of Statistical Planning and Inference. Vol. 2(1), 95-104.
\bibitem{Tsallisbook09} C. Tsallis, {\it Introduction to Nonextensive Statistical Mechanics: Approaching a Complex World}, Springer, New York, 2009.%**** 
\bibitem{FerrariYang10}  Ferrari, D.,  Yang, Y. (2010). "Maximum Lq-likelihood estimation," The Annals of Statistics, 38(2), 753-783.
\bibitem{CanKor18}  \c{C}ankaya, M. N.,  Korbel, J. (2018). "Least informative distributions in maximum q-log-likelihood estimation," Physica A: Statistical Mechanics and its Applications, 509, 140-150.
\bibitem{Lindsay94} Lindsay, B. G. (1994). "Efficiency versus robustness: the case for minimum Hellinger distance and related methods," The annals of statistics, 22(2), 1081-1114.
\bibitem{Jansent}  Hanel, R., and Thurner, S. (2011). "A comprehensive classification of complex statistical systems and an axiomatic derivation of their entropy and distribution functions," EPL (Europhysics Letters), 93(2), 20006.
\bibitem{alphabetadivergences} Cichocki, A.,  Amari, S. I. (2010). "Families of alpha-beta-and gamma-divergences: Flexible and robust measures of similarities," Entropy, 12(6), 1532-1568.
\bibitem{MLqEGamma} Xing, N. (2015). Maximum Lq-Likelihood Estimation for Gamma Distributions (Master's thesis, Graduate Studies).%logqgamma
\bibitem{PardoSD} L. Pardo, {\it Statistical inference based on divergence measures,} CRC Press, Taylor \& Francis Group, 2005.
\bibitem{introga} Mitchell, M. (1998). {\it An introduction to genetic algorithms}, MIT press.
\bibitem{Hub81} Huber, P.J. 1981. {\it Robust Statistics,} Wiley Series in Probability and Statistics, 308, New York.
\bibitem{varincomposite} Varin, C., Reid, N.,  and Firth, D. (2011). "An overview of composite likelihood methods,". Statistica Sinica, 5-42.
\bibitem{Cankaya2018} \c{C}ankaya M.N. "Asymmetric bimodal exponential power distribution on the real line," Entropy, {2018}, 20(1), 23.
\bibitem{Malikcondition} Malik S. C., and Arora S., 1992. {\it  Mathematical analysis,} New Age International.
\bibitem{Weibullref}	Weibull, W. (1951). "Wide applicability," Journal of applied mechanics, 103(730), 293-297.
\bibitem{reviewwei} Almalki, S. J.,  Nadarajah, S. (2014). "Modifications of the Weibull distribution: A review," Reliability Engineering \& System Safety, 124, 32-55.
\bibitem{Tsallis1988} Tsallis C. "Possible generalization of Boltzmann-Gibbs statistics," Journal of Statistical Physics, { 1988}, 52, 479-487.
\bibitem{Cramer46} Cramér, H. (1946). "A contribution to the theory of statistical estimation," Scandinavian Actuarial Journal, 1946(1), 85-94.
\bibitem{Vajda86} Vajda, I. (1986). "Efficiency and robustness control via distorted maximum likelihood estimation," Kybernetika, 22(1), 47-67.
\bibitem{Jan2} Korbel, J., Hanel, R., and Thurner, S. (2020). "Information geometry of scaling expansions of non-exponentially growing configuration spaces," The European Physical Journal Special Topics, 229(5), 787-807.
\bibitem{Basuetal98}  Basu, A., Harris, I. R., Hjort, N. L., and Jones, M. C. (1998). "Robust and efficient estimation by minimising a density power divergence," Biometrika, 85(3), 549-559.
\bibitem{Udristegeo} Calin, O., and Udrişte, C. (2014). {\it Geometric modeling in probability and statistics}. Berlin: Springer.
\bibitem{Muhammedbadwei} Al Mohamad, D. (2018). "Towards a better understanding of the dual representation of phi divergences," Statistical Papers, 59(3), 1205-1253. 
\bibitem{Amari16}  Amari, S. I. (2016). {\it  Information geometry and its applications,} (Vol. 194). Springer.
\bibitem{BrondivergencesBas}  Broniatowski, M., and Vajda, I. (2009). "Several applications of divergence criteria in continuous families," arXiv preprint arXiv:0911.0937.
\bibitem{Canetal19} \c{C}ankaya, M. N., Yal\c{c}{\i}nkaya, A., Altında\v{g}, \"{O}., and Arslan, O. (2019). "On the robustness of an epsilon skew extension for Burr III distribution on the real line,". Computational Statistics, 34(3), 1247-1273.

\bibitem{Canorder20} \c{C}ankaya, M. N. (2020). "M-Estimations of Shape and Scale Parameters by Order Statistics in Least Informative Distributions on q-deformed logarithm," I\v{g}dır Üniversitesi Fen Bilimleri Enstitüsü Dergisi, 10(3), 1984-1996.
\bibitem{glass15data}  Smith, R.L., and Naylor, J.C. (1987). "A comparison of maximum likelihood and Bayesian estimators for the
three-parameter Weibull distribution," Appl. Stat. 36:358–369
\bibitem{alphapowerglass15} Nassar, M., Alzaatreh, A., Mead, M., and Abo-Kasem, O. (2017). "Alpha power Weibull distribution: Properties and applications," Communications in Statistics-Theory and Methods, 46(20), 10236-10252.

\bibitem{Haberman} Haberman, S.J. 1989. "Concavity and Estimation," The Annals of Statistics. JSTOR, Vol.17(4), 1631-1661








\bibitem{mixwei}	Razali, A. M.,  and Al-Wakeel, A. A. (2013). "Mixture Weibull distributions for fitting failure times data," Applied Mathematics and Computation, 219(24), 11358-11364.
\bibitem{hybridga}	Yuan, Q.,  Yang, Z. (2013). "On the performance of a hybrid genetic algorithm in dynamic environments," Applied Mathematics and Computation, 219(24), 11408-11413.

\bibitem{globalopt}	Price, K., Storn, R. M.,  Lampinen, J. A. (2006). {\it  Differential evolution: a practical approach to global optimization,} Springer Science \& Business Media.


\bibitem{Shao03} Shao, J. 2003. {\it  Mathematical Statistics,} Second edition, Springer, 591, USA. 
\bibitem{JizKorhybrid}  Jizba, P.,  and Korbel, J. (2016). "On q-non-extensive statistics with non-Tsallisian entropy," Physica A: Statistical Mechanics and its Applications, 444, 808-827.

\bibitem{weibook} Murthy, D. P., Xie, M., and Jiang, R. 2004. {\it Weibull models} (Vol. 505). John Wiley  Sons.
\bibitem{Fer09} Ferrari, D., and Paterlini, S., "The maximum lq-likelihood method: an application to extreme
quantile estimation in Finance," Methodology and Computing in Applied Probability (2009),
11(1), 3-19. 

\end{thebibliography}
